\documentclass[12pt]{article}
\usepackage[pdftex]{graphicx}

\usepackage{lscape}
\usepackage{setspace}
\usepackage{amssymb,amsmath}
\usepackage{amsfonts}
\usepackage{amsthm}
\usepackage{verbatim}
\usepackage{fullpage}
\usepackage{bbm}
\usepackage[longnamesfirst]{natbib}
\usepackage{multirow}
\usepackage{geometry}
\usepackage{hyperref}
\usepackage{color}
\hypersetup{citebordercolor= 1 1 1, linkbordercolor = 0 1 0}
\usepackage{booktabs}
\usepackage{adjustbox}
\bibpunct{(}{)}{,}{a}{}{,}
\usepackage{bm}
\usepackage{textcomp}
\usepackage{makecell}
\usepackage[bottom]{footmisc}

\usepackage[labelfont=bf]{caption}
\usepackage{rotating}

\usepackage{floatrow}
\usepackage{tabularx}

\newtheorem{theorem}{Theorem}
\newtheorem{assumption}{Assumption}

\newtheorem{lemma}{Lemma}

\newtheorem{remark}{Remark}

\newcommand{\mb}[1]{\mathbb{#1}}

\newcommand{\mf}[1]{\mathbf{#1}}
\newcommand{\mr}[1]{\mathrm{#1}}
\newcommand{\mc}[1]{\mathcal{#1}}

\newcommand{\bs}[1]{\boldsymbol{#1}}

\newcommand{\E}{\mb E}

\newcommand{\beq}{\begin{equation}}
\newcommand{\eeq}{\end{equation}}
\newcommand{\bea}{\begin{eqnarray}}
\newcommand{\eea}{\end{eqnarray}}
\newcommand{\bi}{\begin{itemize}}
\newcommand{\ei}{\end{itemize}}

\hypersetup{
     colorlinks   = true,
     citecolor    = DarkBlue,
     urlcolor     = Black,
     linkcolor    = DarkBlue
}

\usepackage{newfloat}
\usepackage{placeins}
\usepackage{subcaption}

\usepackage[usenames,dvipsnames,svgnames,table]{xcolor}
\definecolor{mygray}{gray}{0.95}

\usepackage[ruled,vlined]{algorithm2e}
\usepackage{accents}
\usepackage{floatrow}
\usepackage{threeparttable}

\usepackage{enumitem}

\DeclareFloatingEnvironment[
  fileext = los ,
  listname = {List of models} ,
  name = Model
]{model}

\usepackage[frozencache, newfloat, cachedir=./minted]{minted} 
\usepackage{makecell}
\usepackage{booktabs}

\begin{document}

\title{
  Inference for Regression with Variables Generated by AI or Machine Learning\thanks{Authors are in alphabetical order. This paper supersedes our February 2024 working paper \url{https://arxiv.org/pdf/2402.15585v1}, which was circulated under a different title. SH acknowledges funding from ERC Consolidator Grant 864863, which supported his and LB's time.  We thank Nick Bloom, Germain Gauthier, Evan Munro, Ashesh Rambachan, David Rossell, and Leif Thorsrud for feedback, and seminar participants at Aarhus, Bocconi, BSE, Bates, Brown, Columbia, ETH Zurich, LSE, Kent, Reserve Bank of Australia, UCSD, UPenn, USC, Warwick, the 3rd Monash-Warwick-Zurich Text-as-Data Workshop, 2024 BSE Summer Institute, 2024 FinEML Conference, 2024 UChicago Machine Learning in Economics Summer Conference, 2024 ISNPS Conference, 2024 ECONDAT Fall Meeting, and the 2024 NASM, ESIF-AIML, and ESAM Conferences of the Econometric Society. Konrad Kurczynski provided excellent research assistance.}
}
\author{
  \begin{tabular}{c@{\hskip 1in}c}
    Laura Battaglia\thanks{Department of Statistics, University of Oxford. \texttt{battaglia@stats.ox.ac.uk}}
     & Timothy Christensen\thanks{Department of Economics, Yale University. \texttt{timothy.christensen@yale.edu}} \\[1ex]
    Stephen Hansen\thanks{Department of Economics, University College London, IFS, and CEPR. \texttt{stephen.hansen@ucl.ac.uk}}
     & Szymon Sacher\thanks{Graduate School of Business, Stanford University. \texttt{sacher@stanford.edu}}
  \end{tabular}
}

\date{April 29, 2025}
\maketitle

\begin{abstract}
  \noindent Researchers now routinely use AI or other machine learning methods to estimate latent variables of economic interest, then plug-in the estimates as covariates in a regression.  We show both theoretically and empirically that naively treating AI/ML-generated variables as ``data'' leads to biased estimates and invalid inference. To restore valid inference, we propose two methods: (1) an explicit bias correction with bias-corrected confidence intervals, and (2) joint estimation of the regression parameters and latent variables.  We illustrate these ideas through applications involving label imputation, dimensionality reduction, and index construction via classification and aggregation.
\end{abstract}

\bigskip


\bigskip


\thispagestyle{empty}
\clearpage{}
\pagenumbering{arabic}

\newpage
\onehalfspacing


\section{Introduction} \label{sec:intro}

Economists now routinely use artificial intelligence (AI) or machine learning (ML) algorithms to generate new variables.  These technologies are used to quantify unstructured data such as text and images, to measure subtle concepts like uncertainty and sentiment, and to create new data sets of variables that were previously too costly, labor-intensive, or otherwise infeasible to collect.

The variables generated by AI and ML algorithms are rarely of interest in themselves, but rather are used in econometric models to address questions of cause and effect, produce forecasts, or estimate counterfactuals. In pioneering work, \citet{bakerMeasuringEconomicPolicy2016} quantifies economic policy uncertainty from news text and uses it as a covariate in regressions and vector autoregressions (VARs). In more recent examples, \citet{magnolfiTripletEmbeddingsDemand2025} measures product differentiation from survey data and \citet{compianiDemandEstimationText2025} measures product substitutability with text and images from online platforms. Both papers then use the derived measures in demand models. \citet{gorodnichenkoVoiceMonetaryPolicy2023} measures tone-of-voice from audio recordings of central bank press conferences, then runs predictive regressions of financial variables on tone. \citet{gabaixAssetEmbeddings2023} imputes firm characteristics from investor holdings data and uses them to explain asset returns. \cite{einavProducingHealthMeasuring2024} measures patient health status from surveys then uses it in an econometric model of nursing home value added. \cite{vafaDecomposingChangesGender2023} measures labor market experience from CVs and uses it to study the gender wage gap. 

In standard practice, AI- or ML-generated variables are treated as regular numerical data when estimating and performing inference in downstream econometric models. We refer to this as the \emph{two-step strategy}: variables are generated in the first step, then used as covariates in the second step. While a pragmatic initial approach, the two-step strategy has largely unknown statistical properties.  One natural concern is that estimators are biased due to measurement error in the AI/ML-generated variables.  Another is that inference suffers from a generated regressor problem \citep{paganEconometricIssuesAnalysis1984}.  Conversely, results in the time-series literature suggest that plugging in estimated variables need not lead to inference problems \citep{stockForecastingUsingPrincipal2002,bernankeMeasuringEffectsMonetary2005,baiConfidenceIntervalsDiffusion2006}.  Without a coherent framework for analyzing the problem, it is difficult to assess which of these perspectives is correct.  More generally, characterizing the statistical guarantees\textemdash or lack thereof\textemdash of the two-step strategy is an important step in developing reliable inference methods for working with variables generated by AI or ML, an area that is still very much in its infancy.

This paper makes two main contributions. First, we show formally that the two-step strategy can lead to invalid inference on downstream regression parameters, even in modern settings where high-performance algorithms are deployed on large data sets.  Second, more constructively, we propose two methods for valid inference: (i) bias-corrected confidence intervals, and (ii) joint estimation of the regression coefficients and latent variables. We document the performance of the methods in several empirical settings.

We consider a downstream regression where the outcome variable $Y_i$ depends on vectors of latent variables $\boldsymbol\theta_i$ and observed variables $\mathbf{q}_i$.  For each observation, the researcher also has an unstructured or high-dimensional data set $\mathbf{x}_i$ for estimating $\boldsymbol\theta_i$. To accommodate many scenarios, we stay agnostic on the form that $\mathbf{x}_i$ takes. For instance, it may be a sequence of words with textual data, an array of RGB values with image data, or a sequence of amplitudes with audio data.  One scenario is \textit{label imputation}, where $\boldsymbol\theta_i$ is a vector of binary labels (e.g., race or gender indicators).  In this case, $\mathbf{x}_i$ (e.g., images) is used as an input to a classifier which produces predicted values $\hat{\bs \theta}_i$.  Another scenario is \textit{dimensionality reduction}, in which a low-dimensional representation (or embedding) $\hat{\bs \theta}_i$ of $\mathbf{x}_i$ is generated with an unsupervised learning model, as in \cite{hansenTransparencyDeliberationFOMC2018}, \cite{bybeeBusinessNewsBusiness2024}, and \cite{ashMoreLawsMore2025} to name a few.  A third scenario is \textit{index construction}.  For example, $\mathbf{x}_i$ could be a set of texts (e.g., articles, paragraphs, or sentences) which are individually classified as containing positive or negative sentiment, then aggregated and normalized to produce a sentiment score $\hat{\bs \theta}_i$, as in \cite{bakerMeasuringEconomicPolicy2016}, \cite{caldaraMeasuringGeopoliticalRisk2022}, and \cite{gorodnichenkoVoiceMonetaryPolicy2023}.

In the two-step strategy, the researcher first computes an estimate $\hat{\boldsymbol\theta}_i$ of $\boldsymbol{\theta}_i$ from $\mf x_i$ for each observation, then regresses $Y_i$ on $\hat{\boldsymbol\theta}_i$ and $\mathbf{q}_i$, and reports point estimates and confidence intervals using standard OLS methods (i.e., treating $\hat{\boldsymbol\theta}_i$ as regular numeric data).  Depending on the context, one may wish to do inference on the coefficients of the latent or observed variables.  In either case, the key question is whether this approach leads to valid inference.

To this end, we introduce an asymptotic framework in which the magnitude of the measurement error and sampling uncertainty remain comparable as the sample size increases. This framework delivers tractable approximations to the finite-sample distribution faced in practice, in which both sources of error play a role.\footnote{Our use of sequences of DGPs to better approximate the finite-sample behavior of estimators has a precedent in a number of contexts in economics. See, e.g., \cite{phillipsUnifiedAsymptoticTheory1987}, \cite{chesherEffectMeasurementError1991}, \cite{staigerInstrumentalVariablesRegression1997}, and \cite{hahnAsymptoticallyUnbiasedInference2002}.}  It also captures the prevailing trend of analyzing increasingly large data sets with increasingly accurate algorithms.  

In this framework, we derive two new results about the two-step strategy.  First, the asymptotic distribution of OLS estimators has a first-order bias due to measurement error.  The bias is increasing in the scale of measurement error relative to sampling uncertainty in the downstream model.  Second, the asymptotic variance of the OLS estimator is the same as if $Y_i$ were regressed on the true $\bs{\theta}_i$ and $\mathbf{q}_i$.  Moreover, OLS standard errors are consistent.  As a result, two-step confidence intervals 
have the correct width but incorrect centering, making them invalid for inference. This differs from a generated regressor problem, where the variance is inflated but there is no location shift. To the extent that the empirical economics literature acknowledges the two-step strategy might be a problem, concerns typically focus on standard errors. Our analysis shows these concerns are misplaced: the primary issue is bias, not incorrect standard errors.

For the case of imputed labels, the potential for AI/ML-generated variables to bias downstream estimators has been flagged in recent work, mainly in data science and political science. See \cite{fongMachineLearningPredictions2021},  \cite{allonMachineLearningPrediction2023}, \cite{angelopoulosPredictionpoweredInference2023,angelopoulosPPIEfficientPredictionPowered2023}, \cite{zhangDebiasingMachineLearningAIGenerated2023}, \cite{zrnicCrosspredictionpoweredInference2024} \cite{miaoTaskAgnosticMachineLearningAssistedInference2024}, \cite{klugerPredictionPoweredInferenceImputed2025} and \cite{sanfordAdversarialDebiasingUnbiased2025} for general ML-generated variables, and \cite{egamiUsingImperfectSurrogates2023,egamiUsingLargeLanguage2024} and \cite{ludwigLargeLanguageModels2025} for variables generated by large language models.\footnote{There is also recent work in economics that considers the complementary problem of imputed dependent variables as opposed to imputed covariates; see \cite{rambachanProgramEvaluationRemotely2025} and \cite{modarressiCausalInferenceOutcomes2025}.} These works demonstrate the inconsistency of OLS estimators in settings 
where the magnitude of measurement error remains fixed as the sample size increases. However, this asymptotic framework isn't necessarily appropriate in modern use cases, where high-quality algorithms are deployed on large data sets. Our analysis provides a new set of results for such cases.

Furthermore, these works propose bias corrections that require a validation sample in which both the true $\bs \theta_i$ and its AI/ML-generated estimate $\hat{\bs \theta}_i$ are observed alongside $(Y_i,\mathbf q_i)$.\footnote{This approach is related to an older literature on estimation with auxiliary data \citep{chenSemiparametricEfficiencyGMM2008}.} The idea is to use the validation data to estimate bias, then  bias-correct estimates from the main sample in which only $\hat{\bs \theta}_i$ is available. Such an approach is possible when the researcher can, albeit at some cost, scrutinize $\mf x_i$ and assign a ground-truth $\bs \theta_i$. But in that case one could simply estimate the model on the validation data alone: the AI/ML-generated data is only useful insofar as it may help improve efficiency. More problematically,  $\bs \theta_i$ is latent in most economic use cases---for instance, one never observes true policy uncertainty, risk, or sentiment---so validation data is unavailable and these existing methods are inapplicable.

Our first inference approach is based on bias correction, but unlike existing approaches it does not require validation data. Instead, we rely on our theoretical results, which characterize the first-order asymptotic bias of OLS estimators and establish consistency of two-step standard errors. This allows us to perform an analytical bias correction, then 
re-center the usual confidence intervals at the bias-corrected estimator to perform valid inference. Our bias corrections are general and widely applicable. We specialize them to AI/ML-generated binary labels and dimension reduction. For the former, bias-correction can be performed without validation data provided one has a measure of the classifier's expected false-positive rate.\footnote{The expected false-positive rate may be estimated from a validation sample, but it may also be available externally. To give a recent example, \cite{bursztynImmigrantNextDoor2024} uses a ML algorithm to classify charitable donors' names by ethnicity. They estimate the accuracy of the classifier using an external sample of North Carolina voter registration data which contains self-reported ethnicity (but not data on donations or other controls).} For the latter, bias correction can be performed using the estimated low-dimensional representation.

It is important to note that the measurement error in AI/ML generated variables may be ``nonclassical'' (i.e., correlated with the true latent $\bs \theta_i$). This makes it difficult for researchers to know even the sign of the bias ex ante: there may be attenuation or amplification. Indeed, Section~\ref{sec:simulations.lessons} shows that bias can be positive or negative in the case of index construction. Nonclassical measurement error is also much more difficult to correct for than classical measurement error, requiring specialized methods\textemdash see, e.g., \cite{schennachMeasurementSystems2022} for a discussion. 

Our bias corrections are convenient to apply, but they may not be available for all types of AI or ML algorithms. They also rely on the magnitudes of measurement error and sampling uncertainty being comparable.  To perform inference without validation data in settings where this is not the case, we introduce a second approach based on joint maximum likelihood estimation of the 
models for latent variable estimation and regression.
In this approach, the model linking $\mf x_i$ with the latent $\bs \theta_i$ is analogous to an ``observation equation'' in state-space models. This requires some more careful modeling of how $\bs \theta_i$ and $\mf x_i$ are related, but we demonstrate its feasibility with three distinct and non-exhaustive applications: AI/ML generated binary labels, dimension reduction, and AI/ML generated indices. 

While joint estimation is straightforward in theory, it presents a computational challenge due to the large number of latent $\bs \theta_i$ that must be integrated out of the likelihood.  To address this, we use Hamiltonian Monte Carlo, a Markov Chain Monte Carlo algorithm that uses information on the gradient of a distribution to sample from it. Implementation is greatly simplified with the use of modern probabilistic programming languages: one simply specifies the likelihood in code, which is then ``automatically'' compiled to perform sampling.\footnote{Previous papers that have performed inference using the joint likelihood approach with unstructured data include \citet{gentzkowMeasuringGroupDifferences2019}, \citet{ruizSHOPPERProbabilisticModel2020}, and \citet{munroLatentDirichletAnalysis2022}.  These typically require custom code to estimate, which makes adapting the model difficult for non-specialists.}

We introduce three applications to illustrate the theoretical results.  The first illustrates label imputation.  \citet{hansenRemoteWorkJobs2023} uses a Large Language Model to classify
each job posting in the Lightcast dataset as offering remote work or not. 
 These imputed labels can be merged with other posting-level metadata to study 
 the causes and consequence of remote work adoption.  We focus on the relationship between wage inequality and remote work by regressing the posted wage on the remote indicator and controls.  The classifier achieves a high test-set accuracy of 99\%, so one might expect that measurement error is inconsequential. 
  But our theory shows that the important quantity is measurement error \textit{relative} to sampling uncertainty, and the Lightcast dataset has hundreds of millions of individual observations.  We show via a case study that bias correction and joint estimation both estimate notably stronger effects of remote work on posted wages than the two-step strategy.

The second application illustrates regressors derived from dimension-reduction algorithms.  \citet{bandieraCEOBehaviorFirm2020} conducts a time-use survey to document behavioral differences among CEOs and their impact on firm performance. The authors use latent Dirichlet allocation \citep{bleiLatentDirichletAllocation2003}\textemdash a factor model for discrete data\textemdash to represent CEO time-use behavior in a low-dimensional space.  This representation is then included as a covariate along with other firm controls in a sales regression.  We replicate this two-step strategy and find the estimated impact of behavior on performance aligns with estimates obtained via bias correction and the joint estimation strategy.  Our theory predicts this will hold when measurement error is low compared to sampling uncertainty. 
We then re-estimate the model using a 10\% subsample of time units for each CEO to scale-up measurement error. Here we find the two-step strategy produces insignificant behavioral effects, while both corrections produce significant effects.

The third application illustrates index construction  via classification and aggregation.  Central bank communication has become a major research and policy topic over the past decade, and linking market reactions to communication often involves quantifying the latter from unstructured data.  We replicate the hawkish sentiment measure from \citet{gorodnichenkoVoiceMonetaryPolicy2023} which classifies individual paragraphs of FOMC statements as hawkish or dovish.  Paragraphs are then aggregated to form a meeting-level share which proxies continuous, latent sentiment.  Following the two-step strategy, we regress the path factor \citep{gurkaynakActionsSpeakLouder2005}\textemdash a measure of movement in the long end of the yield curve\textemdash on sentiment and find weakly positive effects.  With joint estimation, however, the estimated effect size and $R^2$ are both nearly three times larger,  which shows the value of our correction for prediction as well as for inference.

Finally, in simulation exercises calibrated to the empirical applications, we find that the two-step strategy performs poorly\textemdash in terms of bias in estimated coefficients and coverage of confidence intervals\textemdash relative to both bias correction and joint estimation.  This provides further evidence that 
 measurement error distorts inference in the two-step strategy, while our corrections reduce bias and restore valid inference even in challenging empirical settings.

Our overall message is that the increasingly common practice of using regressors generated by AI or ML can lead to invalid inference, but practical solutions exist. We view our proposed bias correction and joint estimation approaches as robust, widely applicable starting points for empirical analysis.  For instance, an emerging line of research uses text-derived sentiment indices as inputs into forecasting models.  
Our analysis can be extended to show how errors in these indices lead to biased forecasts, while our solutions can improve the performance of these forecasting methods.  Likewise, the industrial organization literature increasingly uses embedded representations of firms and products to model market behavior and demand.  Our solutions can be adapted to these settings as well. Going forward, it is important to establish which algorithms and econometric models are most susceptible to measurement error and associated inference problems.  More generally, inference problems arising from the use of AI/ML-generated variables should more widely recognized in order to fully harness the potential of AI/ML methods in empirical economics.

The rest of the paper proceeds as follows.  Section~\ref{sec:warmup} provides a simple setting illustrating why the two-step strategy leads to biased inference and how our proposed solutions can help.   Section~\ref{sec:applications} introduces the more general framework and presents three empirical applications. 
 Sections~\ref{sec:model} and~\ref{sec:integrated} present, respectively, the main theoretical analysis of the two-step strategy and the proposed solutions.  Section~\ref{sec:simulations} presents simulation results and Section~\ref{sec:conclusion} concludes.


\section{A Simple Example} \label{sec:warmup}

This section presents a simple model to illustrate how the two-step strategy leads to biased inference, and how our proposed methods can restore valid inference.

\subsection{Model}

The model is loosely based on \cite{bakerMeasuringEconomicPolicy2016}. Suppose we are interested in the effect $\gamma$ of $\theta_i$ (policy uncertainty in month $i$) on $Y_i$ (employment or investment, say, in month $i+1$) in the regression model
\begin{equation}
    Y_i = \alpha + \gamma \theta_i + \varepsilon_i . \label{eq:toy.obs}
\end{equation}
Policy uncertainty is a nebulous concept that is difficult to precisely define let alone observe. \cite{bakerMeasuringEconomicPolicy2016} forms EPU indices from monthly counts of articles in ten newspapers containing certain terms, which they convert to an index. 
Evidently there is measurement error due to the sampling of articles: one could change the set of newspapers surveyed and obtain a quantitatively different (but related) measure.\footnote{Misclassification of articles is a second source of measurement error. We sidestep this for now for sake of exposition, but account for it in later sections.} To capture this, consider
\begin{equation}
    X_i \sim \mathrm{Binomial}(C_i,\theta_i) , \label{eq:toy.state}
\end{equation}
where $C_i$ is the number of articles sampled in month $i$, $X_i$ is the number of these that relate to uncertainty,
and $\theta_i$ is true policy uncertainty.
We observe $X_i$, $Y_i$, and $C_i$ but not $\theta_i$. One can estimate $\theta_i$ using $\hat\theta_i = X_i/C_i$, as done by \citeauthor{bakerMeasuringEconomicPolicy2016} (\citeyear[p.~1599]{bakerMeasuringEconomicPolicy2016}).

\subsection{Problem with the Two-Step Strategy}

In this example, the two-step strategy computes the OLS estimate $\hat \gamma$ from regressing $Y_i$ on $\hat{\theta}_i$, then performs standard OLS inference for $\gamma$. This approach ignores the fact that $\hat\theta_i$ is a noisy estimate of $\theta_i$, potentially leading to biased estimates and invalid inference.

We use asymptotics to tractably approximate the finite-sample problem faced by the researcher, where $\hat \gamma$ is computed from $(Y_i,X_i,C_i)_{i=1}^n$. Both measurement error and sampling error affect the properties of $\hat \gamma$ in finite samples. We therefore consider a sequence of populations indexed by the sample size $n$, where the distribution of $(Y_i,X_i,\theta_i)$ conditional on $C_i$ is fixed but the distribution of $C_i$ is changing with $n$ so that
\begin{equation} \label{eq:kappa.warm-up}
    \sqrt n \times \E \left[ \frac{1}{C_i} \right] \to \kappa \in [0,\infty) .
\end{equation}
In this sequence of DGPs, the variance of $\hat \theta_i$, which is proportional to $C_i^{-1}$, is of the same order of magnitude as sampling uncertainty. 
The parameter $\kappa$ controls the relative importance of measurement error, with larger values of $\kappa$ giving relatively greater importance to measurement error. Working with this sequence of DGPs therefore allows us to gain insights about how $\hat{\gamma}$ behaves when both measurement and sampling error are present.

Under suitable regularity conditions (see Theorem~\ref{theorem:drifting.general}), one can show that
\begin{equation}\label{eq:toy.result}
    \begin{aligned}
         & \sqrt n(\hat \gamma - \gamma )  \to_d N \left( -\kappa \, \gamma \frac{\E[ \theta_i (1 - \theta_i) ]}{\mathrm{Var}(\theta_i)} , \frac{\E[ \varepsilon_i^2 (\theta_i - \E[\theta_i])^2] }{\mathrm{Var}(\theta_i)^2} \right)  , \\
         & \frac{\sum_{i=1}^n \hat \epsilon_i^2 (\hat \theta_i - \bar \theta_n)^2}{\sum_{i=1}^n (\hat \theta_i - \bar \theta_n)^2} \to_p \frac{\E[ \varepsilon_i^2 (\theta_i - \E[\theta_i])^2] }{\mathrm{Var}(\theta_i)^2}  ,
    \end{aligned}
\end{equation}
where $\bar \theta_n$ is the sample mean of $\hat \theta_i$ and $\hat \epsilon_i$ is the OLS residual.
The first result shows $\hat \gamma$ is asymptotically normally distributed with the same variance as if $Y_i$ was regressed on the true latent $\theta_i$, but with a centering that differs from zero when $\kappa > 0$. The second result shows OLS standard errors are consistent, irrespective of $\kappa$. Taken together, these results imply that two-step confidence intervals (given by $\hat \gamma$ $\pm$ $1.96$ times the OLS standard error) have the correct width, but incorrect centering whenever $\kappa > 0$. Moreover, the asymptotic bias of $\hat \gamma$, and therefore the degree of under-coverage of two-step CIs, is increasing in $\kappa$.

\subsection{Proposed Solutions}

\subsubsection{Bias Correction}

Our first proposed solution is a straightforward bias correction. This approach simply constructs an estimate of the bias and adds it back to the two-step estimator $\hat \gamma$. The bias correction follows easily from (\ref{eq:toy.result}), so the bias corrected estimator $\hat \gamma^{bc}$ is
\[
    \hat \gamma^{bc} = \left( 1 + \frac{\hat \kappa}{\sqrt n} \frac{\sum_{i=1}^n \hat \theta_i (1-\hat \theta_i)}{\sum_{i=1}^n (\hat \theta_i - \bar \theta_n)^2} \right) \hat \gamma ,
\]
where $\hat \kappa = \frac{1}{\sqrt n} \sum_{i=1}^n C_i^{-1}$.
Bias corrected confidence intervals are then simply $\hat \gamma^{bc}$ $\pm$ $1.96$ times the OLS standard error. See Theorem~\ref{theorem:CI} for a formal justification for this approach.

\subsubsection{Joint Estimation}

A second approach is joint maximum likelihood estimation of (\ref{eq:toy.obs}) and (\ref{eq:toy.state}). This approach treats (\ref{eq:toy.state}) analogously to an observation equation in a state-space model, with $\theta_i$ as a latent variable.

We start by assuming the error terms in (\ref{eq:toy.obs}) have probability density function $\sigma^{-1} f(\varepsilon/\sigma)$. Combining with  (\ref{eq:toy.state}), this yields the likelihood
\[
    f(Y_i,X_i|C_i,\theta_i;(\gamma,\alpha,\sigma)) \propto \frac 1 \sigma f \left( \frac{Y_i - \alpha - \gamma \theta_i}{\sigma} \right) (\theta_i)^{X_i}(1-\theta_i)^{C_i - X_i}.
\]
Since $\theta_i$ is latent we proceed in the spirit of random effects and assume $\theta_i$ is drawn from a distribution with probability density function $g$ on $[0,1]$.\footnote{One can easily allow the distribution of $\theta_i$ to depend on covariates, as in correlated random effects. We suppress this for now, but adopt such an approach in the empirical applications.} Note that the effect of this prior will be dominated by the data as $C_i$ becomes large. We then integrate out $\theta_i$ to produce a likelihood in terms of the observed data:
\[
    f(Y_i,X_i|C_i;(\gamma,\alpha,\sigma)) = \int_0^1 f(Y_i,X_i|C_i,\theta_i;(\gamma,\alpha,\sigma)) g(\theta_i) \, d \theta_i.
\]
We estimate $\gamma$ by maximizing the log-likelihood
\[
    L_n((\gamma,\alpha,\sigma)) = \sum_{i=1}^n \log  f(Y_i,X_i|C_i;(\gamma,\alpha,\sigma)).
\]
Inference is performed using standard asymptotics for maximum likelihood estimators.


\section{General Setup and Applications} \label{sec:applications}

In the general model, we wish to estimate and perform inference on the parameters $\bs \gamma$ and $\bs \alpha$ of the linear regression model
\begin{equation}
  \label{eq:reg}
  Y_i = \boldsymbol{\gamma}^T \boldsymbol{\theta}_i  + \boldsymbol{\alpha}^T \mathbf{q}_i + \varepsilon_i,
\end{equation}
where $\boldsymbol{\theta}_i$ is now extended to be a vector of \emph{latent} variables of interest, 
$\mathbf{q}_i$ is a vector of \emph{observed} quantitative variables, 
and $\E \left[ \varepsilon_i (\boldsymbol{\theta}_i , \mathbf{q}_i) \right] = 0 $. For each observation $i$ we also have unstructured or high-dimensional data $\mf x_i$, from which an estimate $\hat{\bs{\theta}}_i$ of $\bs{\theta}_i$ can be derived. Thus, the researcher's dataset is a random sample $(Y_i, \mf q_i, \mf x_i)_{i=1}^n$. The parameter $\boldsymbol \gamma$ is typically the key object of interest, but in some cases (e.g., \citealp{aviviArePatentExaminers2024}) $\boldsymbol \alpha$ may be the focus, with  $\boldsymbol{\theta}_i$ serving as a control variable derived from unstructured data.

The dominant two-step strategy can be summarized as follows:
\begin{enumerate}
  \item[(i)] Compute estimates $\hat{\boldsymbol{\theta}}_i$ of $\boldsymbol{\theta}_i$ for all observations $i = 1,\ldots,n$.
  \item[(ii)] Regress $Y_i$ on $\hat{\boldsymbol{\theta}}_i$ and $\mathbf{q}_i$. Compute standard errors and confidence intervals, treating the $\hat{\boldsymbol{\theta}}_i$ as if they are regular numeric data.
\end{enumerate}
Evidently there is a measurement error problem: the estimates $\hat{\boldsymbol{\theta}}_i$ are proxies for the true latent covariates $\boldsymbol{\theta}_i$ in (\ref{eq:reg}). Step (ii)~overlooks this issue and treats the estimates $\hat{\boldsymbol{\theta}}_i$ as regular numeric data.  This raises the possibility that two-step estimators of $\bs \gamma$ and $\bs \alpha$ are biased. Moreover, conventional standard errors and confidence intervals do not account for the additional variation arising from using $\hat{\boldsymbol{\theta}}_i$ instead of $\boldsymbol{\theta}_i$, raising the possibility of a generated regressors problem. To understand the forces at play, in Section \ref{sec:model} we shall analyze the two-step strategy and formally demonstrate why it can lead to biased estimates and inference.  Many specific cases can be captured by this setup. We consider three here.\footnote{The previous version of the paper considers further extensions, for example regression onto similarity measures between vector representations of unstructured data.}

\paragraph{Application 1: AI/ML-Generated Labels.}
Economists now routinely use AI or ML methods to impute missing covariates. A leading use case involves regressions of an outcome $Y_i$ on a latent binary variable $\theta_i$ (e.g., indicating positive/negative sentiment of a news article or racial group membership) and observed controls $\mf q_i$.\footnote{We present the case of scalar $\bs\theta_i$ in the main text and defer the case of multiple categories to Appendix~\ref{sec:appendix_examples}.} Examples include \cite{goldsmith-pinkhamGenderGapHousing2023}, \cite{adams-prasslFirmConcentrationJob2023}, \cite{argyleRacialDisparitiesBias2025}, and \cite{wuBehavioralResponsesEstate2024}. Unstructured data $\mf x_i$ (e.g., article text or voter registration data) is often used to predict $\theta_i$ using a classification algorithm. The two-step strategy entails first generating a prediction $\hat{\theta}_i$ of $\theta_i$ then regressing $Y_i$ on $\hat{\theta}_i$ and $\mf q_i$. Here the source of measurement error is misclassification: $\hat \theta_i$ may differ from $\theta_i$ for some observations. 
Although sophisticated modern classifiers have low error rates, they are often used to impute missing observations for large data sets. As a result, measurement error from misclassification may be non-negligible relative to sampling error in the downstream regression, invalidating two-step inference.

\paragraph{Application 2: Topic Models.}

A large empirical literature uses topic models to reduce the dimension of unstructured data. 
Examples include \citet{hansenTransparencyDeliberationFOMC2018}, \citet{muellerReadingLinesPrediction2018}, \citet{larsenValueNewsEconomic2019}, \citet{thorsrudWordsAreNew2020}, \citet{adamsDeathCommitteeAnalysis2021}, \citet{bybeeBusinessNewsBusiness2024}, and \citet{ashMoreLawsMore2025} with text, \citet{dracaHowPolarizedAre2021}, and \citet{munroLatentDirichletAnalysis2022} with survey data, and \citet{nimczikJobMobilityNetworks2017} and \citet{olivellaDynamicStochasticBlockmodel2021} with network data.

Here $\mathbf{x}_i = (x_{i,j})_{j=1}^V$ is a $V$-dimensional vector of feature counts, with $x_{i,j}$ counting the number of times feature $j$ appears in observation $i$. For instance, in text applications, $x_{i,j}$ counts the number of times word $j$ appears in document $i$.
The vector $\mathbf{x}_i$ follows a Multinomial distribution with a factor structure. There are $K < V$ distributions $\boldsymbol{\beta}_1,\ldots,\boldsymbol{\beta}_K \in \Delta^{V-1}$, the $(V-1)$-dimensional simplex. Each $\boldsymbol{\beta}_k$ represents a common factor (or ``topic'').  Each observation $i$ is characterized by a latent vector $\boldsymbol{w}_i \in \Delta^{K-1}$. The elements $w_{i,k}$ of $\bs w_i$ represent the weight of $\boldsymbol{\beta}_k$ in generating $\mathbf{x}_i$.  The count probabilities for observation $i$ are $\mathbf{p}_i = \sum_{k=1}^K \boldsymbol{\beta}_{k} w_{i,k} = \mathbf{B}^T \boldsymbol{w}_i$, where $\mathbf{B}^T = [\boldsymbol{\beta}_1,\ldots,\boldsymbol{\beta}_K]$.
Combining these elements yields
\begin{equation}
  \label{eq:obs.1}
  \mf{x}_i|(C_i, \boldsymbol{w}_i)  \sim \mbox{Multinomial}(C_i, \mathbf{B}^T \boldsymbol{w}_i) ,
\end{equation}
where $C_i = \sum_{v=1}^V x_{i,v}$ is the total feature count for  observation $i$. Finally, the sub-vector $\bs \theta_i$ of $\bs w_i$ collects the topic weights for inclusion in the regression.

In the two-step approach, $\mf B$ and $(\bs \theta_i)_{i=1}^n$ are estimated using Latent Dirichlet Allocation \citep[LDA]{bleiLatentDirichletAllocation2003} or more recent methods (e.g., \cite{bingOptimalEstimationSparse2020}, \cite{wuSparseTopicModeling2023}, \cite{keUsingSVDTopic2022}), then $Y_i$ is regressed on $\hat{\bs \theta}_i$ and $\mf q_i$. The source of measurement error is sampling error in the estimated topic weights $(\hat{\bs{\theta}_i})_{i=1}^n$, which is proportional to $C_i^{-1}$. The quantity  $\E[C_i^{-1}]$ controls the overall rate of measurement error.

\paragraph{Application 3: AI/ML-Generated Indices.}

A third use case involves constructing indices by classification and aggregation. For instance, \cite{bakerMeasuringEconomicPolicy2016} constructs policy uncertainty indices by first classifying news articles based on whether they relate to policy uncertainty, and then aggregating the results over time into monthly or quarterly indices.  See also \cite{caldaraMeasuringGeopoliticalRisk2022} and \cite{gorodnichenkoVoiceMonetaryPolicy2023}, among others.  

Extending the simple example of Section~\ref{sec:warmup}, let $C_i$ denote the 
number of articles to be classified in month $i$, and let $N_i$ represent the number of articles classified as pertaining to policy uncertainty. The quantity $\hat \theta_i = N_i/C_i$ is a natural measure of the true latent uncertainty $\theta_i \in [0,1]$, where $\theta_i = 0$ indicates no uncertainty and $\theta_i = 1$ indicates maximal uncertainty. As each individual article classification may be subject to some error, there are now two sources of measurement error: misclassification error and sampling uncertainty (the set of articles are a sample from the broader corpus of news). Both can be accounted for using a topic model. Suppose the misclassification rates are constant across observations. Then $\mf n_i = (N_i, C_i - N_i)^T$ follows the distribution in (\ref{eq:obs.1}), with 
\[
 \mathbf B^T = \left[ \begin{array}{cc}
 \beta_1 & \beta_0 \\ (1-\beta_1) & (1-\beta_0) \end{array} \right], \quad \quad \bs w_i = \left[ \begin{array}{c} \theta_i \\ 1-\theta_i \end{array} \right],
\]
where $\beta_1$ is the probability that an article relating to uncertainty is correctly classified, and $\beta_0$ is the probability that an article not relating to uncertainty is misclassified. 

\bigskip

As in the simple example in Section \ref{sec:warmup}, we consider two strategies to correct bias and restore valid inference.  The first strategy involves bias-corrected estimators and confidence sets, which we formally develop and provide theoretical justification for in Section~\ref{subsection:bias_correction}.  The second is joint estimation, the implementation details of which are discussed in Section \ref{sec:one-step}.  In the remainder of the section, we illustrate the problems with two-step estimation\textemdash and how the proposed solutions fix them\textemdash across three applications drawn from diverse empirical literatures.  Each application corresponds to one of the examples above.

\subsection{Remote Work and Wage Inequality}
\label{sec:applications.1}

Since the COVID-19 pandemic, the incidence of remote work has risen remarkably \citep{barreroWhyWorkingHome2021,aksoyWorkingHomeWorld2022}.  But much of the evidence on remote work comes from surveys which are limited in sample size.  This makes tracking the evolution of remote work across narrow geographies and firms infeasible.  \citet{hansenRemoteWorkJobs2023} instead develops a dataset (available at \url{https://wfhmap.com/}) that measures remote work from a vast corpus of online job postings provided by Lightcast.  Each job posting contains 
metadata on occupation, firm, location, job title, and the posted wage, and a textual description of the job. 
\citet{hansenRemoteWorkJobs2023} uses the posting text to impute a binary label indicating whether or not the posting offers remote work. 
The authors collect human labels from Amazon Mechanical Turk for a sample of postings and use them to fine-tune DistilBERT \citep{sanhDistilBERTDistilledVersion2020}, a Large Language Model. 
The classifier achieves high overall accuracy in a held-out test set (98\%-99\%). 
The authors then impute a remote work label for every posting in the corpus.

The data can be used to answer numerous questions about the causes and consequences of remote work.  One important question is the degree of wage inequality in remote work arrangements, which the data has previously been used to study \citep{lambertResearchGrowingInequality2023}.  To explore this issue, we use regression \eqref{eq:reg} taking $Y_i$ as the log of the advertised wage for posting $i$.\footnote{We form the advertised wage by averaging the \textit{salary\_from} and \textit{salary\_to} fields in the Lightcast data.}  Here $\theta_i \in \{0,1\}$ is a latent indicator of whether posting $i$ offers remote work.  The two-step approach replaces the latent $\theta_i$ with the predicted label $\hat{\theta}_i \in \{0,1\}$ from the classifier. 
 Given the remote work classifier achieves high test-set accuracy, one may believe that such classification error is unlikely to affect inference meaningfully in this application.  However, what matters is not measurement error \textit{per se} but its magnitude relative to sampling error. 
In regressions with a large number of observations, sampling error is potentially small enough that even small classification errors can distort inference.

\begin{table}[t]
    \centering
    \captionsetup{justification=centering}
    \caption{Estimates of Impact of Remote Work on Posted Wages (San Diego, NAICS 72)}
    \label{tab:fe_estimates}
    {\small
    \begin{tabular}{lccc}
        \\[-8pt]
\hline \hline \\[-10pt]
          Controls                                       & \multicolumn{3}{c}{{Estimation Strategy}}                                             \\[2pt]
        \cline{2-4} \\[-8pt]
         & Two-Step                                & Bias Correction & Joint \\ \\[-10pt]
        None
                                                 & \makecell{0.648                                                                             \\ \small [0.600, 0.697]}
                                                 & \makecell{1.052                                                                             \\ \small [0.778, 1.327]}
                                                 & \makecell{0.563                                                                             \\ \small [0.532, 0.595]} \\ \\[-10pt]
        SOC2 effects
                                                 & \makecell{0.364                                                                             \\ \small [0.322, 0.406]}
                                                 & \makecell{0.641                                                                             \\ \small [0.446, 0.836]}
                                                 & \makecell{0.448                                                                             \\ \small [0.415, 0.480]} \\[12pt] \hline \hline
    \end{tabular}
}
    \vspace{0.5em}

    \parbox{0.95\textwidth}{\footnotesize
        \emph{Note}: Point estimates and 95\% confidence intervals for the slope coefficient in a regression of posted log salary on a binary remote work indicator, with and without occupation fixed effects at the SOC2 level. The sample consists of 16,315 job postings for 2022 and 2023 with ``San Diego, CA'' recorded as the city and ``72'' recorded as the NAICS2 industry code of the advertising firm.}
\end{table}

To illustrate the impact of classification error, we begin by organizing the data into NAICS2 industry $\times$ city cells for all job postings in 2022 and 2023 in the United States.  These years coincide with peak post-pandemic incidence of remote work. 
 Our theory in Section \ref{sec:model} shows in which situations classification error is likely to produce the largest distortions, 
 and we use it to select for case study NAICS 72 industry (Accommodation and Food Services) postings in San Diego.  After removing internships and observations with missing salary data, there are 16,315 observations.  The first column of Table \ref{tab:fe_estimates} contains two-step estimates of $\gamma$.  In the initial regression model, we include no controls and find an estimated 65 log point effect on wages of remote work.  Since this may in part reflect occupation composition, we next include SOC2 fixed effects in $\mathbf{q}_i$, which reduces the estimated $\gamma$ to 36 log points.

The positive association between remote work and posted wages is consistent with the descriptive evidence in \citet{lambertResearchGrowingInequality2023}.  Nevertheless, its strength may be distorted in the two-step approach.  A key quantity governing the bias is the expected false positive rate (FPR), which can be estimated by reading a random sample of postings and counting the number that were mis-classified as positive.  In this example, we took a random sample of 1000 postings, read the 26 that were classified as positive, and found nine were classified incorrectly, so the expected FPR is 0.009.  Despite the small FPR, bias correction produces much larger effects: with no controls, the effect increases by 62\% (to 105 log points) and, with controls, by 76\% (to 64 log points).\footnote{Bias-corrected confidence intervals are wider to account for the uncertainty in the estimated FPR.}  In many economic applications of AI and ML methods, classification accuracy is well below its level here.  Finding large effects even in this setting shows that classification error can be of first-order importance.

In addition, we formulate a joint model over $\theta_i$, $\hat{\theta_i}$, and $Y_i$.  While this produces a smaller estimated $\gamma$ without controls, in the preferred specification with occupation effects we continue to find a larger effect that falls outside the two-step confidence intervals.  One reason why joint estimation and bias correction may not coincide is that human labels themselves may not represent the ground truth even in the best-designed audit.  For example, the 2025 Economic Report of the President \citep{councilofeconomicadvisersEconomicReportPresident2025} discusses strengths and weaknesses of the \citet{hansenRemoteWorkJobs2023} database and points out that not every posting will explicitly mention remote work even when the firm in practice offers it.\footnote{The relevant passage is \begin{quote} Job openings data can also shed light on whether remote work is here to stay. While the information can be murky—given that not every hybrid or remote job advertises itself as such, and the tendency to mention remote work in job postings may change over time—examining recent trends is useful \citep{councilofeconomicadvisersEconomicReportPresident2025}. \end{quote}} 

\subsection{CEO Time Use and Firm Performance}
\label{sec:applications.2}

The role of CEOs in shaping firm performance is important for many academic and policy debates, but until recently little data existed on what CEOs do with their time.  To fill this evidence gap, \citet{bandieraCEOBehaviorFirm2020} collects and analyzes survey data on CEO time use in a sample of manufacturing firms.  The paper describes salient differences in executive time use and relates those differences to firm and CEO characteristics and firm outcomes.

The survey consists of five questions with categorical responses: (Q1) the type of activity (meeting, public event, etc.); (Q2) duration of activity (15m, 30m, etc.); (Q3) whether the activity is planned or unplanned; (Q4) the number of participants in the activity; (Q5) the functions of the participants in the activity (HR, finance, suppliers, etc.).  Survey responses are recorded for each 15-minute interval of a given week, e.g. Monday 8am-8:15am, Monday 8:15am-8:30am, and so forth.  The sample consists of 916 CEOs. 

The data are modeled as a topic model, with $V = 654$ answer combinations observed across the five questions, and  $x_{i,j}$ denoting the number of times combination $j$ appears in the diary of CEO $i$.
The authors reduce the dimensionality of the feature space using LDA with $K=2$ topics.  The estimated $\hat{\boldsymbol{\beta}}_1$ places relatively higher mass on features associated with ``management'' like visiting production sites and one-on-one meetings with employees or suppliers, while $\hat{\boldsymbol{\beta}}_2$ places relatively higher mass on features associated with ``leadership'' like communicating with other C-suite executives and holding large, multi-function meetings. The leadership weight $\hat{\theta}_i$\textemdash which the authors call a ``behavior index''\textemdash is thus a measure of the tendency of CEO $i$ to engage in leadership activities.  \citet{bandieraCEOBehaviorFirm2020} regress log sales on $\hat{\theta}_i$ and firm controls, finding a positive association between leadership and firm performance.  The paper's use of the two-step strategy may, however, lead to invalid inference.

To explore this possibility, we first replicate the authors' two-step strategy, estimating $\theta_i$ by LDA, then regressing the log sales of each CEO's firm on $\hat \theta_i$ and controls $\mf q_i$ including log employment, country fixed effects, and survey-wave fixed effects. 
Results reported in Table~\ref{tab:regression_time_use} show that, according to the two-step approach, moving from a CEO who only spends time in management to one who only spends time in leadership is associated with a 40 log point increase in sales.
What's more, neither bias correction nor joint estimation shifts the estimated effect size by a large amount.  This shows that bias is not an inevitable part of using AI/ML-generated variables, but rather depends on the empirical setting.

\begin{table}[t]
    \centering
    \captionsetup{justification=centering}
    \caption{Estimates of Impact of CEO Behavior on Firm Performance}
    \label{tab:regression_time_use}
    {\small
    \begin{tabular}{lccc}
        \\[-8pt]
\hline \hline \\[-10pt]
          Sample                                       & \multicolumn{3}{c}{{Estimation Strategy}}                                             \\[2pt]
        \cline{2-4} \\[-8pt]
         & Two-Step                                & Bias Correction & Joint \\ \\[-10pt]
        Full
                        & \makecell{0.405                                                                              \\ \small [0.224, 0.585]}
                        & \makecell{0.474                                                                              \\ \small [0.294, 0.655]}
                        & \makecell{0.402                                                                              \\ \small [0.240, 0.603]} \\
        \addlinespace
        10\% Subsample
                        & \makecell{0.227                                                                              \\ \small [-0.038, 0.492]}
                        & \makecell{1.054                                                                              \\ \small [0.789, 1.319]}
                        & \makecell{0.439                                                                              \\ \small [0.153, 0.711]} \\[12pt] \hline \hline
                            \end{tabular}}

    \vspace{0.5em}

    \parbox{0.95\textwidth}{\footnotesize
        \emph{Note}: The first row presents point estimates and 95\% confidence intervals for a regression of log sales on the CEO behavior index from the replication data of \cite{bandieraCEOBehaviorFirm2020}.  The second row presents estimates and confidence intervals using a 10\% subsample of time use responses per CEO.
    }

\end{table}

One reason why measurement error is less important here lies in the $\kappa$ expression \eqref{eq:kappa.warm-up} for the simple model.  While not directly applicable here as the topic model structure is more complex, it still allows one to qualitatively compare sampling error (reflected by $\sqrt{n}$) and measurement error (reflected by $\E \left[ C_i^{-1} \right]$).  The empirical analogue of this expression is $0.44$.  In other words, there are a relatively large number of survey responses per CEO 
compared to the number of surveyed CEOs, meaning the measurement error in $\hat \theta_i$ is small relative to sampling error.  To increase measurement error, we take a random 10\% subsample of survey responses for each CEO, which can be thought of as observing half a day of behavior rather than a five-day workweek.  The two-step approach now estimates a smaller insignificant effect, while both bias correction and joint estimation continue to estimate a significant effect.

\subsection{Central Bank Communication}
\label{sec:applications.3}

Public communication is an increasingly important tool for central banks \citep{blinderCentralBankCommunication2008,blinderCrystalBallDarkly2018} and disentangling its effects is an ongoing research challenge.  Communication events are often accompanied by rich, unstructured data like the content of speeches and press conferences, which require some quantification prior to econometric analysis.  To illustrate this, we consider market reactions to FOMC policy announcements.  Using high-frequency yield curve movements around FOMC announcements, \citet{gurkaynakActionsSpeakLouder2005} extracts separate ``target'' and ``path'' factors which account for large shares of observed variation in short- and long-run yields, respectively.  One view is that policy rate news drives short rates (i.e., the target factor) while 
communication in the form of 
written statements drives long rates (i.e., the path factor) by shifting expectations of future actions.

To test this, we estimate \eqref{eq:reg} with $Y_i$ as the path factor for FOMC meeting $i$ (as updated in \citealp{acostaConstructingHighfrequencyMonetary2024}), $\theta_i$ as the hawkish sentiment of the FOMC written statement, and $\mathbf{q}_i$ as a constant and the shadow short rate \citep{wuMeasuringMacroeconomicImpact2016}, which accounts for policy variation during the zero-lower-bound period.  The sample period is Feb 1995 through June 2023, during which 200 FOMC meetings take place.\footnote{The sample is determined by the availability of the path factor and shadow short rate data.}

\begin{figure}[p]
    \centering
    \subfloat[Unscaled Series]{\label{fig:sentiment_1}\includegraphics[scale=0.5]{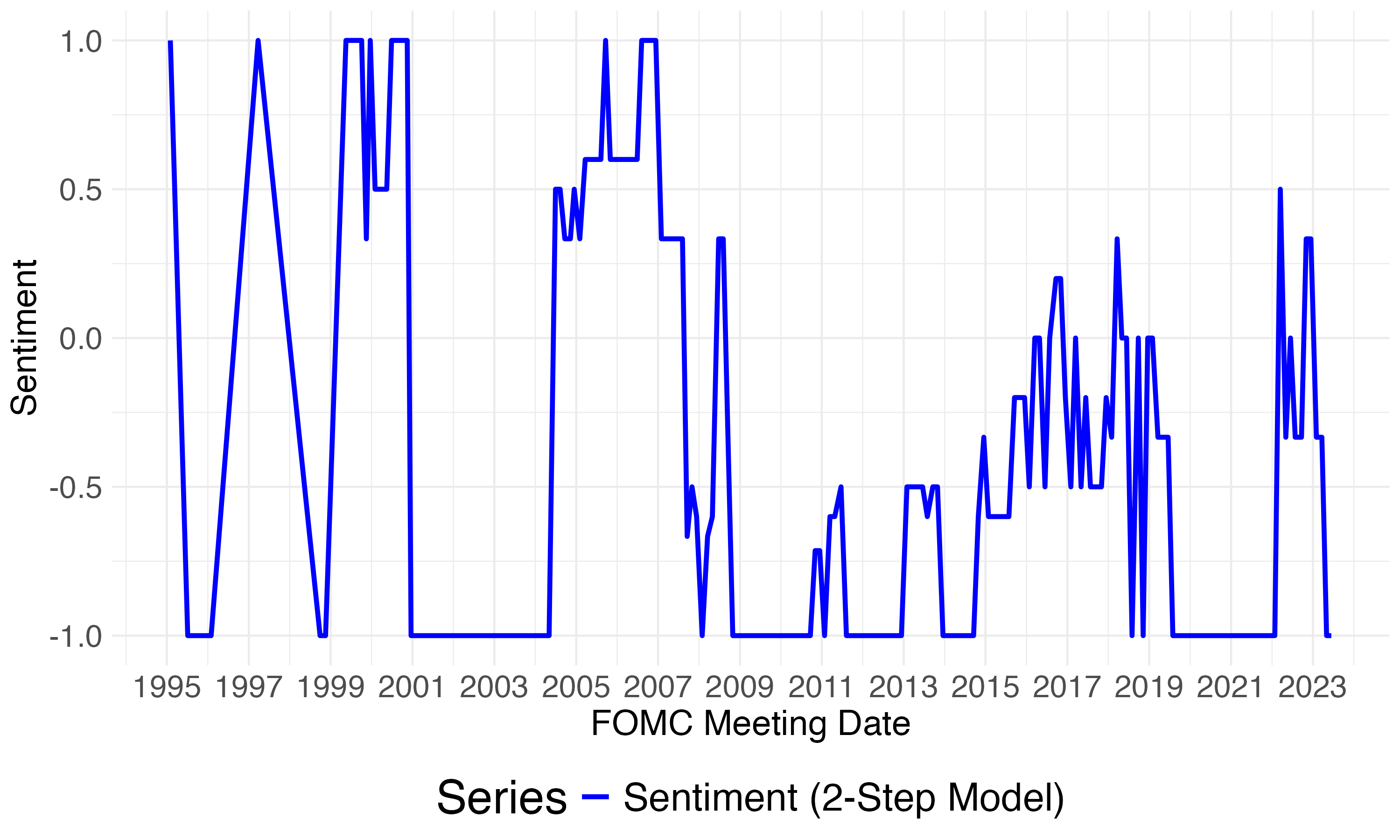}}

    \subfloat[Scaled Series]{\label{fig:sentiment_2}\includegraphics[scale=0.5]{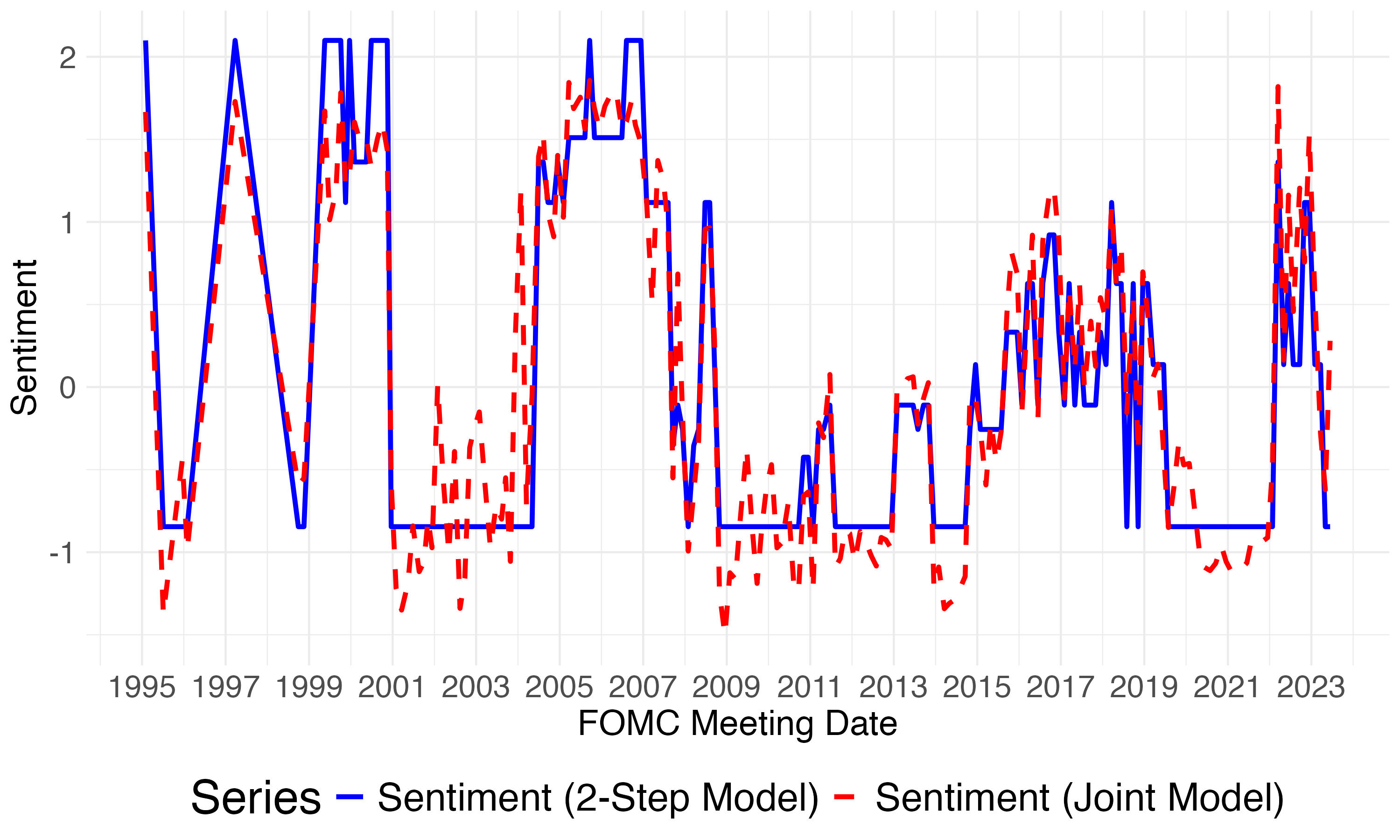}}
    
    \caption{Time Series of FOMC Statement Sentiment \label{fig:sentiment_series}} 
    
        \vspace{1em}
    
    \parbox{0.95\linewidth}{\footnotesize
            \emph{Note}: The left panel plots two-step sentiment $\hat\theta_i$.  The right panel plots standardized series for $\hat\theta_i$ and estimated sentiment from the joint model.  The sample period covers 200 FOMC meetings from Feb 1995 through June 2023.
    }
    
    \vspace{0.5em}
    
\end{figure}

Of course, hawkish sentiment $\theta_i$ is latent.  To estimate it, we replicate the approach of \citet{gorodnichenkoVoiceMonetaryPolicy2023} which uses classification and aggregation to build a sentiment index.\footnote{\citet{gorodnichenkoVoiceMonetaryPolicy2023} studies the market impact of the tone of voice of FOMC Chairs during post-meeting press conferences, and uses the hawkish sentiment of the written statement as a control.}  \cite{gorodnichenkoVoiceMonetaryPolicy2023} provides 1,243 sentences with human labels from FOMC statements from 1997 through 2010.  Each sentence assigned to one of three mutually exclusive categories: \textit{hawkish} (243 sentences), \textit{dovish} (511), or \textit{neutral} (489).  We split the sentences into 1,118 training observations and 125 test observations and fine-tune BERT \citep{devlinBERTPretrainingDeep2019}\textemdash a well-known Large Language Model\textemdash on the training data for label prediction.\footnote{More precisely, we fine tune \texttt{bert-base-uncased} via the Hugging Face library in Python.}  We assign each sentence in the test data to its most likely class based on the trained model and obtain test-set accuracy of 0.85, slightly higher than in the original paper (0.81).  We then predict a label for all paragraphs from FOMC statements in our sample.\footnote{Here we follow \cite{gorodnichenkoVoiceMonetaryPolicy2023} and classify individual \textit{paragraphs} in the full sample although BERT is trained on human-labeled \textit{sentences}.  The optimal level of document granularity for classification prior to aggregation is an open question.  \cite{gorodnichenkoVoiceMonetaryPolicy2023} classifies individual paragraphs of FOMC press conferences in addition to FOMC statements, but press conferences only began in 2011 where statements have been released since the mid-1990s.}  The overall sentiment measure for FOMC meeting $i$ is $\hat\theta_i = N_i/C_i$ where $C_i$ is the number of sentences classified as hawkish or dovish, and $N_i$ is the number classified as hawkish.  Figure~\ref{fig:sentiment_1} shows the evolution of this measure over the sample period.

The first column of Table \ref{tab:regression_sentiment} contains the results of two-step estimation.  There is some evidence consistent with communication rather than policy driving the path factor.  The coefficient on $\hat\theta_i$ is positive and significant while that on the policy rate is essentially zero.  At the same, the effect is somewhat weak.  The standard deviation of the path factor in the sample is 0.103, so the estimated $\hat\gamma$ implies a one-third standard deviation effect of $\theta_i$ moving from $0$ to $1$. The regression also has a low $R^2$.

\begin{table}[t]
    \centering
    \captionsetup{justification=centering}
    \caption{Impact of FOMC Statement Sentiment on Longer-Term Yields}
    \label{tab:regression_sentiment}
{\small
    \begin{tabular}{lcc}
     \\[-8pt]
\hline \hline \\[-10pt]
                                                 & \multicolumn{2}{c}{{Estimation Strategy}}                                             \\[2pt]
        \cline{2-3} \\[-8pt]
         & Two-Step                                &  Joint \\ \\[-10pt]
      
        Sentiment ($\theta_i$) & 0.038                                            & 0.114                   \\
                               & \small{[0.005, 0.071]}                           & \small{[0.027, 0.198]}  \\
        Policy Rate ($q_i$)    & -0.004                                           & -0.003                  \\
                               & \small{[-0.013, 0.004]}                          & \small{[-0.011, 0.004]} \\[2pt]
        \hline \\[-10pt]
        $\beta_0$              &                                                  & 0.009                   \\
                               &                                                  & \small{[0.001, 0.026]}  \\
        $\beta_1$              &                                                  & 0.676                   \\
                               &                                                  & \small{[0.585, 0.768]}  \\[2pt]
        \hline \\[-10pt]
        $R^2$                  & 0.0425                                           & 0.1429                  \\[2pt] \hline \hline
    \end{tabular}
}
    \vspace{0.5em}

    \parbox{0.95\textwidth}{\footnotesize
        \emph{Note}: The first column reports point estimates and 95\% confidence intervals using the two-step approach for estimating the relationship between hawkish sentiment in FOMC statements and the path factor of \citet{gurkaynakActionsSpeakLouder2005} updated by \citet{acostaConstructingHighfrequencyMonetary2024}.  The second column reports results from joint estimation.
    }
\end{table}

We do not develop a bias correction for index construction.\footnote{Our bias corrections are valid provided  measurement error and sampling error are comparable. The relatively short length of FOMC announcements and large misclassification errors we report below suggests measurement error may be sizable in this case. Joint estimation remains valid in such settings.} We do however implement joint estimation.  The second column of Table~\ref{tab:regression_sentiment} shows the results.  The effect size on sentiment nearly triples, as does the $R^2$ value. Confidence intervals are a little wider using this method, as it accounts for the fact that some FOMC announcements are relatively short, making their corresponding $\theta_i$ values difficult to infer.  Figure~\ref{fig:sentiment_2} shows that estimated sentiment from the two approaches co-move strongly, but estimated sentiment from the joint model is smoother.\footnote{While the joint approach is based on maximizing the integrated likelihood, our inference algorithm is based on sampling from the posterior distribution over $\theta_i$ and regression parameters.  Sentiment from the joint model is the average value of $\theta_i$ draws.}  In short, we find results consistent with measurement error in the two-step sentiment estimate which weakens the estimated market reaction to FOMC statements.

These results also have implications for the large literature that regresses outcomes on sentiment.  \citet{shapiroMeasuringNewsSentiment2022} compares a variety of classifiers for predicting human labels for news article sentiment, including BERT, and reports substantial error even for the highest-performing methods (see Tables 1-3 in the paper).  Approaches like ours can restore valid inference in such settings.

Table \ref{tab:regression_sentiment} also shows estimates of $\beta_0$ and $\beta_1$, and shows the latter is well below that implied by the classifier.  Economic theory predicts that the difference between realized and expected hawkishness is the relevant object for generating market reactions, whereas the human labels only classify realized hawkishness.  While a full analysis is outside the scope of the paper, joint estimation in principle allows $\beta_0$ and $\beta_1$ to adjust to account for this broader notion of misclassification.


\section{Why Two-Step Inference is Biased} \label{sec:model}

The empirical applications show that measurement error in AI/ML-generated variables can bias inference in a variety of settings.  Each application has a particular algorithm for estimating $\hat{\bs\theta}_i$ that differs in important respects: supervised vs unsupervised learning, whether $\hat{\bs\theta}_i$ is the aggregation of multiple or single classified documents, and so forth.  Further varieties are also common in the literature.  To better understand the biases arising from measurement error and how to fix them, we begin by abstracting away from algorithmic-specific details of how $\hat{\bs\theta}_i$ is estimated. This allows us to develop a general, broadly applicable framework, which we then tailor to some specific use cases.

\subsection{Theory for the Two-Step Strategy}

We first introduce some notation.
Let
\[
  \bs \psi = \left[ \begin{array}{c} \bs \gamma \\ \bs \alpha \end{array} \right] , \quad \quad
  \bs \xi_i = \left[ \begin{array}{c} \bs \theta_i \\ \mf q_i \end{array} \right] , \quad \quad
  \hat{\bs \xi}_i = \left[ \begin{array}{c} \hat{\bs \theta}_i \\ \mf q_i \end{array} \right] .
\]
The OLS estimator of $\boldsymbol{\psi}$  in the two-step strategy is given by
\begin{equation}
  \label{eq:ols}
  \hat{\boldsymbol{\psi}} = \left( \frac 1n \sum_{i=1}^n \hat{\bs{\xi}}_i^{\phantom{T}}\!\!\!\;\hat{\bs{\xi}}_i^T \right)^{-1} \left( \frac 1n \sum_{i=1}^n \hat{\bs{\xi}}_i Y_i \right).
\end{equation}
The OLS estimators $\hat{\bs\gamma}$ and $\hat{\bs\alpha}$ of $\bs\gamma$ and $\bs \alpha$ are the upper and lower blocks of $\hat{\bs \psi}$.

We use asymptotic theory to derive tractable approximations to the finite-sample distribution of $\hat{\bs \psi}$. Both measurement error in $\hat{\bs \theta}_i$ and downstream sampling error determine the finite-sample distribution. To ensure that asymptotics deliver a useful approximation, we adopt a framework in which both sources of error remain present as the sample size  $n$ becomes large. We do so be allowing the precision of $\hat{\bs\theta}_i$ to increase with $n$ at an appropriate rate. Formally, we consider a sequence of populations in which the distribution of $(Y_i,\boldsymbol{\theta}_i,\mathbf{q}_i)$ is held fixed and the conditional distribution of $\hat{\boldsymbol{\theta}}_i$ given $(Y_i,\boldsymbol{\theta}_i,\mathbf{q}_i)$ varies with $n$, so that
\begin{equation}\label{eq:kappa}
  \frac{1}{\sqrt n} \sum_{i=1}^n \hat{\bs{\theta}}_i(\hat{\boldsymbol{\theta}}_i - \bs{\theta}_i)^T \to_p \kappa \, \bs \Omega
\end{equation}
as $n \to \infty$, where $\bs \Omega$ is a finite non-random matrix and $\kappa$ is a non-negative constant. The constant $\kappa$  represents the \emph{relative} magnitudes of measurement error and sampling error. The matrix $\bs \Omega$ is related to the variance of measurement error in settings where the error in $\hat{\bs \theta}_i$ is ``classical''. However, condition (\ref{eq:kappa}) also allows for ``non-classical'' measurement error, which is needed to accommodate a number of important use cases\textemdash including imputed categorical labels \citep{aignerRegressionBinaryIndependent1973}.
 We show below that (\ref{eq:kappa}) holds for AI/ML-generated labels and topic models and derive expressions for $\kappa$ and $\bs \Omega$, illustrating how $\kappa$ links measurement error in $\hat{\bs \theta}_i$ to the sample size.

We view this asymptotic framework as appropriate for approximating modern use cases where high-performance algorithms are deployed on large data sets.\footnote{In such scenarios, conventional measurement error frameworks in which the size of the measurement error remains fixed seem inappropriate, since measurement error bias eventually dominates and estimators are inconsistent. For completeness, we provide a set of results for this case in Appendix~\ref{sec:appendix_examples}.} While it is not meant to be taken literally, a heuristic interpretation is that it captures the prevailing trend whereby increasingly large datasets are analyzed by increasingly accurate algorithms.

Having introduced the asymptotic framework, we now introduce the assumptions. We first present a general ``high-level'' set of assumptions, then verify them below within the context of the running examples. In what follows, notions of convergence in probability and distribution should be understood as holding along this sequence of populations satisfying condition~(\ref{eq:kappa}). Let $\hat \varepsilon_i = Y_i - \hat{\bs{\psi}}^T \hat{\bs{\xi}}_i$. Throughout, we let $\mf 0$ denote a matrix or vector of zeros whose dimension is determined by the context.

\begin{assumption}\label{assumption:drifting.general}
  \begin{enumerate}[label=(\roman*),nosep]
    \item \label{assumption:drifting.moments}
          $\E\Big[ \|\boldsymbol \xi_i\|^2 \Big] < \infty$, $\E\Big[ \|\varepsilon_i \boldsymbol \xi_i\|^2 \Big] < \infty$, and $\E\Big[\bs{\xi}_i^{\phantom{T}}\!\!\!\;\bs{\xi}_i^T\Big]$ has full rank.

    \item \label{assumption:drifting.lln}
          $\frac 1n \sum_{i=1}^n \hat{\bs{\xi}}_i^{\phantom{T}}\!\!\!\;\hat{\bs{\xi}}_i^T \to_p \E\Big[{\bs{\xi}}_i^{\phantom{T}}\!\!\!\;{\bs{\xi}}_i^T\Big]$,
          $\frac{1}{\sqrt n} \sum_{i=1}^n \hat{\bs{\theta}}_i(\hat{\boldsymbol{\theta}}_i - \bs{\theta}_i)^T \to_p \kappa \, \mathbf \Omega$, and
          $\frac{1}{\sqrt n} \sum_{i=1}^n (\hat{\boldsymbol{\theta}}_i^{\phantom{T}} - {\bs{\theta}}_i^{\phantom{T}}\!\!\!\;)\mf{q}_i^T \to_p \mathbf 0$ as $n \to \infty$.

    \item \label{assumption:drifting.clt}
          $\frac{1}{\sqrt n} \sum_{i=1}^n \hat{\bs{\xi}}_i \varepsilon_i \to_d N(\mathbf 0, \E[\varepsilon_i^2 {\bs{\xi}}_i^{\phantom{T}}\!\!\!\;{\bs{\xi}}_i^T] )$ and
          $\frac 1n \sum_{i=1}^n \hat \varepsilon_i^2 \hat{\bs{\xi}}_i^{\phantom{T}}\!\!\!\;\hat{\bs{\xi}}_i^T \to_p \E[\varepsilon_i^2 {\bs{\xi}}_i^{\phantom{T}}\!\!\!\;{\bs{\xi}}_i^T]$ as $n \to \infty$.
  \end{enumerate}
\end{assumption}

Assumption~\ref{assumption:drifting.general}\ref{assumption:drifting.moments} is standard. Assumption~\ref{assumption:drifting.general}\ref{assumption:drifting.lln} can be verified using appropriate laws of large numbers. The first part of Assumption~\ref{assumption:drifting.general}\ref{assumption:drifting.clt} imposes a standard CLT condition while the second part is only used to establish consistency of standard errors. We focus on the case where the data are independent and identically distributed to simplify exposition, though this can be relaxed and the results can easily be extended to general types of dependence.  

We now present our main result for this section, which shows that $\hat{\boldsymbol{\psi}}$ is consistent, derives its asymptotic distribution, and establishes consistency of standard errors. Let
\[
  \mathbf V = \E\Big[\bs{\xi}_i^{\phantom{T}}\!\!\!\;\bs{\xi}_i^T\Big]^{-1} \E\Big[\varepsilon_i^2 \bs{\xi}_i^{\phantom{T}}\!\!\!\;\bs{\xi}_i^T\Big] \E\Big[\bs{\xi}_i^{\phantom{T}}\!\!\!\;\bs{\xi}_i^T\Big]^{-1}
\]
denote the asymptotic variance of the OLS estimator in the infeasible regression of $Y_i$ on the true latent $\bs \theta_i$ and $\mathbf q_i$. Also let
\begin{equation} \label{eq:vhat}
  \hat{\mathbf V} = \left( \frac 1n \sum_{i=1}^n \hat{\bs{\xi}}_i^{\phantom{T}}\!\!\!\;\hat{\bs{\xi}}_i^T \right)^{-1}
  \left( \frac 1n \sum_{i=1}^n \hat \varepsilon_i^2 \hat{\bs{\xi}}_i^{\phantom{T}}\!\!\!\;\hat{\bs{\xi}}_i^T \right)
  \left( \frac 1n \sum_{i=1}^n \hat{\bs{\xi}}_i^{\phantom{T}}\!\!\!\;\hat{\bs{\xi}}_i^T \right)^{-1}
\end{equation}
denote the covariance matrix estimator for the regression of $Y_i$ on $\hat{\bs \theta}_i$ and $\mathbf q_i$.

\begin{theorem} \label{theorem:drifting.general}
  Suppose that Assumption~\ref{assumption:drifting.general} holds. Then as $n \to \infty$:
  \begin{enumerate}
    \item The OLS estimator $\hat{\bs \psi}$ from regressing $Y_i$ on $\hat{\bs \theta}_i$ and $\mathbf q_i$ is consistent and asymptotically normally distributed, but with a centering that may differ from zero:
          \begin{equation}
            \label{eq:two-step.drifting.result.1}
            \sqrt n ( \hat{\bs{\psi}} - \bs{\psi} ) \\[4pt]
            \to_d
            N \left( -\kappa \, \E\Big[\bs{\xi}_i^{\phantom{T}}\!\!\!\;\bs{\xi}_i^T\Big]^{-1} \left[ \begin{array}{cc}
                \bs{\Omega} & \mathbf 0 \\
                \mathbf 0   & \mathbf 0\end{array} \right] \bs \psi , \,\mathbf V \right);
          \end{equation}
    \item Two-step standard errors are consistent:
          \begin{equation} \label{eq:two-step.drifting.result.2}
            \hat{\mathbf V} \to_p \mathbf V.
          \end{equation}
  \end{enumerate}
\end{theorem}

Theorem~\ref{theorem:drifting.general} shows that two-step inference is  \emph{invalid} when $\kappa > 0$. In this case, standard errors are consistent but the asymptotic distribution is centered away from the origin due to measurement error bias. As a result, confidence intervals based on the usual two-step strategy have the correct width but incorrect centering, leading to coverage rates below nominal coverage.\footnote{We omit discussion of the case $\kappa = +\infty$ where measurement error dominates sampling error. In that case, the coverage rates of standard OLS confidence intervals approach zero as $n$ becomes large.}
The bias---and thus the degree of under-coverage---increases in $\kappa$. 
Simulations reported in Section~\ref{sec:simulations} show that coverage distortions can be severe even for small values of $\kappa$. Moreover, because measurement error can be nonclassical, it can be difficult to know even the sign of the bias: there may be attenuation or amplification.
These critiques apply to inference on $\bs \alpha$ as well as $\bs \gamma$, and are therefore relevant for researchers using unstructured data to create control variables.

On the other hand,  two-step inference is \emph{valid} when $\kappa = 0$. In this case, measurement error is of smaller order than sampling error and can effectively be ignored.  Here $\hat{\bs \psi}$ has the same $N(\mathbf 0, \mathbf V)$ asymptotic distribution as the (infeasible) OLS estimator obtained by regressing $Y_i$ on the true latent $\boldsymbol{\theta}_i$ and standard errors computed using $\hat{\boldsymbol{\theta}}_i$ are consistent.

\begin{remark} \normalfont
  These implications contrast with a generated regressors problem, where the asymptotic variance is inflated but there is no location shift. In the classical generated regressor problem  \citep{paganEconometricIssuesAnalysis1984}, the $\hat{\boldsymbol{\theta}}_i$ depend on a common finite-dimensional parameter that is estimated in the first stage. This across-observation dependence causes the term in~(\ref{eq:kappa}) to converge to a random variable rather than a constant, leading to the variance inflation.
\end{remark}

\begin{remark} \normalfont
  Our asymptotic framework is related to an econometrics literature on ``small'' classical measurement error (e.g., \citet{chesherEffectMeasurementError1991}), in which the variance of measurement error shrinks to zero at rate $n^{-1/2}$. Recently, \citet{evdokimovSimpleEstimationSemiparametric2023} show how to bias-correct GMM estimators in this context. Their approach imposes no structure on the source of measurement error and uses instrumental variables to identify its variance. Our setting is 
  different: measurement error arises due to first-stage estimation of $\bs \theta_i$, allowing us to analytically characterize and correct bias without an instrument. 
\end{remark}

We now derive expressions for $\kappa$ and $\Omega$ in the running examples. We shall use these in the next section to perform bias corrections.

\subsubsection{Application 1: AI/ML-Generated Labels.}
\label{sec:labels.theory.main}

To mimic scenarios where high-performance classifiers are deployed at scale, we consider a sequence of DGPs where the distribution of $(Y_i, \bs \theta_i, \mf q_i)$ is fixed but the distribution of $\mf x_i|(Y_i, \bs \theta_i, \mf q_i)$ varies with $n$ so that $\mf x_i$ becomes increasingly more informative about $\bs \theta_i$. For brevity we present conditions for the case of a deterministic classifier: $\hat \theta_i = \pi(\mf x_i)$ for some function $\pi$ taking values in $\{0,1\}$. We treat $\pi$ as deterministic, conditioning on the external training data set where necessary. Appendix~\ref{sec:appendix_examples.labels} presents a more general treatment allowing stochastic classifiers and multiple categories.

\begin{assumption}\label{assumption:drifting.ml.main}
  \begin{enumerate}[label=(\roman*),nosep]

    \item \label{assumption:kappa.ml.main}
          $\sqrt n \, \E[\hat \theta_i(1-\theta_i)] \to \kappa$.

    \item \label{assumption:dgp.ml.2.main}
          $\E\left[ \|\mathbf q_i\|^4 \right] < \infty$, $\E\left[\varepsilon_i^4\right] < \infty$, and $\E\left[\bs{\xi}_i^{\phantom{T}}\!\!\!\;\bs{\xi}_i^T\right]$ has full rank.

    \item \label{assumption:q.unbiased.2.main}
          $\E[ (\hat \theta_i - \theta_i) \mf{q}_i] = \mf 0$.

    \item \label{assumption:epsilon.2.main}
          $\E[ \hat \theta_i \varepsilon_i] = \mf 0$.
  \end{enumerate}
\end{assumption}

Assumption~\ref{assumption:drifting.ml.main}\ref{assumption:kappa.ml.main} is the key drifting-sequence condition. It says that the false-positive rate $\E[\hat \theta_i(1-\theta_i)]$ goes to zero at rate $n^{-1/2}$ (or faster). This allows for the classifier to produce misclassifications which individually occur with low probability, but whose cumulative effect will be non-negligible relative to sampling error when $\kappa > 0$. Assumption~\ref{assumption:drifting.ml.main}\ref{assumption:dgp.ml.2.main} is standard. Assumption~\ref{assumption:drifting.ml.main}\ref{assumption:q.unbiased.2.main} says $\mathbf q_i$ and the prediction errors $\hat \theta_i - \theta_i$ are orthogonal.  It is straightforward to relax this condition; doing so will simply add nonzero off-diagonal terms in the bias expression in~(\ref{eq:two-step.drifting.result.1}) without altering our main point: the two-step strategy can lead to biased inference. Finally, Assumption~\ref{assumption:drifting.ml.main}\ref{assumption:epsilon.2.main} says the true regression errors $\varepsilon_i$ are uncorrelated with the AI/ML-generated prediction $\hat \theta_i$. As $\varepsilon_i$ and $\theta_i$ are assumed uncorrelated, in effect this condition simply requires that the prediction error $\hat \theta_i - \theta_i$ and $\varepsilon_i$ are uncorrelated.

\begin{theorem} \label{theorem:drifting.ml.main}
  Suppose that Assumption~\ref{assumption:drifting.ml.main} holds. Then Assumption~\ref{assumption:drifting.general} holds and the OLS estimator $\hat{\bs \psi}$ has asymptotic distribution given by (\ref{eq:two-step.drifting.result.1}) with $\kappa = \lim_{n \to \infty} \sqrt n \, \E[\hat \theta_i(1-\theta_i)]$ and $\Omega = 1$. Moreover, two-step standard errors are consistent.
\end{theorem}

\subsubsection{Application 2: Topic Models.}
\label{sec:topic.theory.main}

In many modern empirical settings, there may be a large number of observations (large $n$) and a large amount of unstructured data per observation (large $C_i$). To mimic such scenarios, we consider a sequence of populations in which the distribution of $(Y_i,\mf q_i,\bs w_i)$ is fixed and the conditional distribution of $(\mf x_i, C_i)$ given $(Y_i,\mf q_i,\bs w_i)$ varies with $n$ so that (\ref{eq:obs.1}) holds and
\begin{equation}\label{eq:kappa.topic}
  \sqrt n \,\E \left[ \frac{1}{C_i} \right] \to \kappa .
\end{equation}
Smaller values of $\kappa$ correspond to settings where there is a relatively more unstructured data per observation, and hence relatively less measurement error. Also note that since $C_i$ enters via its inverse,  if most documents are large but a few are small, then $\kappa$ may still be large.

To simplify some expressions, in what follows we implicitly assume that the document size $C_i$ is independent of $(\boldsymbol{w}_i, \mathbf{q}_i,Y_i)$. We also assume that $\mf x_i$ and $\mf q_i$ are independent conditional on $(C_i,\bs{w}_i)$, and that $\varepsilon_i$ and $(\mf x_i,C_i)$ are independent conditional on $(\boldsymbol{w}_i,\mathbf q_i)$. In effect, the latter two assumptions ensure the multinomial sampling error and regression errors are uncorrelated.  These assumptions seem very reasonable and can be relaxed: doing so simply complicates the expressions below. We also slightly strengthen (\ref{eq:reg}) to require that $\E[\varepsilon_i(\bs w_i, \mf q_i)] = \mf 0$. That is, no relevant topic weights have been omitted from the regression.

\begin{assumption}\label{assumption:drifting}

  \begin{enumerate}[label=(\roman*),nosep]

    \item \label{assumption:kappa.topic}
          $ \sqrt n \, \E \left[ \frac{1}{C_i} \right] \to \kappa$.

    \item \label{assumption:B.rank.2}
          $\mf{B}$ has full rank.

    \item \label{assumption:B.rate}
          $\sqrt n (\hat{\mf{B}} - \mf{B}) \to_p \mathbf 0$.

    \item \label{assumption:theta.rate}
          $\sqrt n \max_{1 \leq i \leq n} \| \hat{\bs{\theta}}_i - \bs S(\hat{\mf{B}} \hat{\mf{B}}^T)^{-1} \hat{\mf{B}} (\mf x_i/C_i) \| \to_p 0$.

    \item \label{assumption:dgp.q.2}
          $\E\left[ \|\mathbf q_i\|^4 \right] < \infty$, $\E\left[\varepsilon_i^4\right] < \infty$, and $\E\left[\bs{\xi}_i^{\phantom{T}}\!\!\!\;\bs{\xi}_i^T\right]$ has full rank.

    \item \label{assumption:C}
          $C_i \gtrsim (\log n)^{1+\epsilon}$ almost surely for some $\epsilon > 0$.

  \end{enumerate}

\end{assumption}

Assumption~\ref{assumption:drifting}\ref{assumption:kappa.topic} is the key drifting sequence condition.
Assumption~\ref{assumption:drifting}\ref{assumption:B.rank.2} says that none of the topics are redundant.
\cite{bingOptimalEstimationSparse2020}, \cite{wuSparseTopicModeling2023}, and \cite{keUsingSVDTopic2022} show various estimators $\hat{\mathbf{B}}$ converge at the optimal rate $(nC)^{-1/2}$ (up to log terms) where, for simplicity, all $C_i$ are of the same order $C$ (i.e., $C_i \asymp C$). Hence, their estimators all satisfy Assumption~\ref{assumption:drifting}\ref{assumption:B.rate} when $C$ grows with $n$, as we have here by (\ref{eq:kappa.topic}).
Assumption~\ref{assumption:drifting}\ref{assumption:theta.rate} imposes some structure on the $\hat{\bs \theta}_i$ to facilitate derivations. This condition is not vacuous: we have $\boldsymbol{\theta}_i = \boldsymbol S(\mathbf{B} \mathbf{B}^T)^{-1}\mathbf{B} \E[\mathbf x_i/C_i|C_i, \boldsymbol w_i]$ by display (\ref{eq:obs.1}) and Assumption~\ref{assumption:drifting}\ref{assumption:B.rank.2}. Hence, one could take $\hat{\boldsymbol{\theta}}_i = \boldsymbol{S}(\hat{\mf{B}} \hat{\mf{B}}^T)^{-1} \hat{\mf{B}} (\mathbf x_i/C_i)$, in which case Assumption~\ref{assumption:drifting}\ref{assumption:theta.rate} trivially holds.
Assumption~\ref{assumption:drifting}\ref{assumption:dgp.q.2} is standard.
Assumption~\ref{assumption:drifting}\ref{assumption:C} is made to simplify arguments establishing consistency of standard errors and can be relaxed. 
This condition trivially holds in view of (\ref{eq:kappa.topic}) when all $C_i$ are of the same order, since in that case $C_i \gtrsim n^{1/2}$.
We note that these conditions implicitly assume $\mf B$ is identified. We defer discussion of identification to Appendix~\ref{sec:appendix_examples.topic}.

\begin{theorem} \label{theorem:two-step.drifting}
  Suppose that Assumption~\ref{assumption:drifting} holds. Then Assumption~\ref{assumption:drifting.general} holds and the OLS estimator $\hat{\bs \psi}$ has asymptotic distribution given by (\ref{eq:two-step.drifting.result.1}) with $\kappa = \lim_{n \to \infty} \sqrt n \, \E[C_i^{-1}]$ and
  \[
    \mathbf \Omega = \boldsymbol S (\mf{B} \mf{B}^T)^{-1} \mf{B}\, \mr{diag}(\mf{B}^T \E[\bs{w}_i]) \mf{B}^{T} (\mf{B} \mf{B}^T)^{-1} \boldsymbol S^T - \E\left[\bs{\theta}_i^{\phantom{T}}\!\!\!\;\bs{\theta}_i^T\right] .
  \]
  Moreover, two-step standard errors are consistent.
\end{theorem}

\begin{remark} \normalfont
  A conceptually related problem involves factor-augmented regressions. In their simplest form, latent factors $\mathbf{F}_t$ are imputed from a vector of $N$ predictor variables $\mathbf{x}_t$ using PCA, then the estimated factors $\hat{\mathbf{F}}_t$ are used as covariates in a regression.
  \cite{baiConfidenceIntervalsDiffusion2006} show this approach leads to valid inference provided $\sqrt{T}/N \to 0$, where $T$ is the time-series dimension  and $N$ is the cross-sectional dimension.\footnote{See \cite{cahanFactorbasedImputationMissing2023} and references therein for the related problem of using factor models to impute missing observations.} Their $T$ is analogous to our $n$, and, within the context of topic models, their $1/N$ is analogous to our $\E[C_i^{-1}]$. Thus, their condition $\sqrt{T}/N \to 0$ is analogous to $\kappa = 0$. \cite{goncalvesBootstrappingFactoraugmentedRegression2014} show that if $\sqrt{T}/N$ converges to a constant, analogous to $\kappa > 0$, then there is a bias that shifts the location of the asymptotic distribution. At an abstract level, Theorem~\ref{theorem:drifting.general} can be seen as generalizing this finding to a broad class of scenarios. Our theory for AI/ML-generated labels and topic models is new and does not follow from these existing works.
\end{remark}


\section{How To Debias Inference} \label{sec:integrated}

We now propose two methods for performing valid inference on $\bs \gamma$ and $\bs \alpha$: (1) bias corrected estimators and confidence intervals, and (2) joint estimation of the upstream and downstream regression models.  These are the approaches illustrated in the applications in Section~\ref{sec:applications}.

Each method has its strengths and weaknesses.  The bias correction is simple to implement and scales well, but the formulas for the bias we develop are for particular settings.  While these encompass leading applications, they are not exhaustive, and other formulas will need to be derived for other settings by specializing our general characterization of bias.  Joint estimation is more flexible and can handle cases where the general bias formula may be difficult to specialize, but it is also more computationally demanding.

\subsection{Bias-Corrected Estimators and Confidence Intervals \label{subsection:bias_correction}}

Theorem~\ref{theorem:drifting.general} shows that the asymptotic bias of the two-step estimator $\hat{\bs \psi}$ takes the form
\[
  \boldsymbol{b} = -\kappa \, \E\Big[\bs{\xi}_i^{\phantom{T}}\!\!\!\;\bs{\xi}_i^T\Big]^{-1} \left[ \begin{array}{cc}
      \bs{\Omega} & \mathbf 0 \\
      \mathbf 0   & \mathbf 0\end{array} \right] \boldsymbol \psi.
\]
We use this formula to construct bias-corrected estimators and confidence intervals (CIs) for $\bs \gamma$ and $\bs \alpha$. Given consistent estimators $\hat \kappa$ and $\hat{\bs \Omega}$ of $\kappa$ and $\bs \Omega$, one can construct the following bias-corrected estimators
\begin{align*}
  \hat{\boldsymbol{\psi}}^{bca} & = \left( \mathbf I + \frac{\hat \kappa}{\sqrt n} \left( \frac 1n \sum_{i=1}^n \hat{\boldsymbol{\xi}}_i^{\phantom{T}}\!\!\!\;\hat{\boldsymbol{\xi}}_i^T \right)^{-1} \left[ \begin{array}{cc}
                                                                                                                                                                                                                   \hat{\boldsymbol{\Omega}} & \mathbf 0 \\
                                                                                                                                                                                                                   \mathbf 0                 & \mathbf 0\end{array} \right] \right) \hat{\boldsymbol{\psi}} ,     \\[4pt]
  \hat{\boldsymbol{\psi}}^{bcm} & = \left( \mathbf I - \frac{\hat \kappa}{\sqrt n}\left( \frac 1n \sum_{i=1}^n \hat{\boldsymbol{\xi}}_i^{\phantom{T}}\!\!\!\;\hat{\boldsymbol{\xi}}_i^T \right)^{-1} \left[ \begin{array}{cc}
                                                                                                                                                                                                                \hat{\boldsymbol{\Omega}} & \mathbf 0 \\
                                                                                                                                                                                                                \mathbf 0                 & \mathbf 0\end{array} \right] \right)^{-1} \hat{\boldsymbol{\psi}} .
\end{align*}
The first estimator $\hat{\boldsymbol{\psi}}^{bca} $ performs an additive bias correction to the OLS estimator $\hat{\boldsymbol \psi}$. Simulations below show that this estimator performs well when the bias in $\hat{\boldsymbol \psi}$ is relatively small. The additive correction may not be sufficient when the bias in $\hat{\boldsymbol \psi}$ is large, as it relies upon $\hat{\boldsymbol \psi}$ being a reasonable estimator of $\boldsymbol \psi$. The bias corrected estimator $\hat{\boldsymbol{\psi}}^{bcm}$ performs a more aggressive (multiplicative) bias correction. We recommend this second estimator when the bias in $\hat{\boldsymbol \psi}$ is expected to be large, provided the maximum eigenvalue of
\[
  \frac{\hat \kappa}{\sqrt n}\left( \frac 1n \sum_{i=1}^n \hat{\boldsymbol{\xi}}_i^{\phantom{T}}\!\!\!\;\hat{\boldsymbol{\xi}}_i^T \right)^{-1} \left[ \begin{array}{cc}
      \hat{\boldsymbol{\Omega}} & \mathbf 0 \\
      \mathbf 0                 & \mathbf 0\end{array} \right]
\]
is less than one in absolute value.\footnote{This condition ensures invertibility of
  $
    \left( \mathbf I - \frac{\hat \kappa}{\sqrt n}\left( \frac 1n \sum_{i=1}^n \hat{\boldsymbol{\xi}}_i^{\phantom{T}}\!\!\!\;\hat{\boldsymbol{\xi}}_i^T \right)^{-1} \left[ \begin{array}{cc}
          \hat{\boldsymbol{\Omega}} & \mathbf 0 \\
          \mathbf 0                 & \mathbf 0\end{array} \right] \right).
  $
  }
We discuss how to consistently estimate $\kappa$ and $\bs\Omega$ within the context of our running examples below.

Bias corrected CIs for the regression coefficients are constructed by centering at the bias-corrected estimators and using OLS standard errors. A valid $100(1-a)$\% CI for the $j$th component $\psi_j$ of $\bs \psi$ is given by
\[
  \mr{CI}_n(\psi_j) = \left[ \hat \psi_j^{bc} - z_{1-a/2} \frac{\hat \sigma_j}{\sqrt n} , \hat \psi_j^{bc} + z_{1-a/2} \frac{\hat \sigma_j}{\sqrt n} \right],
\]
where $\hat \psi_j^{bc}$ denotes the $j$th entry of $\hat{\bs \psi}^{bca}$ or $\hat{\bs \psi}^{bcm}$, $z_{1-a/2}$ is the $1-a/2$ quantile of the normal distribution (e.g., 1.96 for a 95\% CI), and $\hat \sigma_j$ denotes the square root of the $j$th diagonal entry of $\hat{\mathbf V}$ from (\ref{eq:vhat}).

The following result shows that the bias-corrected estimators are asymptotically normal with the correct centering, and the bias-corrected confidence intervals 
are valid:

\begin{theorem}[Validity of Bias-Corrected Inference]\label{theorem:CI}
  Suppose that Assumption~\ref{assumption:drifting.general} holds, that $\E[\varepsilon_i^2 \bs{\xi}_i^{\phantom{T}}\!\!\!\;\bs{\xi}_i^T]$ has full rank, and that $\hat \kappa \to_p \kappa$ and $\hat{\boldsymbol \Omega} \to_p \boldsymbol \Omega$. Then as $n \to \infty$:
  \begin{enumerate}
    \item Bias-corrected estimators are first-order asymptotically equivalent and asymptotically normally distributed with the correct centering:
          \[
            \sqrt n \left( \hat{\bs{\psi}}^{bcm} - \bs{\psi} \right) = \sqrt n \left( \hat{\bs{\psi}}^{bca} - \bs{\psi} \right) + o_p(1) \to_d N (  \bs 0 \,, \mf V );
          \]
    \item Bias-corrected confidence intervals have correct coverage:
          \[
            \lim_{n \to \infty}\Pr( \psi_j \in \mr{CI}_n(\psi_j) ) = 1-a.
          \]
  \end{enumerate}
\end{theorem}

We now show how to apply bias correction in the context of our two running examples. In each case, we propose consistent estimators $\hat \kappa$ and $\hat{\bs\Omega}$ of $\kappa$ and $\bs\Omega$. Theorem~\ref{theorem:CI} then implies that the bias-corrected CIs have correct coverage. We also recommend reporting $\hat \kappa$ as a diagnostic for assessing the relative importance of measurement error and sampling error.

\subsubsection{Application 1: AI/ML-Generated Labels.}

Here $\kappa = \lim_{n \to \infty} \sqrt n \, \E[ \hat \theta_i(1-\theta_i)]$ and $\bs \Omega = 1$. 
One may estimate $\kappa$ by taking a sample of observations of size $m \ll n$. Then, each observation $i = 1,\ldots,m$ for which $\hat \theta_i = 1$ is inspected, and assigned a true label $\theta_i$. The estimator of $\kappa$ is
\[
  \hat \kappa = \sqrt n \widehat{FPR}, \quad \quad \widehat{FPR} = \frac 1m \sum_{i=1}^m \hat \theta_i (1-\theta_i).
\]
We establish validity of this approach allowing $m/n \to 0$ asymptotically. This is important for accommodating modern use cases where ML/AI methods are deployed to impute labels on massive data sets (large $n$), but where correctly labeling data can be costly (small $m$). For instance, \cite{boxellJournalistIdeologyProduction2022} impute binary labels representing political slant for a corpus of millions of newspaper articles using a validation sample size in the tens of thousands. Other approaches recently advocated in the literature for correcting measurement error (e.g., \cite{fongMachineLearningPredictions2021,allonMachineLearningPrediction2023,egamiUsingImperfectSurrogates2023}) are shown to be valid when $m/n \to c > 0$, so that the validation and original sample sizes are comparable. As far as we are aware, the theoretical properties of these proposed methods are unknown in modern scenarios where $m/n \to 0$.

We emphasize that our approach does not require constructing a full validation sample of size $m$ (and thus performing costly inspection of all $m$ observations), but only those for which $\hat \theta_i = 1$. This can greatly reduce the burden on the researcher. For instance, in our empirical application to remote work, we take a subsample of size $m = 1000$. Of these, only 26 observations have $\hat \theta_i = 1$, so only 26 job postings need to be inspected. By contrast, constructing a validation sample would require inspecting all $m = 1000$ observations.

Finally, our bias correction can be implemented using external data in which $Y_i$ and/or $\mf q_i$ are missing, since it only requires a subsample of $\theta_i$ and $\hat \theta_i$. This makes our approach more broadly applicable than other methods that require a full validation data set. For instance, \cite{bursztynImmigrantNextDoor2024} use a ML algorithm to classify a data set of charitable donors' names by ethnicity. As true ethnicity is latent, they estimate the accuracy of the classifier using an external sample of North Carolina voter registration data which contains self-reported ethnicity (but is missing data on charitable donations). 

The following result shows that $\hat \kappa$ is consistent, allowing $m/n \to 0$.
\begin{lemma}\label{lem:kappa.ex1}
  Suppose that $\sqrt n \, \E[\hat \theta_i(1-\theta_i)] \to \kappa > 0$ and $n/m^2 \to 0$. Then $\hat \kappa \to_p \kappa$.
\end{lemma}

Consistency of $\hat\kappa$ suffices for asymptotic validity of the bias-corrected CIs. However, the estimation of $\hat \kappa$ from a small subsample can introduce additional variability that, while asymptotically negligible, may be important to account for in finite samples. To this end, we introduce the following finite-sample correction to standard errors. Recall $\hat{\mf V}$ from (\ref{eq:vhat}). Also let
\[
  \hat{\bs \Gamma} =  \left( \frac 1n \sum_{i=1}^n \hat{\bs{\xi}}_i^{\phantom{T}}\!\!\!\;\hat{\bs{\xi}}_i^T \right)^{-1}
  \left( \begin{array}{cc} 1 & \mf 0 \\ \mf 0 & \mf 0 \end{array} \right).
\]
The adjusted covariance matrix estimators $\hat{\mf V}^{bca}$ and $\hat{\mf V}^{bcm}$ for $\hat{\bs \psi}^{bca}$ and $\hat{\bs \psi}^{bcm}$ are given by
\[
  \begin{aligned}
    \hat{\mf V}^{bca} & = (\mf I + \widehat{FPR} \, \hat{\bs \Gamma}) \hat{\mf V} (\mf I + \widehat{FPR} \, \hat{\bs \Gamma}^T) + \frac 1m \widehat{FPR}(1-\widehat{FPR}) \hat{\bs \Gamma} \left( \hat{\mf V} + n \hat{\bs \psi} \hat{\bs \psi}^T \right) \hat{\bs \Gamma}^T  ,          \\
    \hat{\mf V}^{bcm} & = (\mf I - \widehat{FPR} \, \hat{\bs \Gamma})^{-1} \hat{\mf V} (\mf I - \widehat{FPR} \, \hat{\bs \Gamma}^T)^{-1} + \frac 1m \widehat{FPR}(1-\widehat{FPR}) \hat{\bs \Gamma} \left( \hat{\mf V} + n \hat{\bs \psi} \hat{\bs \psi}^T \right) \hat{\bs \Gamma}^T , 
  \end{aligned}
\]
respectively.
The form of these adjustments follows from the law of total variance. Under the conditions of Theorem~\ref{theorem:drifting.general} and Lemma~\ref{lem:kappa.ex1}, we have $\hat{\mf V}^{bca} \to_p \mf V$ and $\hat{\mf V}^{bcm} \to_p \mf V$, with $\mf V$ the asymptotic variance derived in Theorem~\ref{theorem:drifting.general}. In practice, we recommend reporting standard errors computed from $\hat{\mf V}^{bca}$ if using $\hat{\bs \psi}^{bca}$ or $\hat{\mf V}^{bcm}$ if using $\hat{\bs \psi}^{bcm}$.

Codes to implement these bias corrections and standard error formulas are available in the Python package \texttt{ValidMLInference}.

\subsubsection{Application 2: Topic Models.}

In view of the expressions for $\kappa$ and $\bs \Omega$ derived in Theorem~\ref{theorem:two-step.drifting}, consider
\[
  \begin{aligned}
    \hat \kappa      & = \frac{1}{\sqrt n} \sum_{i=1}^n C_i^{-1} ,                                                                                                                                                                                                                           &
    \hat{\bs \Omega} & = \boldsymbol S (\hat{\mf{B}} \hat{\mf{B}}^T)^{-1} \hat{\mf{B}}\, \mr{diag}(\hat{\mf{B}}^T \bar{\bs w}_n) \hat{\mf{B}}^{T} (\hat{\mf{B}} \hat{\mf{B}}^T)^{-1} \boldsymbol S^T - \frac 1n \sum_{i=1}^n \hat{\bs{\theta}}_i^{\phantom{T}}\!\!\!\;\hat{\bs{\theta}}_i^T,
  \end{aligned}
\]
where $\bar{\bs w}_n = \frac 1n \sum_{i=1}^n \hat{\bs w}_i$. The following result shows that these estimators are consistent.

\begin{lemma}\label{lem:kappa.ex2}

  Suppose that Assumption~\ref{assumption:drifting} holds and that $\bar{\bs w}_n \to_p \E[\bs{w}_i]$. Then $\hat \kappa \to_p \kappa$ and $\hat{\bs \Omega} \to_p \bs \Omega$.

\end{lemma}

\subsection{Joint Estimation \label{sec:one-step}}

Our second approach for correcting the bias from the two-step strategy begins by formulating a likelihood $l\left(Y_i, h(\mathbf{x}_i), \bs \theta_i \mid \mathbf{q}_i, \mathbf{v}_i, \bs\gamma, \bs\alpha, \bs\zeta\right)$. Here $h$ is a function of the high-dimensional or unstructured observables, $\mathbf{v}_i$ are covariates that potentially enter the model beyond the downstream regression, and $\bs\zeta$ are nuisance parameters.  We discuss how to form a likelihood below. Integrating the latent $\bs \theta_i$ out yields a likelihood $l\left(Y_i, h(\mathbf{x}_i) \mid \mathbf{q}_i, \mathbf{v}_i, \bs\gamma, \bs\alpha, \bs\zeta\right)$ depending only on observables, which can then be used for maximum likelihood estimation of model parameters.  This idea can be applied generically, and here we adapt it to each of the three applications in Section~\ref{sec:applications}. 

\subsubsection{Remote Work and Wage Inequality}

In the remote work application, $Y_i$ is log posted wages, $\theta_i \in \{0,1\}$ is an indicator for whether job posting $i$ offers remote work ($\theta_i = 1$) or not ($\theta_i = 0$), and $\mathbf{x}_i$ is job posting text.  We use $h(\mathbf{x}_i) = \hat\theta_i \in \{0,1\}$ and thereby formulate a likelihood over the predicted class label associated with $\mathbf{x}_i$ rather than over $\mathbf{x}_i$ directly.  No additional covariates enter the model so $\mathbf{v}_i$ is empty.  The integrated likelihood is
\begin{equation}
  \label{eqn:one_step_remote}
l(Y_i ,\hat \theta_i = d;(\gamma, \bs\alpha, \bs \zeta)) = \omega_{d1} MN(Y_i - \gamma - \bs\alpha^T \mathbf{q}_i; \bs \lambda_1) + \omega_{d0} MN(Y_i - \bs\alpha^T \mathbf{q}_i; \bs \lambda_0) ,
\end{equation}
for $d \in \{0,1\}$,  $\bs \omega = (\omega_{00}, \omega_{10}, \omega_{01}, \omega_{11}) \in \Delta^3$ with $\omega_{ab} = \Pr(\hat \theta_i = a, \theta_i = b)$, and $MN(\,\cdot\,;\bs \lambda) = \sum_{l=1}^L \lambda_{l} \phi(\,\cdot\,;\lambda_{\mu l},\lambda_{\sigma^2 l})$ is a mixture of normal distributions with $L$ components and mixing weights $\lambda_{1},\ldots,\lambda_{L} \in \Delta^{L-1}$, where $\phi(\,\cdot\,;\mu,\sigma^2)$ denotes the $N(\mu,\sigma^2)$ density, and the component means are normalized so that $\sum_{l=1}^L \lambda_{l} \lambda_{\mu l} = 0$. Thus, here $\bs \zeta = (\bs \omega, \bs \lambda_0,\bs \lambda_1)$. We use $L = 3$ for the results in Section~\ref{sec:applications.1}.

\subsubsection{CEO Time Use and Firm Performance}

In the CEO time use application, $Y_i$ is log sales of firm $i$, $\theta_i \in [0,1]$ is a behavior index, and $\mathbf{x}_i$ is the vector of counts of time-use feature combinations across all surveyed 15-minute time intervals.  Here the likelihood is directly over $\mathbf{x}_i$, i.e. $h(\mathbf{x}_i) = \mathbf{x}_i$.  In the spirit of correlated random effects in panel data models, we specify a distribution for the behavior index $\theta_i$ conditional on $J$ covariates $\mathbf g_i$, which may include $\mathbf q_i$ or other variables.\footnote{\citet{robertsStructuralTopicModels2014} presents a model in which a logistic normal distribution over topic shares is parameterized by covariates but without a downstream regression.  \citet{bleiSupervisedTopicModels2010} and \citet{ahrensBayesianTopicRegression2021} present models in which linear combinations of topic shares explain a normally distributed response variable, but do not allow covariates to enter the distribution over topic shares.}  The likelihood is implied by 
\begin{equation}
\begin{aligned}
  \theta_i \mid (C_i,\mathbf q_i,\mathbf{g}_i) &\sim \text{LogisticNormal}\left(\bs\phi^T\mathbf{g}_i, \sigma_\theta^2 \right), \label{eqn:one_step_ceo} \\
  \mathbf x_{i} \mid (\theta_i,C_i,\mathbf q_i,\mathbf{g}_i) &\sim \text{Multinomial}\left(C_i, \mathbf{B}^T \mathbf{w}_i \right), \\
  Y_i \mid (\theta_i,C_i,\mathbf q_i,\mathbf{g}_i) &\sim \text{Normal}\left(\gamma\theta_i + \boldsymbol\alpha^T \mathbf{q}_i, \sigma_Y^2\right),
\end{aligned}
\end{equation}
where $\mathbf{w}_i = (1-\theta_i, \theta_i)^T$.  \citet{bandieraCEOBehaviorFirm2020} shows log employment and an indicator for whether the CEO has an MBA are correlated with behavior, so we include these in $\mathbf g_i$ together with a constant.  As the likelihood also depends on the number of observed time units $C_i$, we have $\mathbf{v}_i = (C_i, \mathbf g_i)$. 

\subsubsection{Central Bank Communication}

In this application, $Y_i$ is the path factor in FOMC meeting $i$, $\theta_i \in [0,1]$ is the hawkish sentiment of the FOMC written statement, and $\mathbf{x}_i$ is the text of the statement.  We take $h(\mathbf{x}_i) = (N_i, C_i)$ and form a likelihood over the number of paragraphs in FOMC statements classified as hawkish ($N_i$) conditional on the total number classified as hawkish or dovish ($C_i$). We treat $\mathbf{v}_i$ as empty, but our analysis could be extended to allow the prior for $\theta_i$ to depend on covariates.  The likelihood is implied by 
\begin{equation} \label{eqn:one_step_sentiment}
\begin{aligned}
  \theta_i  \mid (C_i,\mathbf{q}_i) &\sim U[0,1] , \\
  N_i \mid (\theta_i,C_i,\mathbf{q}_i) &\sim \text{Binomial}\left(C_i, (1-\theta_i)\beta_0 + \theta_i\beta_1 \right) , \\
  Y_i \mid (\theta_i,C_i,\mathbf{q}_i) &\sim \text{Normal}\left(\gamma\theta_i + \boldsymbol\alpha^T \mathbf{q}_i, \sigma_Y^2\right) .
\end{aligned}
\end{equation}
In addition, we include additional terms in the likelihood for the testing data:
\begin{equation}
\begin{aligned}
N_{10} &\sim \text{Binomial}(N_{10} + N_{00}, \beta_0) , & 
N_{11} &\sim \text{Binomial}(N_{11} + N_{01}, \beta_1) ,
\end{aligned}
\end{equation}
where $N_{ab}$ is the number of test-set observations classified an $a$ with true label $b$.\footnote{In the application, $N_{11} = 14$, $N_{00} = 54$, $N_{10} = 1$, $N_{01} = 1$.}

\subsection{Inference Approach for Intractable Likelihoods}

In certain cases, one can directly integrate-out $\bs \theta_i$ from the likelihood, as in \eqref{eqn:one_step_remote}.  However, in more complex problems, such as \eqref{eqn:one_step_ceo}, there are two challenges.  First, the integration for computing the likelihood has no closed-form solution and must be performed numerically. Moreover, this numerical integration must be done observation-by-observation. As such, standard likelihood-based estimation may not be computationally feasible.  

In such cases, one may use Bayesian computation to perform valid frequentist inference.  In this approach, we introduce a prior for the model parameters $\boldsymbol \delta = (\bs\gamma, \bs\alpha, \bs\zeta)$ and treat the latent $\boldsymbol{\theta}_i$ as ``parameters'' drawn from a prior distribution as illustrated above.  We sample from the posterior distribution of $\left(\bm\delta,(\boldsymbol \theta_i)_{i=1}^n\right)$ conditional on the observed data $(Y_i, h(\mathbf x_i), \mathbf q_i, \mathbf v_i)_{i=1}^n$. The marginal draws for $\boldsymbol \delta$ represent draws from the posterior distribution for $\boldsymbol \delta$ based on the integrated likelihood $l\left(Y_i, h(\mathbf{x}_i) \mid \mathbf{q}_i, \mathbf{v}_i, \bs\gamma, \bs\alpha, \bs\zeta\right)$. Thus,  $\bs{\theta}_i$ is implicitly integrated of the likelihood as part of the sampling procedure.

It is important to emphasize that while our suggested approach uses Bayesian computation, inference is frequentist. 
The maximum likelihood estimator $\hat{\boldsymbol{\delta}}$ of $\boldsymbol{\delta}$ is asymptotically normal under standard regularity conditions (e.g., Theorem 5.41 of \citealp{vaartAsymptoticStatistics1998}). By the Bernstein--von Mises Theorem (see Theorem 10.1 of \citealp{vaartAsymptoticStatistics1998} and discussion), the posterior mean $\bar{\boldsymbol{\delta}}$ of $\boldsymbol{\delta}$ is first-order asymptotically equivalent to the MLE $\hat{\boldsymbol{\delta}}$. Moreover, the posterior distribution of $\boldsymbol\delta$ is asymptotically normal with mean $\bar{\boldsymbol \delta}$ and variance (when appropriately scaled with $n$) equal to the asymptotic variance of the MLE. As such, Bayesian credible sets for elements of $\boldsymbol{\delta}$\textemdash or any of its components such as $\boldsymbol{\gamma}$ or $\boldsymbol{\alpha}$\textemdash are valid frequentist confidence sets. 
This approach is also \emph{efficient} for inference on $\boldsymbol{\delta}$ and its components, as it is asymptotically equivalent to likelihood-based inference.

We implement this approach for the CEO time use and central bank communication applications.  For the former, we follow standard practice in the topic modeling literature and use $\text{Dirichlet}(0.2)$ priors for $\bm\beta_1$ and $\bm\beta_2$.  For the logistic normal prior for $\boldsymbol{\theta}_i$, we set $\sigma_\theta^2 = 1$ and use a $\text{Normal}(0,4)$ prior for each element of $\bs\phi$.  We also use $\text{Normal}(0, 100)$ priors for all downstream regression coefficients, and a $\text{Gamma}(1,10)$ prior for $\sigma_Y$.  

For the central bank communication application, we use the same priors in the downstream regression.  We use a $\text{Beta}(2, 5)$ prior for $\beta_0$ and a $\text{Beta}(5, 2)$ prior for $\beta_1$.  We adopt asymmetric priors on the classification probabilities to resolve the label switching problem inherent in topic models.

We perform sampling with Hamiltonian Monte Carlo implemented in \texttt{NumPyro} \citep{phanComposableEffectsFlexible2019}.  See \citet{sacherHamiltonianMonteCarlo2024} for more details on HMC and probabilistic programming.


\section{Simulation Evidence} \label{sec:simulations}

In this section, we provide simulation evidence illustrating the finite-sample performance of our methods for correcting bias and performing valid inference.

\subsection{AI/ML-Generated Labels}
\label{sec:simulations.app1}

This simulation is calibrated to the remote work empirical application in Section~\ref{sec:applications.1}. The same contains 16,315 job postings. Bias-corrected estimates of the intercept and slope are approximately 10 and 1. The residual standard deviation is approximately 0.3 (respectively, 0.5) for observations with $\hat \theta_i = 0$ ($\hat \theta_i = 1$). We therefore generate data according to
\[
 Y_i = 10 + \theta_i + (0.3 + 0.2 \theta_i) \varepsilon_i,
\]
with $\varepsilon_i \sim N(0,1)$. The sample mean of $\hat \theta_i$ is 0.025 and  $\widehat{FPR} = 0.009$, which corresponds to $\hat \kappa \approx 1.1$. We draw samples of size $n =$ 8,000, 16,000, and 32,000, with $\theta_i \sim \text{Bernoulli}(p)$ for $p =$ 0.025, 0.05, and 0.5, and $\kappa =$ 0.5, 1, and 2, with 1000 simulations for each configuration.

We implement the two-step strategy, additive and multiplicative bias corrections, and joint estimation. For the latter, we use the likelihood from (\ref{eqn:one_step_remote}) with a single Gaussian component ($L = 1$).
To implement both bias corrections, we generate a sample of $(\theta_i,\hat \theta_i)$ of size $m = 1000$. We use two estimators of the false-positive rate. The first is the empirical frequency $\widehat{FPR} = \frac 1m \sum_{i=1}^m \hat \theta_i(1-\theta_i)$. We may interpret this as the posterior mean of $r := \E[\hat \theta_i(1-\theta_i)]$ under an improper $r^{-1}(1-r)^{-1}$ prior. This prior puts most of its mass at the endpoints of the interval $[0,1]$. However, our approach is based on the premise that $r$ should be small. We therefore consider a Bayes estimator $\widehat{FPR}_B = \frac{\sum_{i=1}^m \hat \theta_i(1-\theta_i) + \frac 12}{m + \frac{5}{2}}$, which is the posterior mean of $r$ under a proper $r^{-1/2}(1-r)$ prior. Results are presented in Tables~\ref{tab:wfhsim.1}-\ref{tab:wfhsim.3}.

\begin{table}[t!]
\begin{center}
\caption{\label{tab:wfhsim.1}Simulation Results: Generated Labels, $p = 0.025$}
{\small
\begin{tabular}{lcrrrcrrrcrrr}
\\[-8pt]
\hline \hline \\[-10pt]
 & & \multicolumn{3}{c}{Bias} & & \multicolumn{3}{c}{RMdSE} & & \multicolumn{3}{c}{Coverage} \\ \cline{3-5} \cline{7-9} \cline{11-13} \\[-8pt]
$\kappa$ & & \multicolumn{1}{c}{0.5} & \multicolumn{1}{c}{1} & \multicolumn{1}{c}{2} & 
		   & \multicolumn{1}{c}{0.5} & \multicolumn{1}{c}{1} & \multicolumn{1}{c}{2} & 
		   & \multicolumn{1}{c}{0.5} & \multicolumn{1}{c}{1} & \multicolumn{1}{c}{2} \\[2pt]
 & & \multicolumn{11}{c}{$n = 8000$} \\ \cline{3-13} \\[-10pt]
2-Step & & -0.228 & -0.459 & -0.916 & & 0.228 & 0.459 & 0.916 & & 0.001 & 0.000 & 0.000 \\
BCA-1 & & -0.055 & -0.214 & -0.843 & & 0.076 & 0.214 & 0.843 & & 0.899 & 0.589 & 0.000 \\
BCA-2 & & -0.039 & -0.203 & -0.841 & & 0.073 & 0.204 & 0.841 & & 0.929 & 0.643 & 0.000 \\
BCM-1 & & -0.003 & -0.006 & -0.715 & & 0.091 & 0.170 & 0.806 & & 0.887 & 0.758 & 0.265 \\
BCM-2 & & 0.023 & 0.030 & -0.729 & & 0.090 & 0.177 & 0.827 & & 0.904 & 0.773 & 0.257 \\
Joint & & 0.000 & 0.002 & -0.008 & & 0.045 & 0.056 & 0.111 & & 0.935 & 0.930 & 0.817 \\
 \\
 & & \multicolumn{11}{c}{$n = 16000$} \\ \cline{3-13} \\[-10pt]
2-Step & & -0.161 & -0.323 & -0.648 & & 0.161 & 0.323 & 0.648 & & 0.000 & 0.000 & 0.000 \\
BCA-1 & & -0.028 & -0.110 & -0.423 & & 0.058 & 0.111 & 0.423 & & 0.884 & 0.794 & 0.049 \\
BCA-2 & & -0.011 & -0.097 & -0.417 & & 0.055 & 0.101 & 0.417 & & 0.928 & 0.841 & 0.058 \\
BCM-1 & & -0.001 & -0.005 & -0.012 & & 0.071 & 0.110 & 0.303 & & 0.874 & 0.832 & 0.542 \\
BCM-2 & & 0.024 & 0.024 & 0.037 & & 0.070 & 0.114 & 0.313 & & 0.915 & 0.851 & 0.548 \\
Joint & & 0.002 & 0.002 & 0.002 & & 0.030 & 0.034 & 0.055 & & 0.952 & 0.951 & 0.897 \\
 \\
 & & \multicolumn{11}{c}{$n = 32000$} \\ \cline{3-13} \\[-10pt]
2-Step & & -0.116 & -0.231 & -0.459 & & 0.116 & 0.231 & 0.459 & & 0.000 & 0.000 & 0.000 \\
BCA-1 & & -0.022 & -0.056 & -0.214 & & 0.048 & 0.069 & 0.214 & & 0.878 & 0.859 & 0.505 \\
BCA-2 & & -0.004 & -0.040 & -0.204 & & 0.044 & 0.062 & 0.204 & & 0.938 & 0.907 & 0.566 \\
BCM-1 & & -0.009 & -0.007 & -0.015 & & 0.055 & 0.085 & 0.165 & & 0.874 & 0.868 & 0.735 \\
BCM-2 & & 0.014 & 0.019 & 0.022 & & 0.054 & 0.085 & 0.168 & & 0.927 & 0.896 & 0.758 \\
Joint & & -0.002 & 0.000 & 0.001 & & 0.019 & 0.024 & 0.027 & & 0.950 & 0.952 & 0.957 \\[2pt] \hline \hline
\end{tabular}
}
\vskip 2pt
\parbox{0.95\textwidth}
{
\footnotesize
\emph{Note:} Median bias (Bias), root median square error (RMdSE), and coverage of 95\% confidence intervals (Coverage) across different $\kappa$ and $n$. Results are presented for the two-step strategy (2-step), additive bias correction using $\widehat{FPR}$ (BCA-1) and $\widehat{FPR}_B$ (BCA-2), multiplicative bias correction using $\widehat{FPR}$ (BCM-1) and $\widehat{FPR}_B$ (BCM-2), and joint estimation (Joint).
}
\end{center}
\end{table}

\begin{table}[t!]
\begin{center}
\caption{\label{tab:wfhsim.2}Simulation Results: Generated Labels, $p = 0.05$}
{\small
\begin{tabular}{lcrrrcrrrcrrr}
\\[-8pt]
\hline \hline \\[-10pt]
 & & \multicolumn{3}{c}{Bias} & & \multicolumn{3}{c}{RMdSE} & & \multicolumn{3}{c}{Coverage} \\ \cline{3-5} \cline{7-9} \cline{11-13} \\[-8pt]
$\kappa$ & & \multicolumn{1}{c}{0.5} & \multicolumn{1}{c}{1} & \multicolumn{1}{c}{2} & 
		   & \multicolumn{1}{c}{0.5} & \multicolumn{1}{c}{1} & \multicolumn{1}{c}{2} & 
		   & \multicolumn{1}{c}{0.5} & \multicolumn{1}{c}{1} & \multicolumn{1}{c}{2} \\[2pt]
 & & \multicolumn{11}{c}{$n = 8000$} \\ \cline{3-13} \\[-10pt]
2-Step & & -0.116 & -0.236 & -0.469 & & 0.116 & 0.236 & 0.469 & & 0.019 & 0.000 & 0.000 \\
BCA-1 & & -0.015 & -0.058 & -0.224 & & 0.039 & 0.065 & 0.224 & & 0.931 & 0.868 & 0.319 \\
BCA-2 & & -0.005 & -0.050 & -0.219 & & 0.037 & 0.060 & 0.219 & & 0.956 & 0.893 & 0.335 \\
BCM-1 & & -0.002 & -0.005 & -0.010 & & 0.043 & 0.063 & 0.131 & & 0.927 & 0.884 & 0.751 \\
BCM-2 & & 0.010 & 0.008 & 0.009 & & 0.041 & 0.064 & 0.130 & & 0.944 & 0.899 & 0.760 \\
Joint & & 0.000 & -0.003 & -0.004 & & 0.028 & 0.031 & 0.043 & & 0.956 & 0.951 & 0.927 \\
 \\
 & & \multicolumn{11}{c}{$n = 16000$} \\ \cline{3-13} \\[-10pt]
2-Step & & -0.081 & -0.165 & -0.332 & & 0.081 & 0.165 & 0.332 & & 0.018 & 0.000 & 0.000 \\
BCA-1 & & -0.007 & -0.031 & -0.110 & & 0.031 & 0.044 & 0.110 & & 0.916 & 0.893 & 0.690 \\
BCA-2 & & 0.003 & -0.022 & -0.103 & & 0.029 & 0.040 & 0.104 & & 0.948 & 0.917 & 0.719 \\
BCM-1 & & -0.001 & -0.004 & 0.000 & & 0.033 & 0.048 & 0.090 & & 0.910 & 0.886 & 0.817 \\
BCM-2 & & 0.011 & 0.008 & 0.014 & & 0.034 & 0.049 & 0.088 & & 0.938 & 0.915 & 0.823 \\
Joint & & 0.000 & 0.001 & 0.002 & & 0.018 & 0.021 & 0.023 & & 0.948 & 0.953 & 0.935 \\
 \\
 & & \multicolumn{11}{c}{$n = 32000$} \\ \cline{3-13} \\[-10pt]
2-Step & & -0.059 & -0.118 & -0.236 & & 0.059 & 0.118 & 0.236 & & 0.006 & 0.000 & 0.000 \\
BCA-1 & & -0.009 & -0.015 & -0.060 & & 0.025 & 0.034 & 0.063 & & 0.897 & 0.904 & 0.835 \\
BCA-2 & & 0.001 & -0.006 & -0.052 & & 0.024 & 0.032 & 0.056 & & 0.949 & 0.935 & 0.864 \\
BCM-1 & & -0.006 & -0.004 & -0.005 & & 0.027 & 0.041 & 0.065 & & 0.892 & 0.902 & 0.871 \\
BCM-2 & & 0.005 & 0.008 & 0.008 & & 0.027 & 0.039 & 0.062 & & 0.938 & 0.929 & 0.880 \\
Joint & & -0.001 & -0.001 & 0.000 & & 0.013 & 0.015 & 0.016 & & 0.957 & 0.949 & 0.961 \\[2pt] \hline \hline
\end{tabular}
}
\vskip 2pt
\parbox{0.95\textwidth}
{
\footnotesize
\emph{Note:} Median bias (Bias), root median square error (RMdSE), and coverage of 95\% confidence intervals (Coverage) across different $\kappa$ and $n$. Results are presented for the two-step strategy (2-step), additive bias correction using $\widehat{FPR}$ (BCA-1) and $\widehat{FPR}_B$ (BCA-2), multiplicative bias correction using $\widehat{FPR}$ (BCM-1) and $\widehat{FPR}_B$ (BCM-2), and joint estimation (Joint).
}
\end{center}
\end{table}

\begin{table}[t!]
\begin{center}
\caption{\label{tab:wfhsim.3}Simulation Results: Generated Labels, $p = 0.5$}
{\small
\begin{tabular}{lcrrrcrrrcrrr}
\\[-8pt]
\hline \hline \\[-10pt]
 & & \multicolumn{3}{c}{Bias} & & \multicolumn{3}{c}{RMdSE} & & \multicolumn{3}{c}{Coverage} \\ \cline{3-5} \cline{7-9} \cline{11-13} \\[-8pt]
$\kappa$ & & \multicolumn{1}{c}{0.5} & \multicolumn{1}{c}{1} & \multicolumn{1}{c}{2} & 
		   & \multicolumn{1}{c}{0.5} & \multicolumn{1}{c}{1} & \multicolumn{1}{c}{2} & 
		   & \multicolumn{1}{c}{0.5} & \multicolumn{1}{c}{1} & \multicolumn{1}{c}{2} \\[2pt]
 & & \multicolumn{11}{c}{$n = 8000$} \\ \cline{3-13} \\[-10pt]
2-Step & & -0.022 & -0.045 & -0.089 & & 0.022 & 0.045 & 0.089 & & 0.351 & 0.002 & 0.000 \\
BCA-1 & & 0.000 & -0.002 & -0.008 & & 0.009 & 0.011 & 0.014 & & 0.944 & 0.943 & 0.925 \\
BCA-2 & & 0.002 & 0.000 & -0.007 & & 0.009 & 0.010 & 0.014 & & 0.957 & 0.956 & 0.937 \\
BCM-1 & & 0.000 & 0.000 & 0.000 & & 0.010 & 0.011 & 0.016 & & 0.942 & 0.941 & 0.925 \\
BCM-2 & & 0.002 & 0.002 & 0.002 & & 0.009 & 0.012 & 0.016 & & 0.954 & 0.949 & 0.932 \\
Joint & & 0.000 & 0.000 & 0.000 & & 0.008 & 0.008 & 0.009 & & 0.944 & 0.956 & 0.954 \\
 \\
 & & \multicolumn{11}{c}{$n = 16000$} \\ \cline{3-13} \\[-10pt]
2-Step & & -0.015 & -0.032 & -0.063 & & 0.015 & 0.032 & 0.063 & & 0.344 & 0.002 & 0.000 \\
BCA-1 & & 0.000 & -0.001 & -0.004 & & 0.007 & 0.009 & 0.012 & & 0.933 & 0.950 & 0.930 \\
BCA-2 & & 0.002 & 0.000 & -0.003 & & 0.007 & 0.009 & 0.012 & & 0.956 & 0.955 & 0.950 \\
BCM-1 & & 0.000 & 0.000 & 0.000 & & 0.007 & 0.009 & 0.013 & & 0.931 & 0.944 & 0.935 \\
BCM-2 & & 0.002 & 0.002 & 0.002 & & 0.007 & 0.009 & 0.012 & & 0.950 & 0.949 & 0.945 \\
Joint & & 0.000 & 0.000 & 0.000 & & 0.005 & 0.005 & 0.006 & & 0.947 & 0.944 & 0.960 \\
 \\
 & & \multicolumn{11}{c}{$n = 32000$} \\ \cline{3-13} \\[-10pt]
2-Step & & -0.011 & -0.023 & -0.045 & & 0.011 & 0.023 & 0.045 & & 0.322 & 0.005 & 0.000 \\
BCA-1 & & -0.001 & -0.001 & -0.002 & & 0.005 & 0.007 & 0.009 & & 0.919 & 0.920 & 0.927 \\
BCA-2 & & 0.001 & 0.001 & 0.000 & & 0.005 & 0.007 & 0.009 & & 0.955 & 0.948 & 0.938 \\
BCM-1 & & -0.001 & 0.000 & 0.000 & & 0.005 & 0.008 & 0.010 & & 0.917 & 0.918 & 0.921 \\
BCM-2 & & 0.001 & 0.002 & 0.002 & & 0.005 & 0.007 & 0.010 & & 0.951 & 0.943 & 0.931 \\
Joint & & 0.000 & 0.000 & 0.000 & & 0.004 & 0.005 & 0.004 & & 0.943 & 0.932 & 0.943 \\[2pt] \hline \hline
\end{tabular}
}
\vskip 2pt
\parbox{0.95\textwidth}
{
\footnotesize
\emph{Note:} Median bias (Bias), root median square error (RMdSE), and coverage of 95\% confidence intervals (Coverage) across different $\kappa$ and $n$. Results are presented for the two-step strategy (2-step), additive bias correction using $\widehat{FPR}$ (BCA-1) and $\widehat{FPR}_B$ (BCA-2), multiplicative bias correction using $\widehat{FPR}$ (BCM-1) and $\widehat{FPR}_B$ (BCM-2), and joint estimation (Joint).
}
\end{center}
\end{table}

Our baseline set of simulations with $p = 0.025$ (Table~\ref{tab:wfhsim.1}) is very challenging: in a sample of size 8,000, we expect only 200 observations with $\theta_i = 1$. When $\kappa = 1$ about half of these will be incorrectly imputed with $\hat \theta_i = 0$, and when $\kappa = 2$ almost all will be incorrectly imputed. In either case, we expect a very large bias of two-step estimators, and the results in Table~\ref{tab:wfhsim.1} confirm this, with large bias and zero coverage. Our bias corrections rely on the two-step estimator having a bias that is of the same order as sampling uncertainty. One might therefore expect they will perform poorly in small samples when two-step estimators are severely biased. Table~\ref{tab:wfhsim.1} shows this is not necessarily the case. Consider the configuration with $n =$ 16,000 and $\kappa = 1$, which is closest to the empirical application. The two-step estimator under-estimates the coefficient by about 32\%, the additive bias corrections by about 10\%, while the multiplicative bias corrected estimators are approximately unbiased, even in this challenging design. Indeed, the multiplicative bias correction appears to perform well across all configurations except when $n$ is smallest and $\kappa$ is largest. 

Bias corrected confidence sets appear to under-cover by about 10\% with $n =$ 16,000 and $\kappa = 1$, but the results for $n =$ 32,000 show their coverage moves toward nominal coverage as $n$ increases. Coverage is much closer to nominal coverage in the less challenging designs with $p = 0.05$ (Table~\ref{tab:wfhsim.2}) and $p = 0.5$ (Table~\ref{tab:wfhsim.3}), even with small $n$. Coverage also seems to be improved using the Bayes estimator $\widehat{FPR}_B$ instead of the empirical frequency. 

Finally, joint estimation produces approximately unbiased estimators across all designs, including the most challenging cases. It also produces confidence sets with coverage closest to nominal coverage and, as expected, the lowest risk of the competing approaches. Overall, these results demonstrate the sound performance of both our proposed solutions in a challenging, empirically calibrated setting.

\subsection{Topic Models and AI/ML-Generated Indices}
\label{sec:simulations.app2}

These simulations are calibrated to the empirical application to AI/ML-generated indices in Section~\ref{sec:applications.1}. We draw a latent share $\theta_i \sim U[0,1]$ then generate $N_i \sim \text{Binomial}(C_i,\beta_1 \theta_i + \beta_0 (1-\theta_i))$ with $\beta_1 = 0.9 $ and $\beta_0 = 0.1$ to introduce misclassification error. We set \[Y_i = -0.05 + 0.11 \theta_i + 0.1 \varepsilon_i,\] where $\varepsilon_i \sim N(0,1)$ and parameters are similar to the estimates in the empirical application. We generate samples of $(Y_i,N_i,C_i)$ of size $n = $ 200 (as in the application), 800, and 3,200. For $n = 200$ we use $C_i$ from the empirical application so that $\kappa \approx 4.57$. We then increase $C_i$ by a factors of two and four, to generate samples with $\kappa \approx 2.28$ and $\kappa \approx 1.14$, respectively. For the larger sample sizes we replicate the empirical $C_i$ and multiply by factors of two and four so $\kappa$ remains constant. We generate 1000 samples for each configuration.

We implement the two-step strategy, both bias corrections, and joint estimation. For the latter, we use the likelihood described in \eqref{eqn:one_step_sentiment}. For the two-step strategy, we require a suitable estimate of $\theta_i$. We produce estimates based on the topic model representation from Section~\ref{sec:applications} in order to account for potential misclassification error. To provide a clear comparison with joint estimation, we use the estimates of $\beta_0$ and $\beta_1$ from joint estimation, then estimate $\theta_i$ in accordance with the discussion in Section~\ref{sec:topic.theory.main}, using
\[
 \hat \theta_i = \max(0, \min( 1, \bs S (\hat{\mf B}^T)^{-1} (\mf n_i/C_i))),
\]
where
\[
 \bs S = \left[ \begin{array}{cc} 1 & 0 \end{array} \right], \quad
 \hat{\mf B} = \left[ \begin{array}{cc} \hat \beta_1 & \hat \beta_0 \\ 1 - \hat \beta_1 & 1 - \hat \beta_0 \end{array} \right] , \quad 
 \mf n_i = \left[ \begin{array}{c} N_i \\ C_i - N_i \end{array} \right].
\]
We also implement additive and multiplicative bias corrections using the formulas derived for topic models with the $\bs S$ and $\hat{\mf B}$ as above. Results are presented in Table~\ref{tab:indexsim}.

\begin{table}[t!]
\begin{center}
\caption{\label{tab:indexsim}Simulation Results: AI/ML-Generated Indices}
{\small
\begin{tabular}{lcrrrcrrrcrrr}
\\[-8pt]
\hline \hline \\[-10pt]
 & & \multicolumn{3}{c}{Bias} & & \multicolumn{3}{c}{RMdSE} & & \multicolumn{3}{c}{Coverage} \\ \cline{3-5} \cline{7-9} \cline{11-13} \\[-8pt]
$\kappa$ & & \multicolumn{1}{c}{4.57} & \multicolumn{1}{c}{2.28} & \multicolumn{1}{c}{1.14} & 
		   & \multicolumn{1}{c}{4.57} & \multicolumn{1}{c}{2.28} & \multicolumn{1}{c}{1.14} & 
		   & \multicolumn{1}{c}{4.57} & \multicolumn{1}{c}{2.28} & \multicolumn{1}{c}{1.14} \\[2pt]
 & & \multicolumn{11}{c}{$n = 200$} \\ \cline{3-13} \\[-10pt]
2-Step & & -0.446 & -0.299 & -0.180 & & 0.049 & 0.033 & 0.022 & & 0.309 & 0.677 & 0.839 \\
BCA & & -0.148 & -0.019 & 0.018 & & 0.026 & 0.021 & 0.019 & & 0.716 & 0.813 & 0.880 \\
BCM & & 0.240 & 0.183 & 0.081 & & 0.040 & 0.029 & 0.021 & & 0.495 & 0.678 & 0.844 \\
Joint & & -0.003 & 0.007 & 0.004 & & 0.024 & 0.020 & 0.018 & & 0.945 & 0.948 & 0.938 \\
2-Step-Share & & -0.433 & -0.218 & -0.037 & & 0.048 & 0.025 & 0.018 & & 0.378 & 0.824 & 0.931 \\
 \\
 & & \multicolumn{11}{c}{$n = 800$} \\ \cline{3-13} \\[-10pt]
2-Step & & -0.301 & -0.193 & -0.113 & & 0.033 & 0.021 & 0.013 & & 0.147 & 0.533 & 0.806 \\
BCA & & -0.015 & 0.002 & 0.005 & & 0.011 & 0.010 & 0.010 & & 0.823 & 0.880 & 0.911 \\
BCM & & 0.182 & 0.063 & 0.023 & & 0.021 & 0.011 & 0.010 & & 0.507 & 0.798 & 0.897 \\
Joint & & 0.004 & -0.006 & -0.006 & & 0.011 & 0.010 & 0.010 & & 0.956 & 0.950 & 0.950 \\
2-Step-Share & & -0.215 & -0.041 & 0.084 & & 0.024 & 0.01 & 0.012 & & 0.507 & 0.942 & 0.894 \\
 \\
 & & \multicolumn{11}{c}{$n = 3200$} \\ \cline{3-13} \\[-10pt]
2-Step & & -0.194 & -0.110 & -0.060 & & 0.021 & 0.012 & 0.007 & & 0.053 & 0.456 & 0.773 \\
BCA & & -0.001 & 0.007 & 0.005 & & 0.005 & 0.005 & 0.005 & & 0.859 & 0.896 & 0.917 \\
BCM & & 0.060 & 0.025 & 0.010 & & 0.007 & 0.005 & 0.005 & & 0.700 & 0.869 & 0.913 \\
Joint & & -0.005 & -0.002 & -0.003 & & 0.005 & 0.005 & 0.005 & & 0.942 & 0.941 & 0.943 \\
2-Step-Share & & -0.042 & 0.085 & 0.158 & & 0.006 & 0.009 & 0.017 & & 0.887 & 0.739 & 0.353 \\[2pt] \hline \hline
\end{tabular}
}
\vskip 2pt
\parbox{0.95\textwidth}
{
\footnotesize
\emph{Note:} Median relative bias (Bias), root median square error (RMdSE), and coverage of 95\% confidence intervals (Coverage) across different $\kappa$ and $n$. Results are presented for the two-step strategy (2-step), additive and multiplicative bias corrections (BCA) and (BCM), joint estimation (Joint), and the two-step strategy regressing onto the empirical share $\hat \theta_i = N_i/C_i$ (2-Step-Share).
}
\end{center}
\end{table}

First consider the bias results. In our baseline setting with $n = 200$ and $\kappa = 4.57$ as in the empirical application, the median relative bias of the two-step estimate of $\gamma$ is $-0.446$, showing that the two-step strategy under-estimates $\gamma$ by nearly half. By contrast, joint estimates are nearly unbiased. The additive bias correction improves upon the two-step strategy, producing an estimate about 85\% that of the true effect size when $\kappa = 4.57$, and approximately unbiased estimates for smaller $\kappa$. Conversely, the multiplicative correction tends to over-estimate the true effect size, suggesting it performs too aggressive a bias correction in small samples. In larger samples, the additive correction produces approximately unbiased estimates, and the performance of the multiplicative correction improves, especially for smaller values of $\kappa$. The second panel of Table~\ref{tab:indexsim} also suggests the additive bias correction is preferred from a root median square error perspective, producing values that are on par with joint estimation.

Turning to coverage, we see that the two-step CIs have coverage well below nominal coverage, especially for larger values of $n$ and $\kappa$. Conversely, joint estimation produces CIs with coverage close to nominal coverage even for the smallest sample size. Coverage of the additively bias-corrected CIs is significantly better than that of the multiplicatively bias-corrected CIs which, in turn, is better than two-step coverage. For large $n$ and small $\kappa$, CIs based on the additive bias correction have coverage close to nominal coverage.

\subsection{Lessons for the Empirical Applications}
\label{sec:simulations.lessons}

We conclude this section by revisiting some of the empirical applications in light of the simulation evidence presented above. First consider the application to remote work (Section~\ref{sec:applications.1}). The bias-corrected estimates reported in Section~\ref{sec:applications.1} are based on the multiplicative bias correction, which the above results showed performed best in simulations mimicking this design, and the empirical estimate $\widehat{FPR}$. Results for the different bias corrections with different estimates of the false-positive rate are presented in Table~\ref{tab:empirical.applications.2.revisit}. As before, estimated effect sizes are larger for the multiplicative correction, and the multiplicatively corrected CIs are slightly to the right of those using the additive correction. Additively corrected CIs overlap slightly with two-step CIs reported in Section~\ref{sec:applications.1}, while multiplicatively corrected CIs do not.

In view of the simulation evidence above for the topic model, the bias corrections reported in Section~\ref{sec:applications.2} use the additive correction, but the multiplicative correction produced values that agreed to within $\pm0.01$.

\begin{table}[t!]
\begin{center}
\caption{\label{tab:empirical.applications.2.revisit}Additional Results: Remote Work Application}
{\small
\begin{tabular}{lcccccccc}
\\[-8pt]
\hline \hline \\[-10pt]
 & & \multicolumn{3}{c}{No Fixed Effects} & & \multicolumn{3}{c}{With Fixed Effects}  \\ \cline{3-5} \cline{7-9} \\[-8pt]
 & & Estimate & Std Err& 95\% CI & & Estimate & Std Err& 95\% CI \\[2pt]
BCA-1 & & 0.897 & 0.119 & [0.663, 1.131] & & 0.521 & 0.081 & [0.362, 0.680] \\
BCA-2 & & 0.910 & 0.124 & [0.667, 1.154] & & 0.530 & 0.084 & [0.365, 0.694] \\
BCM-1 & & 1.052 & 0.140 & [0.778, 1.327] & & 0.641 & 0.100 & [0.446, 0.836] \\
BCM-2 & & 1.088 & 0.148 & [0.798, 1.379] & & 0.668 & 0.106 & [0.460, 0.876] \\[2pt] \hline \hline
\end{tabular}
}
\vskip 2pt
\parbox{0.95\textwidth}
{
\footnotesize
\emph{Note:} Results are presented for the additive bias correction using $\widehat{FPR}$ (BCA-1) and $\widehat{FPR}_B$ (BCA-2), and multiplicative bias correction using $\widehat{FPR}$ (BCM-1) and $\widehat{FPR}_B$ (BCM-2).
}
\end{center}
\end{table}

For central bank communication application in Section~\ref{sec:applications.3}, we implemented the two-step strategy using the empirical share $\hat \theta_i = N_i/C_i$ to estimate $\theta_i$. Table~\ref{tab:indexsim} presents a 
set of simulation results for this case as well. This share estimate is prone to two sources of measurement error: misclassification error and upstream sampling error. The results in Table~\ref{tab:indexsim} hold the former fixed and let the latter go to zero appropriately with the sample size. Evidently, these measurement errors work in different directions. With small $n$ and large $\kappa$,\footnote{Note that the $\kappa$ here is defined within the context of the topic model for the simulations in Section~\ref{sec:simulations.app2}, where upstream sampling error is the only source of measurement error. An appropriate $\kappa$ as in (\ref{eq:kappa}) for the share variable $\hat \theta_i = N_i/C_i$ must account for both misclassification error and upstream sampling error.} upstream sampling error is more important. This is a ``classical'' measurement error that causes attenuation bias. With $n = 200$ and $\kappa = 4.57$, as in the empirical application, the bias is comparable to that obtained by implementing the two-step strategy using the topic model, underestimating the true effect size by around 43\%. However, bias goes in a different direction when $n$ is large and $\kappa$ is small, so that misclassification error is more important than upstream sampling error. For instance, with $n =$ 3,200 and $\kappa = 1.14$, two-step regression onto shares now \emph{over-estimates} the true effect size by around 16\%. Correspondingly, the two-step confidence interval has coverage well below nominal coverage. 

An important take-away from these results is that  the bias of two-step estimators using AI/ML-generated variables can behave differently than in classical measurement error settings, making it difficult for researchers to determine even the sign of the bias.


\section{Conclusion} \label{sec:conclusion}

The leading approach for analyzing unstructured or high-dimensional data follows a two-step strategy. First, latent variables of economic interest are estimated using an AI-powered information retrieval algorithm or other ML method. Second, the AI- or ML-generated variables are plugged-in to downstream econometric models, and are treated as regular numeric ``data'' for the purposes of estimation and inference. 

This paper highlights, both theoretically and empirically, how measurement error introduced in the first step leads to biased estimates and invalid inference for the downstream regression coefficients.  The degree of bias, and therefore the degree to which it distorts inference, depends on the relative importance of measurement error and sampling error, but it can be substantial in practice. 

To address this problem, we propose two robust alternative inference methods: (1) an explicit bias correction with bias-corrected confidence intervals; and (2) joint maximum likelihood estimation.  In a series of simulations and applications involving label imputation, dimensionality reduction, and index construction via classification and aggregation, we show that the two-step strategy produces material biases whereas both proposed methods perform well.

\FloatBarrier
\newpage
\appendix

\section{Proofs of Main Results} \label{sec:appendix_proofs}

Here we present proofs of just the main results in the text. Proofs of additional results are deferred to Appendix~\ref{sec:appendix_supplemental}.

\begin{proof}[Proof of Theorem~\ref{theorem:drifting.general}]
We first prove part 2. We have $\frac 1n \sum_{i=1}^n \hat{\bs{\xi}}_i^{\phantom{T}}\!\!\!\;\hat{\bs{\xi}}_i^T \to_p \E\Big[{\bs{\xi}}_i^{\phantom{T}}\!\!\!\;{\bs{\xi}}_i^T\Big]$ by Assumption~\ref{assumption:drifting.general}\ref{assumption:drifting.lln}. Result (\ref{eq:two-step.drifting.result.2}) now follows by Assumptions~\ref{assumption:drifting.general}\ref{assumption:drifting.moments} and~\ref{assumption:drifting.general}\ref{assumption:drifting.clt} and Sltusky's theorem.

To establish (\ref{eq:two-step.drifting.result.1}), first write
\[
 \begin{aligned}
 \sqrt n(\hat{\bs\psi} - \bs \psi) 
 & = \Bigg( \frac 1n \sum_{i=1}^n \hat{\bs{\xi}}_i^{\phantom{T}}\!\!\!\;\hat{\bs{\xi}}_i^T \Bigg)^{-1} \Bigg( \frac{1}{\sqrt{n}} \sum_{i=1}^n \hat{\bs{\xi}}_i^{\phantom{T}}\!\!\!\;(Y_i - \hat{\bs{\xi}}_i^T \bs{\psi}) \Bigg) \\
 & = \Bigg( \frac 1n \sum_{i=1}^n \hat{\bs{\xi}}_i^{\phantom{T}}\!\!\!\;\hat{\bs{\xi}}_i^T \Bigg)^{-1} \Bigg( \frac{1}{\sqrt{n}} \sum_{i=1}^n \hat{\bs{\xi}}_i (\bs{\theta}_i - \hat{\bs{\theta}}_i)^T \bs{\gamma} \Bigg) 
 + \Bigg( \frac 1n \sum_{i=1}^n \hat{\bs{\xi}}_i^{\phantom{T}}\!\!\!\;\hat{\bs{\xi}}_i^T \Bigg)^{-1} \Bigg(  \frac{1}{\sqrt{n}} \sum_{i=1}^n \hat{\bs{\xi}}_i \varepsilon_i \Bigg) \\
 &
 =: T_{1,n} + T_{2,n} .
 \end{aligned}
\]
It follows by Assumption~\ref{assumption:drifting.general}\ref{assumption:drifting.moments}-\ref{assumption:drifting.clt} and the Continuous Mapping Theorem that
\[
 T_{2,n} \to_d N \left( \bs 0 , \mathbf V \right).
\] 
For the remaining term, we have
\[
 \frac{1}{\sqrt{n}} \sum_{i=1}^n \hat{\bs{\xi}}_i (\bs{\theta}_i - \hat{\bs{\theta}}_i)^T \bs{\gamma} 
 = \left[ \begin{array}{c} 
 -\frac{1}{\sqrt{n}} \sum_{i=1}^n \hat{\bs{\theta}}_i  (\hat{\bs{\theta}}_i - \bs{\theta}_i)^T \bs{\gamma} \\
 -\frac{1}{\sqrt{n}} \sum_{i=1}^n {\mf{q}}_i  (\hat{\bs{\theta}}_i - \bs{\theta}_i)^T \bs{\gamma} \end{array} \right]
 \to_p 
 \left[ \begin{array}{c} 
 -\kappa \bs{\Omega} \bs{\gamma} \\
 \mathbf 0 \end{array} \right]
\]
by Assumption~\ref{assumption:drifting.general}\ref{assumption:drifting.lln}. 
Hence,
\[
 T_{1,n} \to_p -\kappa \, \E\left[\bs{\xi}_i^{\phantom{T}}\!\!\!\;\bs{\xi}_i^T\right]^{-1} \left[ \begin{array}{c} 
 \bs{\Omega} \bs{\gamma} \\
 \mathbf 0 \end{array} \right]
\]
by Assumptions~\ref{assumption:drifting.general}\ref{assumption:drifting.moments} and~\ref{assumption:drifting.general}\ref{assumption:drifting.lln} and Slutsky's theorem.
\end{proof}

\begin{proof}[Proof of Theorem~\ref{theorem:drifting.ml.main}]
This is a special case of Theorem~\ref{theorem:drifting.ml}.
\end{proof}

\begin{proof}[Proof of Theorem~\ref{theorem:two-step.drifting}]
Assumption~\ref{assumption:drifting.general}\ref{assumption:drifting.moments} is implied by Assumption~\ref{assumption:drifting}\ref{assumption:dgp.q.2}.

The second and third parts of Assumption~\ref{assumption:drifting.general}\ref{assumption:drifting.lln} hold by Lemmas~\ref{lem:theta.bias} and~\ref{lem:theta.q.bias} in Appendix~\ref{sec:appendix_supplemental}. For the first part, we have 
\[
 \frac 1n \sum_{i=1}^n \hat{\bs{\xi}}_i^{\phantom{T}}\!\!\!\;\hat{\bs{\xi}}_i^T
 = 
 \frac 1n \sum_{i=1}^n \hat{\bs{\xi}}_i(\hat{\bs{\xi}}_i - \bs{\xi}_i)^T + \frac 1n \sum_{i=1}^n (\hat{\bs{\xi}}_i - \bs{\xi}_i)\bs{\xi}_i^T + \frac 1n \sum_{i=1}^n {\bs{\xi}}_i^{\phantom{T}}\!\!\!\;{\bs{\xi}}_i^T \,.
\]
The third term converges in probability to $\E [ {\bs{\xi}}_i^{\phantom{T}}\!\!\!\;{\bs{\xi}}_i^T]$ by Assumption~\ref{assumption:drifting}\ref{assumption:dgp.q.2}, while the first two terms are $o_p(1)$ by Lemmas~\ref{lem:theta.bias} and~\ref{lem:theta.q.bias}. 

Now consider Assumption~\ref{assumption:drifting.general}\ref{assumption:drifting.clt}. For the first part, we have
\[
 \frac{1}{\sqrt{n}} \sum_{i=1}^n \hat{\bs{\xi}}_i \varepsilon_i = 
 \left[ \begin{array}{c}
 \frac{1}{\sqrt n} \sum_{i=1}^n \hat{\bs{\theta}}_i \varepsilon_i \\
 \frac{1}{\sqrt n} \sum_{i=1}^n \mf{q}_i \varepsilon_i
 \end{array} \right].
\]
We may deduce by arguments similar to those in the proof of Lemma~\ref{lem:theta.bias} that 
\[
 \left\| \frac{1}{\sqrt n} \sum_{i=1}^n \hat{\bs{\theta}}_i \varepsilon_i - \bs S (\mf{B} \mf{B}^T)^{-1} \mf{B} \left( \frac{1}{\sqrt{n}} \sum_{i=1}^n \hat{\mf{p}}_i \varepsilon_i \right) \right\| \to_p 0 
\]
by Assumption~\ref{assumption:drifting}\ref{assumption:B.rank.2}-\ref{assumption:theta.rate}, where $\hat{\mf p}_i = \mf x_i/C_i$. Moreover, with $\mf p_i = \E[ \mf x_i/C_i| C_i, \bs w_i]$, we have
\[
 \begin{aligned}
 \E \left[ \varepsilon_i^2 \left\| \hat{\mf{p}}_i  - \mf{p}_i \right\|^2 \right]
 & = \E \left[ \E \left[ \left. \varepsilon_i^2 \right| \bs w_i, \mf q_i \right] \E \left[ \left. \left\| \hat{\mf{p}}_i  - \mf{p}_i \right\|^2 \right| \bs w_i, \mf q_i \right] \right] \\
 & = \E \left[ \E \left[ \left. \varepsilon_i^2 \right| \bs w_i, \mf q_i \right] \E \left[ \left. \E \left[ \left.  \left\| \hat{\mf{p}}_i  - \mf{p}_i \right\|^2 \right| \bs w_i, \mf q_i, C_i \right]  \right| \bs w_i, \mf q_i \right] \right] \\
 & = \E \left[ \E \left[ \left. \varepsilon_i^2 \right| \bs w_i, \mf q_i \right] \E \left[ \left. \frac{1}{C_i} \mr{tr} \left\{ \mr{diag}(\mf{B}^T \bs{w}_i ) - \mf{B}^{T} \bs{w}_i^{\phantom{T}}\!\!\!\;\bs{w}_i^T \mf B \right\} \right| \bs w_i, \mf q_i \right] \right] \\
 & = \E \left[ \frac{1}{C_i} \right] \E \left[ \varepsilon_i^2 \left( \mr{diag}(\mf{B}^T \bs{w}_i ) - \mf{B}^{T} \bs{w}_i^{\phantom{T}}\!\!\!\;\bs{w}_i^T \mf B \right) \right] \to 0,
 \end{aligned}
\]
where the first equality is by independence of $(\mf{x}_i,C_i)$ and $\varepsilon_i$ conditional on $(\bs w_i, \mf q_i)$, the second is by iterated expectations, the third is by Lemma~\ref{lem:expected.p} in Appendix~\ref{sec:appendix_supplemental} and independence of $\mf x_i$ and $\mf q_i$ conditional on $(C_i,\bs{w}_i)$, and the fourth is by (\ref{eq:kappa.topic}) and independence of $C_i$ and $(Y_i,\mf q_i, \bs w_i)$.
Therefore, $\frac{1}{\sqrt{n}} \sum_{i=1}^n (\hat{\mf{p}}_i - \mf{p}_i) \varepsilon_i \to_p \mf 0$ and so $\left\| \frac{1}{\sqrt{n}} \sum_{i=1}^n \hat{\bs{\xi}}_i \varepsilon_i - \frac{1}{\sqrt n} \sum_{i=1}^n \bs{\xi}_i \varepsilon_i  \right\| \to_p 0.$
It follows by the central limit theorem and Assumption~\ref{assumption:drifting}\ref{assumption:dgp.q.2} that
\[
 \frac{1}{\sqrt{n}} \sum_{i=1}^n \hat{\bs{\xi}}_i \varepsilon_i \to_d N \left( \bs{0} , \E \left[ \varepsilon_i^2 \bs{\xi}_i^{\phantom{T}}\!\!\!\;\bs{\xi}_i^T \right] \right).
\]

It remains to show $\frac 1n \sum_{i=1}^n \hat \varepsilon_i^2 \hat{\bs{\xi}}_i^{\phantom{T}}\!\!\!\;\hat{\bs{\xi}}_i^T \to_p 
 \E\left[\varepsilon_i^2 \bs{\xi}_i^{\phantom{T}}\!\!\!\;\bs{\xi}_i^T\right]$. To this end, first write 
\begin{multline*}
 \quad \frac 1n \sum_{i=1}^n \hat \varepsilon_i^2 \hat{\bs{\xi}}_i^{\phantom{T}}\!\!\!\;\hat{\bs{\xi}}_i^T 
 = \frac 1n \sum_{i=1}^n \varepsilon_i^2 \bs{\xi}_i^{\phantom{T}}\!\!\!\;\bs{\xi}_i^T
 + \frac 1n \sum_{i=1}^n \varepsilon_i^2 \left( \hat{\bs{\xi}}_i^{\phantom{T}}\!\!\!\;\hat{\bs{\xi}}_i^T - {\bs{\xi}}_i^{\phantom{T}}\!\!\!\;{\bs{\xi}}_i^T \right) \\
 + \frac 1n \sum_{i=1}^n (\hat{\varepsilon}_i^2 - \varepsilon_i^2)\hat{\bs{\xi}}_i^{\phantom{T}}\!\!\!\;\hat{\bs{\xi}}_i^T 
 =: T_{1,n} + T_{2,n} + T_{3,n}. \quad
\end{multline*}
Evidently, $T_{1,n} \to_p \E \left[ \varepsilon_i^2 \bs{\xi}_i^{\phantom{T}}\!\!\!\;\bs{\xi}_i^T \right]$ by Assumption~\ref{assumption:drifting}\ref{assumption:dgp.q.2}. For $T_{2,n}$, we have 
\[
 T_{2,n} = \left[ \begin{array}{cc}
 \frac 1n \sum_{i=1}^n \varepsilon_i^2 (\hat{\bs{\theta}}_i^{\phantom{T}}\!\!\!\;\hat{\bs{\theta}}_i^T - \bs{\theta}_i^{\phantom{T}}\!\!\!\;\bs{\theta}_i^T)
 & \frac 1n \sum_{i=1}^n \varepsilon_i^2 (\hat{\bs{\theta}}_i - \bs{\theta}_i) \mf{q}_i^T \\
 \frac 1n \sum_{i=1}^n \varepsilon_i^2 \mf{q}_i (\hat{\bs{\theta}}_i - \bs{\theta}_i)^T & \mf 0 \end{array} \right].
\]
Consider the upper-left block. We may deduce by arguments similar to those in the proof of Lemma~\ref{lem:theta.bias} that 
\[
 \left\| \frac 1n \sum_{i=1}^n \varepsilon_i^2 (\hat{\bs{\theta}}_i^{\phantom{T}}\!\!\!\;\hat{\bs{\theta}}_i^T - \bs{\theta}_i^{\phantom{T}}\!\!\!\;\bs{\theta}_i^T) - 
 \bs S (\mf{B} \mf{B}^T)^{-1} \mf{B} \left(
 \frac 1n \sum_{i=1}^n \varepsilon_i^2 (\hat{\mf{p}}_i^{\phantom{T}}\!\!\!\;\hat{\mf{p}}_i^T - \mf{p}_i^{\phantom{T}}\!\!\!\;\mf{p}_i^T) 
 \right) \mf{B}^T (\mf{B} \mf{B}^T)^{-1} \bs S^T \right\| \to_p 0,
\]
by Assumption~\ref{assumption:drifting}\ref{assumption:B.rank.2}-\ref{assumption:dgp.q.2}. Since $\mf p_i, \hat{\mf p}_i \in \Delta^{V-1}$, we have $\|\hat{\mf{p}}_i^{\phantom{T}}\!\!\!\;\hat{\mf{p}}_i^T - \mf{p}_i^{\phantom{T}}\!\!\!\;\mf{p}_i^T\| \leq 2 \|\hat{\mf{p}}_i - \mf{p}_i\|$ and so 
\[
 \left\| \frac 1n \sum_{i=1}^n \varepsilon_i^2 (\hat{\mf{p}}_i^{\phantom{T}}\!\!\!\;\hat{\mf{p}}_i^T - \mf{p}_i^{\phantom{T}}\!\!\!\;\mf{p}_i^T) \right\|
 \leq 2 \left( \max_{1 \leq i \leq n} \|\hat{\mf{p}}_i - \mf{p}_i\| \right) \frac 1n \sum_{i=1}^n \varepsilon_i^2 \to_p 0,
\]
by Lemma~\ref{lem:p.uniform}  in Appendix~\ref{sec:appendix_supplemental} and Assumption~\ref{assumption:drifting}\ref{assumption:dgp.q.2}. Now consider the off-diagonal blocks. By arguments similar to those in the proof of Lemma~\ref{lem:theta.q.bias}, we have
\[
 \left\| \frac 1n \sum_{i=1}^n \varepsilon_i^2 \mf q_i (\hat{\bs{\theta}}_i - \bs{\theta}_i)^T - 
 \left(
 \frac 1n \sum_{i=1}^n \varepsilon_i^2 \mf{q}_i (\hat{\mf{p}}_i - \mf{p}_i)^T 
 \right) \mf{B}^T (\mf{B} \mf{B}^T)^{-1} \bs S^T \right\| \to_p 0,
\]
by Assumption~\ref{assumption:drifting}\ref{assumption:B.rank.2}-\ref{assumption:dgp.q.2}. But note that 
\[
 \left\| \frac 1n \sum_{i=1}^n \varepsilon_i^2 \mf{q}_i (\hat{\mf{p}}_i - \mf{p}_i)^T  \right\| 
 \leq \left( \max_{1 \leq i \leq n} \|\hat{\mf{p}}_i - \mf{p}_i\| \right) \frac 1n \sum_{i=1}^n \varepsilon_i^2 \| \mf q_i\| \to_p 0,
\]
by Lemma~\ref{lem:p.uniform} in Appendix~\ref{sec:appendix_supplemental} and Assumption~\ref{assumption:drifting}\ref{assumption:dgp.q.2}. Therefore, $T_{2,n} \to_p \mathbf 0$.

Now consider $T_{3,n}$. We have $\hat{\varepsilon}_i - \varepsilon_i = \hat{\bs{\xi}}_i^T (\bs{\psi} - \hat{\bs{\psi}}) + (\bs{\theta}_i - \hat{\bs{\theta}}_i )^T \bs{\gamma}$, 
where
\begin{multline*}
 \quad \max_{1 \leq i \leq n} \left| \hat{\bs{\xi}}_i^T (\hat{\bs{\psi}} - \bs{\psi}) \right| \leq  \left( \max_{1 \leq i \leq n} \| \mf q_i\| \right) \|\bs{\alpha} - \hat{\bs\alpha}\| \\
 + \left( \max_{1 \leq i \leq n}  \| ( \hat{\bs{\theta}}_i - \bs S (\hat{\mf{B}} \hat{\mf{B}}^T)^{-1} \hat{\mf{B}} \hat{\mf{p}}_i) \| + \| \bs S (\hat{\mf{B}} \hat{\mf{B}}^T)^{-1} \hat{\mf{B}} \| \right) \|\hat{\bs{\gamma}} - \bs{\gamma}\|  \to_p 0, \quad
\end{multline*}
where the first term is by $\sqrt n$-consistency of $\hat{\bs{\alpha}}$ and the fact that $n^{-1/4} \max_{1 \leq i \leq n} \| \mf q_i\| \to_p 0$ by Assumption~\ref{assumption:drifting}\ref{assumption:dgp.q.2}, and the second term follows by Assumption~\ref{assumption:drifting}\ref{assumption:theta.rate}, consistency of $\hat{\boldsymbol{\gamma}}$, $\|\hat{\mf p}_i\| \leq 1$, and because $\|(\hat{\mf{B}} \hat{\mf{B}}^T)^{-1} \hat{\mf{B}}\| = O_p(1)$ by Assumption~\ref{assumption:drifting}\ref{assumption:B.rank.2}-\ref{assumption:B.rate}. Moreover,
\begin{multline*}
 \max_{1 \leq i \leq n} |(\hat{\bs{\theta}}_i - \bs{\theta}_i)^T \bs{\gamma}| 
 \leq 
 \bigg( \max_{1 \leq i \leq n} \| \hat{\bs \theta}_i - \bs S (\hat{\mf{B}} \hat{\mf{B}}^T)^{-1} \hat{\mf{B}}\hat{\mf p}_i \| \\
 + \| \bs S \| \left\| (\hat{\mf{B}} \hat{\mf{B}}^T)^{-1} \hat{\mf{B}} - (\mf{B} \mf{B}^T)^{-1} \mf{B} \right\| 
 + \left\|\bs S  (\mf{B} \mf{B}^T)^{-1} \mf{B} \right\| \max_{1 \leq i \leq n} \left\| \hat{\mf{p}}_i - \mf{p}_i \right\| \bigg) \left\| \bs{\gamma} \right\|.
\end{multline*}
Consider the three terms in parentheses on the right-hand side of this display. The first two terms converge in probability to zero by Assumption~\ref{assumption:drifting}\ref{assumption:B.rank.2}-\ref{assumption:theta.rate}, and the third converges in probability to zero by Lemma~\ref{lem:p.uniform}. 
Hence, $\max_{1 \leq i \leq n} |\hat{\varepsilon}_i - \varepsilon_i| \to_p 0$.

Now, since 
\[
 \hat \varepsilon_i^2 - \varepsilon_i^2 = 2(\hat \varepsilon_i - \varepsilon_i) \varepsilon_i + (\hat \varepsilon_i - \varepsilon_i)^2,
\]
we have
\[
 T_{3,n} = \frac 2n \sum_{i=1}^n (\hat \varepsilon_i - \varepsilon_i) \varepsilon_i \hat{\bs{\xi}}_i^{\phantom{T}}\!\!\!\;\hat{\bs{\xi}}_i^T + \frac 1n \sum_{i=1}^n (\hat \varepsilon_i - \varepsilon_i)^2 \hat{\bs{\xi}}_i^{\phantom{T}}\!\!\!\;\hat{\bs{\xi}}_i^T,
\]
and so 
\[
\begin{aligned}
 \|T_{3,n}\| & \leq 2 \left( \max_{1 \leq i \leq n} |\hat{\varepsilon}_i - \varepsilon_i| \right) \frac 1n \sum_{i=1}^n |\varepsilon_i| \|\hat{\bs{\xi}}_i\|^2 +
 \left( \max_{1 \leq i \leq n} |\hat{\varepsilon}_i - \varepsilon_i|^2 \right) \frac 1n \sum_{i=1}^n  \|\hat{\bs{\xi}}_i\|^2 \\
 & = \left( \max_{1 \leq i \leq n} |\hat{\varepsilon}_i - \varepsilon_i| \right) \mr{tr} \left\{ \frac 2n \sum_{i=1}^n |\varepsilon_i| \hat{\bs{\xi}}_i^{\phantom{T}}\!\!\!\;\hat{\bs{\xi}}_i^T \right\} +
 \left( \max_{1 \leq i \leq n} |\hat{\varepsilon}_i - \varepsilon_i|^2 \right)  \mr{tr} \left\{ \frac 1n \sum_{i=1}^n \hat{\bs{\xi}}_i^{\phantom{T}}\!\!\!\;\hat{\bs{\xi}}_i^T \right\} \to_p 0,
\end{aligned}
\]
because $\frac 1n \sum_{i=1}^n |\varepsilon_i| \hat{\bs{\xi}}_i^{\phantom{T}}\!\!\!\!\;\hat{\bs{\xi}}_i^T = O_p(1) $ by control of $T_{1,n}$ and $T_{2,n}$, which imply $\frac 1n \sum_{i=1}^n \varepsilon_i^2 \hat{\bs{\xi}}_i^{\phantom{T}}\!\!\!\;\hat{\bs{\xi}}_i^T = O_p (1)$, and $\frac 1n \sum_{i=1}^n \hat{\bs{\xi}}_i^{\phantom{T}}\!\!\!\;\hat{\bs{\xi}}_i^T = O_p(1)$ by Lemmas~\ref{lem:theta.bias} and~\ref{lem:theta.q.bias}, and Assumption~\ref{assumption:drifting}\ref{assumption:dgp.q.2}. 
\end{proof}

\begin{proof}[Proof of Theorem~\ref{theorem:CI}]
We first prove part 1.  In the proof of Theorem~\ref{theorem:drifting.general}, it was shown that $\left( \frac 1n \sum_{i=1}^n \hat{\bs{\xi}}_i^{\phantom{T}}\!\!\!\;\hat{\bs{\xi}}_i^T \right)^{-1} \to_p \E\Big[\bs{\xi}_i^{\phantom{T}}\!\!\!\;\bs{\xi}_i^T\Big]^{-1}$. Hence, by consistency of $\hat \kappa$ and $\hat{\bs\Omega}$, we have that
\begin{equation} \label{eq:CI.1}
 \hat \kappa \left( \frac 1n \sum_{i=1}^n \hat{\boldsymbol{\xi}}_i^{\phantom{T}}\!\!\!\;\hat{\boldsymbol{\xi}}_i^T \right)^{-1} \left[ \begin{array}{cc} 
 \hat{\boldsymbol{\Omega}} & \mathbf 0 \\
 \mathbf 0 & \mathbf 0 \end{array} \right] \to_p 
 \kappa \, \E\Big[\bs{\xi}_i^{\phantom{T}}\!\!\!\;\bs{\xi}_i^T\Big]^{-1} \left[ \begin{array}{cc} 
 \boldsymbol{\Omega} & \mathbf 0 \\
 \mathbf 0 & \mathbf 0 \end{array} \right].
\end{equation}
Since the matrix on the right-hand side is finite, we have that 
\[
 \frac{\hat \kappa}{\sqrt n} \left( \frac 1n \sum_{i=1}^n \hat{\boldsymbol{\xi}}_i^{\phantom{T}}\!\!\!\;\hat{\boldsymbol{\xi}}_i^T \right)^{-1} \left[ \begin{array}{cc} 
 \hat{\boldsymbol{\Omega}} & \mathbf 0 \\
 \mathbf 0 & \mathbf 0 \end{array} \right] = O_p(n^{-1/2}),
\]
and so
\[
 \left( \mf I -  \frac{\hat \kappa}{\sqrt n} \left( \frac 1n \sum_{i=1}^n \hat{\boldsymbol{\xi}}_i^{\phantom{T}}\!\!\!\;\hat{\boldsymbol{\xi}}_i^T \right)^{-1} \left[ \begin{array}{cc} 
 \hat{\boldsymbol{\Omega}} & \mathbf 0 \\
 \mathbf 0 & \mathbf 0 \end{array} \right] \right)^{-1} 
 = \mf I + \frac{\hat \kappa}{\sqrt n} \left( \frac 1n \sum_{i=1}^n \hat{\boldsymbol{\xi}}_i^{\phantom{T}}\!\!\!\;\hat{\boldsymbol{\xi}}_i^T \right)^{-1} \left[ \begin{array}{cc} 
 \hat{\boldsymbol{\Omega}} & \mathbf 0 \\
 \mathbf 0 & \mathbf 0 \end{array} \right] + O_p(n^{-1}),
\]
because $(\mf I - \mf A)^{-1} = \mf I + \mf A + O(\|\mf A\|^2)$ as $\|\mf A\| \to 0$ and the inverse exists with probability approaching one. Post-multiplying both sides by $\hat{\bs \psi}$ gives that $\hat{\bs \psi}^{bcm} = \hat{\bs \psi}^{bca} + O_p(n^{-1})$. 

Since both bias-corrected estimators are first-order asymptotically equivalent, it suffices to analyze $\hat{\bs \psi}^{bca}$. We have
\[
 \sqrt n(\hat{\bs \psi}^{bca} - \bs \psi) = \sqrt n (\hat{\bs \psi} - \bs \psi) + \hat \kappa \left( \frac 1n \sum_{i=1}^n \hat{\boldsymbol{\xi}}_i^{\phantom{T}}\!\!\!\;\hat{\boldsymbol{\xi}}_i^T \right)^{-1} \left[ \begin{array}{cc} 
 \hat{\boldsymbol{\Omega}} & \mathbf 0 \\
 \mathbf 0 & \mathbf 0 \end{array} \right] \hat{\bs \psi}.
\]
The first term is asymptotically normal with mean and variance given by (\ref{eq:two-step.drifting.result.1}). For the second term, Theorem~\ref{theorem:drifting.general} implies $\hat{\bs\psi} \to_p \bs \psi$, so it follows by (\ref{eq:CI.1}) that 
\[
 \hat \kappa \left( \frac 1n \sum_{i=1}^n \hat{\boldsymbol{\xi}}_i^{\phantom{T}}\!\!\!\;\hat{\boldsymbol{\xi}}_i^T \right)^{-1} \left[ \begin{array}{cc} 
 \hat{\boldsymbol{\Omega}} & \mathbf 0 \\
 \mathbf 0 & \mathbf 0 \end{array} \right] \hat{\bs \psi} \to_p \kappa \, \E\Big[\bs{\xi}_i^{\phantom{T}}\!\!\!\;\bs{\xi}_i^T\Big]^{-1} \left[ \begin{array}{cc} 
 \bs{\Omega} & \mathbf 0 \\
 \mathbf 0 & \mathbf 0\end{array} \right] \bs \psi.
\]
Combining with (\ref{eq:two-step.drifting.result.1}) and using the continuous mapping theorem, we conclude that 
\[
 \sqrt n(\hat{\bs \psi}^{bca} - \bs \psi)  \to_d N(\bs 0, \mf V),
\]
as required.

Part 2.~now follows from part 1., consistency of $\hat{\mf V}$, which was established in Theorem~\ref{theorem:drifting.general}, and positive-definiteness of $\mf V$.
\end{proof}

\begin{proof}[Proof of Lemma~\ref{lem:kappa.ex1}]
Let $FP_i := \E[\hat \theta_i(1-\theta_i)|\mf x_i, \mf p_i]$. Since $\sqrt n \, \E[\hat \theta_i(1-\theta_i)] \to \kappa > 0$, it is enough to show that 
\begin{equation} \label{eq:kappa.ex1.1}
 \frac{\widehat{FPR}}{\frac 1m \sum_{i=1}^m FP_i} \to_p 1,
\end{equation}
and
\begin{equation} \label{eq:kappa.ex1.2}
 \frac{\frac 1m \sum_{i=1}^m FP_i}{\E[\hat \theta_i(1-\theta_i)]} \to_p 1.
\end{equation}

We first show (\ref{eq:kappa.ex1.2}). By Chebyshev's inequality, with probability at least $1-C^{-2}$ we have
\[
 \left| \sum_{i=1}^m FP_i - m \, \E[\hat \theta_i(1-\theta_i)] \right| \leq  C \sqrt{m \, \E[(\hat \theta_i(1-\theta_i))^2]},
\]
for any $C > 0$. But $\E[(\hat \theta_i(1-\theta_i))^2] < \E[\hat \theta_i(1-\theta_i)]$ because $0 \leq \hat \theta_i(1-\theta_i) \leq 1$. Moreover, the conditions $\sqrt n \, \E[\hat \theta_i(1-\theta_i)] \to \kappa > 0$ and $n/m^2 \to 0$ together imply $m \, \E[\hat \theta_i(1-\theta_i)] \to + \infty$. Setting $C = \epsilon_n \sqrt{m \, \E[\hat \theta_i(1-\theta_i)]}$ with $0 < \epsilon_n < \frac 12$ and $\epsilon_n \to 0$ sufficiently slowly that $C \to \infty$, we deduce that 
\begin{equation} \label{eq:kappa.ex1.3}
 \left| \frac{ \sum_{i=1}^m FP_i }{m \, \E[\hat \theta_i(1-\theta_i)]} - 1 \right| \leq \epsilon_n 
\end{equation}
holds with probability approaching one. This proves (\ref{eq:kappa.ex1.2}).

Now consider (\ref{eq:kappa.ex1.1}). Conditional on $(\mf x_i, \mf q_i)_{i=1}^m$, each $\hat \theta_i (1-\theta_i)$ are independent Bernoulli random variables with success probability $FP_i$. By Chernoff's inequality, for any $\delta > 0$ we have
\[
 \Pr \left( \left. \left| \frac{\widehat{FPR}}{\frac{1}{m} \sum_{i=1}^m FP_i} - 1 \right| > \delta \right| (\mf x_i, \mf q_i)_{i=1}^m \right) \leq 2 e^{-\delta^2 \sum_{i=1}^m FP_i/3} .
\]
Letting $\mc A_n$ denote the event upon which (\ref{eq:kappa.ex1.3}) holds, we then have
\[
 \begin{aligned}
 \Pr \left( \left. \left| \frac{\widehat{FPR}}{\frac{1}{m} \sum_{i=1}^m FP_i} - 1 \right| > \delta \right. \right) 
 & \leq 
 \E \left[ \Pr \left( \left. \left| \frac{\widehat{FPR}}{\frac{1}{m} \sum_{i=1}^m FP_i} - 1 \right| > \delta \right|  (\mf x_i, \mf q_i)_{i=1}^m  \right) \mb I[  (\mf x_i, \mf q_i)_{i=1}^m  \in \mc A_n) ] \right] \\
 & \quad + \Pr( (\mf x_i, \mf q_i)_{i=1}^m  \in \mc A_n^c) \\
 & \leq 2 e^{-\delta^2 (1-\epsilon_n) m \, \E[\hat \theta_i(1-\theta_i)] /3} + \Pr((\mf x_i, \mf q_i)_{i=1}^m  \in A_n^c) \to 0,
 \end{aligned}
\]
since $\epsilon_n \to 0$ and $m \,\E[\hat \theta_i(1-\theta_i)] \to + \infty$.
\end{proof}

\begin{proof}[Proof of Lemma~\ref{lem:kappa.ex2}]
Consistency of $\hat{\bs \Omega}$ follows by similar arguments to Lemmas~\ref{lem:convergence.theta} and~\ref{lem:theta.bias}, using the condition $\bar{\bs w}_n \to_p \E[\bs{w}_i]$. Moreover, by Chebyshev's inequality, for any $\delta > 0$ we have
\[
 \Pr \left( \left| \hat \kappa - \sqrt n \, \E[C_i^{-1}] \right| > \delta \right) \leq \frac{1}{\delta^2} \E[C_i^{-2}].
\]
As $C_i \geq 1$ and $\sqrt n \, \E[C_i^{-1}] \to \kappa \geq 0$, we have $\E[C_i^{-2}] \leq \E[C_i^{-1}] \to 0$. Hence, $\hat \kappa \to_p \kappa$.
\end{proof}


\let\oldbibliography\thebibliography
\renewcommand{\thebibliography}[1]{\oldbibliography{#1}
  \setlength{\itemsep}{0pt}}

\bibliography{Unstruct.bib}


\FloatBarrier
\newpage
\renewcommand*{\thepage}{A.\arabic{page}}
\setcounter{page}{1}

\newpage
\numberwithin{table}{section}
\numberwithin{figure}{section}


\

\begin{center}
  {
    \LARGE Supplemental Appendix: Inference for Regression with \\[8pt] Variables Generated by AI or Machine Learning
  }

  \

  \begin{tabular}{c@{\hskip 1in}c}
    \large Laura Battaglia    & \large Timothy Christensen \\
    \large Oxford             & \large Yale                \\[2ex]
    \large Stephen Hansen     & \large Szymon Sacher       \\
    \large UCL, IFS, and CEPR & \large Meta
  \end{tabular}

  \

  \

  \date{\large April 29, 2025}

\end{center}

\section{Additional Results and Discussion} \label{sec:appendix_examples}

\subsection{Fixed-DGP Asymptotics}

For completeness, here we study the large-sample properties of $\hat{\boldsymbol{\psi}}$ as the number of observations grows ($n \to \infty$), while the distribution of $(Y_i,\boldsymbol{\xi}_i,\hat{\boldsymbol{\xi}}_i)_{i=1}^n$ remains fixed. This fixed-DGP framework approximates empirical settings with a large number of observations and a large amount of measurement error per observation.

\begin{assumption}\label{assumption:fixed.general}
\begin{enumerate}[label=(\roman*),nosep]
\item \label{assumption:general.moments}
$\E\Big[ \|\boldsymbol \xi_i\|^2 \Big] < \infty$,  and $\E\Big[\bs{\xi}_i^{\phantom{T}}\!\!\!\;\bs{\xi}_i^T\Big]$ has full rank.

\item \label{assumption:general.lln}
$\frac 1n \sum_{i=1}^n {\bs{\xi}}_i^{\phantom{T}}\!\!\!\;{\bs{\xi}}_i^T \to_p \E\Big[{\bs{\xi}}_i^{\phantom{T}}\!\!\!\;{\bs{\xi}}_i^T\Big]$, 
$\frac 1n \sum_{i=1}^n (\hat{\boldsymbol{\theta}}_i - \bs{\theta}_i)(\hat{\bs{\theta}}_i - \bs{\theta}_i)^T \to_p \mathbf H$ with $\mathbf H$ a finite non-random symmetric matrix, 
$\frac 1n \sum_{i=1}^n (\hat{\boldsymbol{\theta}}_i^{\phantom{T}} - {\bs{\theta}}_i^{\phantom{T}}\!\!\!\;)\bs{\theta}_i^T \to_p \mathbf G$ with $\mathbf G$ a finite non-random matrix,
$\frac 1n \sum_{i=1}^n (\hat{\boldsymbol{\theta}}_i^{\phantom{T}} - {\bs{\theta}}_i^{\phantom{T}}\!\!\!\;)\mf{q}_i^T \to_p \mathbf 0$, and 
$\frac 1n \sum_{i=1}^n \hat{\bs{\xi}}_i \varepsilon_i \to_p \mathbf 0$ as $n \to \infty$.
\end{enumerate}
\end{assumption}

Assumption~\ref{assumption:fixed.general}\ref{assumption:general.moments} is standard. Assumption~\ref{assumption:fixed.general}\ref{assumption:general.lln} requires that $(\hat{\bs{\xi}}_i,\bs{\xi}_i,\varepsilon_i)$ satisfy some laws of large numbers. Only the last two conditions in  Assumption~\ref{assumption:fixed.general}\ref{assumption:general.lln} are substantive. They ensure that the measurement errors $\hat{\bs{\theta}}_i - \bs{\theta}_i$ are uncorrelated with $\mf{q}_i$ and the regression errors $\varepsilon_i$ are uncorrelated with $\hat{\bs{\xi}}_i$ asymptotically.  Let $\bs \Delta = \mathbf H + \mathbf G + \mathbf G^T$.

\begin{theorem}\label{theorem:two-step.general}

Suppose that Assumption~\ref{assumption:fixed.general} holds. Then
\begin{equation} \label{eq:two-step.fixed}
 \hat{\bs{\psi}} \to_p 
 \left( 
 \E\Big[{\bs{\xi}}_i^{\phantom{T}}\!\!\!\;{\bs{\xi}}_i^T\Big] 
 + \left[ \begin{array}{cc}
 \boldsymbol \Delta & \mathbf 0 \\
 \mathbf 0 & \mathbf 0 
 \end{array} \right] 
 \right)^{-1} 
 \left( 
 \E\Big[{\bs{\xi}}_i^{\phantom{T}}\!\!\!\;{\bs{\xi}}_i^T\Big]
 + \left[ \begin{array}{cc}
 \mathbf G & \mathbf 0 \\
 \mathbf 0 & \mathbf 0 
 \end{array} \right] 
 \right) \boldsymbol \psi,
\end{equation}
as $n \to \infty$, provided the inverse exists. In particular, if $\boldsymbol \Delta$ and $\mathbf G$ are small,
\begin{equation} \label{eq:two-step.fixed.plim}
 \mathrm{plim}(\hat{\bs{\psi}}) 
 = \bs{\psi} - \E\Big[{\bs{\xi}}_i^{\phantom{T}}\!\!\!\;{\bs{\xi}}_i^T\Big]^{-1} 
 \left[ \begin{array}{c} 
 (\mathbf H + \mathbf G^T) \bs \gamma \\ \mathbf 0 
 \end{array} \right] 
 + O\left(\|\boldsymbol \Delta\| \max\{\|\boldsymbol \Delta\|, \|\mathbf G\|\}\right).
\end{equation} 
\end{theorem}

Theorem~\ref{theorem:two-step.general} shows that $\hat{\bs{\psi}}$ is inconsistent due to measurement error in $\hat{\boldsymbol{\theta}}_i$. 
More constructively, Theorem~\ref{theorem:two-step.general} shows that the measurement error bias is, to first order, proportional to the precision of $\bs{\theta}_i$. The matrix $\mathbf H$ represents the variance of measurement error, while $\mathbf G$ represents the covariance of measurement error and $\bs \theta_i$. In many cases, measurement error is ``classical'' ($\mathbf H > \mathbf 0$, $\mathbf G = \mathbf 0$), but in ``non-classical'' settings, such as latent binary labels \citep{aignerRegressionBinaryIndependent1973}, $\mathbf G \neq \mathbf 0$. Expressions for $\mathbf H$ and $\mathbf G$ in the context of AI/ML-generated labels and topic models are derived in the following subsections.

\subsection{AI/ML-Generated Labels}
\label{sec:appendix_examples.labels}

Here we first generalize the basic framework from Section~\ref{sec:labels.theory.main} to allow for multiple categories and randomized classifiers. Let the vector $\bs \theta_i = (\theta_{i,k})_{k=1}^K$ indicate membership of $K+1$ distinct categories labeled $0, 1, \ldots, K$. Thus, if individual $i$ belongs to category $k$, we have $\theta_{i,k} = 1$ and $\theta_{i,j} = 0$ for all $j \neq k$. Let $p_k(\mf x_i)$ denote the true conditional probability $\Pr(\theta_{i,k} = 1|\mf x_i)$, and let $\mf p(\mf x_i) = (p_k(\mf x_i))_{k=1}^K$. If $\mf q_i$ is relevant for predicting $\bs \theta_i$, then we implicitly treat $\mf q_i$ as a component of $\mf x_i$ to simplify notation.

For the classifier, we introduce a function $\mf r(\mf x_i, \,\cdot\,): [0,1] \to \{0,1\}^K$ and, for each observation $i$, a random variable $U_i \sim U[0,1]$ drawn independent of $(\mf x_i, \mf q_i, Y_i, \bs \theta_i)$ and all other $U_j$, $j \neq i$,  so that $\mf r(\mf x_i, U_i)| \mf x_i \sim \mbox{Multinomial}(1, \bs \pi(\mf x_i))$. Here $\bs \pi(\mf x_i) = (\pi_k(\mf x_i))_{k=1}^K$, where $\pi_k(\mf x_i)$ denotes the probability that the classifier assigns label $k$ given $\mf x_i$. This nests deterministic classifiers, where $\mathbf r(\mf x_i,U_i) = \mathbf r(\mf x_i) = \bs \pi(\mf x_i)$, with $\pi_k(\mf x_i) = 1$ for at most one $k \in \{1,\ldots,K\}$ (with $k$ depending on $\mf x_i$) and $\pi_j(\mf x_i) = 0$ for all $j \neq k$.

\subsubsection{Fixed-DGP Asymptotics}

We first provide primitive conditions for the fixed-DGP case and derive expressions for the corresponding asymptotic bias. The following assumptions are required to hold for a fixed distribution of $(Y_i,\mathbf{q}_i,\mf x_i, \bs \theta_i)_{i=1}^n$ as the sample size $n$ becomes large.

\begin{assumption}\label{assumption:fixed.ml}
  \begin{enumerate}[label=(\roman*),nosep]

\item \label{assumption:theta.hat.ml}
$\max_{1 \leq i \leq n} \| \hat{\boldsymbol{\theta}}_i - \mf r(\mf x_i, U_i) \| \to_p 0$.

\item \label{assumption:dgp.ml}
$\E\left[ \|\mathbf q_i\|^2 \right] < \infty$, $\E\left[(1 + \|\mathbf q_i\|) |\varepsilon_i|\right] < \infty$, and $\E\left[\bs{\xi}_i^{\phantom{T}}\!\!\!\;\bs{\xi}_i^T\right]$ has full rank.

\item \label{assumption:q.unbiased}
$\E\left[ \left(\bs \pi(\mf x_i) - \mf p(\mf x_i)\right) \mf{q}_i^T\right] = \mf 0$.

\item \label{assumption:fixed.ml.eps}
$\E\left[ \bs \pi(\mf x_i) \varepsilon_i \right] = \mf 0$.
\end{enumerate}
\end{assumption}

Assumption~\ref{assumption:fixed.ml}\ref{assumption:theta.hat.ml} imposes minimal structure on the AI/ML-generated predictions $\hat{\bs \theta}_i$. It allows the classifier to be pre-trained, in which case one should interpret the analysis as holding conditional on the training sample. Assumption~\ref{assumption:fixed.ml}\ref{assumption:dgp.ml} imposes standard moment and rank conditions. Assumption~\ref{assumption:fixed.ml}\ref{assumption:q.unbiased} requires the classifier's prediction errors to be uncorrelated with the controls $\mf q_i$. As discussed in the main text, it is straightforward to relax this condition without changing our main point. Finally, Assumption~\ref{assumption:fixed.ml}\ref{assumption:fixed.ml.eps} says that the regression errors $\varepsilon_i$ are uncorrelated with $\bs \pi(\mf x_i)$. A sufficient conditions is $\E[\varepsilon_i |\mf x_i , \mf q_i ] = 0$.

\begin{theorem}\label{theorem:two-step.ml}

Suppose that Assumption~\ref{assumption:fixed.ml} holds. Then Assumption~\ref{assumption:fixed.general} holds and the OLS estimator $\hat{\bs \psi}$ has probability limit given by (\ref{eq:two-step.fixed}), with 
\[
 \begin{aligned}
 \mathbf H 
 & = \E \left[ \mr{diag}(\bs \pi(\mf x_i) + \mf p(\mf x_i)) - \bs \pi(\mf x_i) \mf p(\mf x_i)^T - \mf p(\mf x_i) \bs \pi(\mf x_i)^T \right],
 \\
 \mathbf G 
 & = \E \left[ \bs \pi(\mf x_i) \mf p(\mf x_i)^T - \mr{diag}(\mf p(\mf x_i)) \right].
 \end{aligned}
\]
If $\mathbf H$ and $\mathbf G$ are small, then first-order bias is proportional to
\[
 \mathbf H + \mathbf G^T = \E \left[ \mr{diag}(\bs \pi(\mf x_i)) - \bs \pi(\mf x_i) \mf p(\mf x_i)^T \right].
\]
\end{theorem}

\subsubsection{Sequence-of-DGPs Asymptotics}

Consider the matrix $\mf H$ from Theorem~\ref{theorem:two-step.ml}. 
Misclassification rates for each of the $K$ labels are collected along the diagonal of $\mf H$:
\[
 (\mf H)_{k,k} = \E \left[ \pi_k(\mf x_i) + p_k(\mf x_i) - 2 \pi_k(\mf x_i) p_k(\mf x_i) \right], \quad k = 1,\ldots,K.
\]
The first condition we require is that the sum of the misclassification rates vanishes:
\begin{equation} \label{eq:kappa.ml.0}
  \mr{tr}(\mf H) \to 0
\end{equation}
as $n \to \infty$. 
This condition requires that the true probabilities $p_k(\mf x_i) = \Pr(\theta_{i,k} = 1 | \mf x_i)$ converge to zero or one (i.e., accurate prediction of $\bs \theta_i$ is possible given $\mf x_i$), and that the differences $\pi_k(\mf x_i) - p_k(\mf x_i)$ converge to zero (i.e., the classifier produces correct labels). 

We also place some structure on the false-positive rates. Let $FP(\mf x_i) = \sum_{k=1}^K FP_{k}(\mf x_i)$ denote the total false-positive rate for individual $i$, where $FP_{k}(\mf x_i) = \pi_k(\mf x_i)(1 - p_k(\mf x_i))$ denotes the individual's false-positive probability for label $k$. We require
\begin{equation} \label{eq:kappa.ml}
 \lim_{n \to \infty} \sqrt n \, \E \left[ \mr{diag}(\bs \pi(\mf x_i)) - \bs \pi(\mf x_i) \mf p(\mf x_i)^T \right] = \kappa \, \bs \Omega, 
\end{equation}
where
\begin{equation} \label{eq:kappa.ml.1}
 \lim_{n \to \infty} \sqrt n \, \E \left[ FP(\mf x_i) \right] = \kappa \geq 0,
\end{equation}
and
\begin{equation} \label{eq:omega.ml}
 \lim_{n \to \infty} \frac{\E \left[ \mr{diag}(\bs \pi(\mf x_i)) - \bs \pi(\mf x_i) \mf p(\mf x_i)^T \right]}{\E \left[ FP(\mf x_i) \right]} = \bs \Omega,
\end{equation}
assuming both limits exist. In words, $\kappa = 0$ corresponds to a case where the false-positive rate across all categories vanishes faster than sampling error. Conversely, $\kappa > 0$ allows the total false-positive rate to be the same order as sampling error. If there is a single category ($K = 1$) then $\bs \Omega = 1$. More generally, $\bs \Omega$ quantifies the relative frequency with which false-positives occur among the $K$ alternatives.

\begin{assumption}\label{assumption:drifting.ml}
  \begin{enumerate}[label=(\roman*),nosep]

\item \label{assumption:kappa.ml.2}
Conditions (\ref{eq:kappa.ml.0})-(\ref{eq:omega.ml}) hold.

\item \label{assumption:theta.hat.ml.2}
$\sqrt n \max_{1 \leq i \leq n} \| \hat{\boldsymbol{\theta}}_i - \mf r(\mf x_i, U_i) \| \to_p 0$.

\item \label{assumption:dgp.ml.2}
$\E\left[ \|\mathbf q_i\|^4 \right] < \infty$, $\E\left[\varepsilon_i^4\right] < \infty$, and $\E\left[\bs{\xi}_i^{\phantom{T}}\!\!\!\;\bs{\xi}_i^T\right]$ has full rank.

\item \label{assumption:q.unbiased.2}
$\sqrt n \, \E\left[ \left(\bs \pi(\mf x_i) - \mf p(\mf x_i)\right) \mf{q}_i^T\right] \to \mf 0$.

\item \label{assumption:errors.ml.2}
$\sqrt n \, \E \left[ \bs \pi(\mf x_i) \varepsilon_i \right] \to \mf 0$.
\end{enumerate}
\end{assumption}

Assumption~\ref{assumption:drifting.ml}\ref{assumption:kappa.ml.2} formalizes the asymptotic framework. Assumption~\ref{assumption:drifting.ml}\ref{assumption:theta.hat.ml.2} slightly strengthens Assumption~\ref{assumption:fixed.ml}\ref{assumption:theta.hat.ml} to require convergence at a faster-than-root-$n$ rate. Assumption~\ref{assumption:drifting.ml}\ref{assumption:dgp.ml.2} is standard. Assumption~\ref{assumption:drifting.ml}\ref{assumption:q.unbiased.2} and~\ref{assumption:drifting.ml}\ref{assumption:errors.ml.2} weaken Assumptions~\ref{assumption:fixed.ml}\ref{assumption:q.unbiased} and~\ref{assumption:fixed.ml}\ref{assumption:fixed.ml.eps}. As before, it is possible to relax these conditions without changing our main point.

\begin{theorem} \label{theorem:drifting.ml}
Suppose that Assumption~\ref{assumption:drifting.ml} holds. Then Assumption~\ref{assumption:drifting.general} holds and the OLS estimator $\hat{\bs \psi}$ has asymptotic distribution given by (\ref{eq:two-step.drifting.result.1}) with $\kappa$ given in (\ref{eq:kappa.ml.1}) and $\mathbf \Omega$ given in (\ref{eq:omega.ml}). Moreover, two-step standard errors are consistent.
\end{theorem}

\subsection{Topic Models}
\label{sec:appendix_examples.topic}

Section~\ref{sec:topic.theory.main} presented a set of results for the sequence-of-DGPs asymptotic framework. Here we present a complementary set of results for the fixed-DGP case, before turning to a discussion of identification of the topic model in our setting.

\subsubsection{Fixed-DGP Asymptotics}
\label{sec:appendix_examples.topic.fixed}

As in the main text, we implicitly assume that the document size $C_i$ is independent of $(\boldsymbol{w}_i, \mathbf{q}_i,Y_i)$. We also assume that $\mf x_i$ and $\mf q_i$ are independent conditional on $(C_i,\bs{w}_i)$, and that $\varepsilon_i$ and $(\mf x_i,C_i)$ are independent conditional on $(\boldsymbol{w}_i,\mathbf q_i)$. In effect, the latter two assumptions ensure the multinomial sampling error and regression errors are uncorrelated.  These assumptions seem very reasonable and can be relaxed: doing so simply complicates the expressions below. We also implicitly require that $\E[\varepsilon_i(\bs w_i, \mf q_i)] = \mf 0$. That is, no relevant topic weights have been omitted from the regression.
 The following assumptions are required to hold for a fixed distribution of $(Y_i,\mathbf{q}_i,\mathbf{x}_i,\boldsymbol{w}_i,C_i)_{i=1}^n$ as $n$ becomes large.

\begin{assumption}\label{assumption:fixed}
  \begin{enumerate}[label=(\roman*),nosep]
\item \label{assumption:B.rank}
$\mf{B}$ has full rank.

\item \label{assumption:B.consistent}
$\hat{\mf{B}} \to_p \mf{B}$.

\item \label{assumption:theta.hat}
$\max_{1 \leq i \leq n} \| \hat{\boldsymbol{\theta}}_i - \bs{S}(\hat{\mf{B}} \hat{\mf{B}}^T)^{-1} \hat{\mf{B}} (\mf x_i/C_i) \| \to_p 0$.

\item \label{assumption:dgp.q}
$\E\left[ \|\mathbf q_i\|^2 \right] < \infty$, $\E\left[(1 + \|\mathbf q_i\|) |\varepsilon_i|\right] < \infty$, and $\E\left[\bs{\xi}_i^{\phantom{T}}\!\!\!\;\bs{\xi}_i^T\right]$ has full rank.
\end{enumerate}
\end{assumption}

Assumption~\ref{assumption:fixed}\ref{assumption:B.rank} says that none of the topics are redundant.  
Assumption~\ref{assumption:fixed}\ref{assumption:B.consistent} says that $\hat{\mathbf B}$ is consistent, which is satisfied by many estimators for topic models including those of \cite{bingOptimalEstimationSparse2020}, \cite{wuSparseTopicModeling2023}, and \cite{keUsingSVDTopic2022}.  
Assumption~\ref{assumption:fixed}\ref{assumption:theta.hat} imposes some structure on the $\hat{\boldsymbol{\theta}}_i$.  This condition is not vacuous: $\boldsymbol{\theta}_i = \boldsymbol S(\mathbf{B} \mathbf{B}^T)^{-1}\mathbf{B} \E[\mathbf x_i/C_i|C_i, \boldsymbol w_i]$ by Assumption~\ref{assumption:fixed}\ref{assumption:B.rank}, so one could set $\hat{\boldsymbol{\theta}}_i = \boldsymbol{S}(\hat{\mf{B}} \hat{\mf{B}}^T)^{-1} \hat{\mf{B}} (\mathbf x_i/C_i)$. 
Assumption~\ref{assumption:fixed}\ref{assumption:dgp.q} is standard.

\begin{theorem}\label{theorem:two-step.topic}

Suppose that Assumption~\ref{assumption:fixed} holds. Then Assumption~\ref{assumption:fixed.general} holds and the OLS estimator $\hat{\bs \psi}$ has probability limit given by (\ref{eq:two-step.fixed}) with 
\[
 \mathbf H 
 = \E\left[\frac{1}{C_i} \right] \left( \boldsymbol S (\mf{B} \mf{B}^T)^{-1} \mf{B}\, \mr{diag}(\mf{B}^T \E[\bs{w}_i]) \mf{B}^{T} (\mf{B} \mf{B}^T)^{-1} \boldsymbol S^T - \E\left[\bs{\theta}_i^{\phantom{T}}\!\!\!\;\bs{\theta}_i^T\right] \right), \quad 
 \mathbf G = \mathbf 0.
\]
\end{theorem}

\subsubsection{Identification}

Recall that we are sampling $(\mf x_i,C_i)$ according to the topic model (\ref{eq:obs.1}) in either a fixed-DGP setting (Appendix~\ref{sec:appendix_examples.topic.fixed}) or from a sequence of DGPs in which the distribution of $(\mf x_i, \bs w_i)|C_i$ is fixed but the distribution of $C_i$ is changing with $n$ so that (\ref{eq:kappa.topic}) holds (Section~\ref{sec:topic.theory.main}). In either case, the sampling framework is one in which the number of observations $n$, and hence the number of $\bs w_i$, is increasing.

This setting differs from recent works in econometrics that have studied identification. \citet{keRobustMachineLearning2024} and \cite{freyaldenhovenTestabilityAnchorWords2023} consider a fixed-$n$ setting where each $C_i \to \infty$ so that $\frac{\mf x_i}{C_i} \to_p \mf p_i := \E[\frac{\mf x_i}{C_i}|C_i,\bs w_i]$, the vector of multinomial probabilities for observation $i$. In a fixed-$n$ setting, identification concerns uniqueness of the factorization  $\mathbf{P} = \mathbf{B}^T \boldsymbol{W}$ with $\mathbf{P} = [\mathbf{p}_1,\ldots,\mathbf{p}_n] \in (\Delta^{V-1})^n$, $\boldsymbol{W} = [\boldsymbol{w}_1,\ldots,\boldsymbol{w}_n] \in (\Delta^{K-1})^n$ and $\mathbf B \in (\Delta^{V-1})^K$, conditional on the $n$ (fixed) sampled units. 

Within a fixed-$n$ context, identification of $\mf B$ is commonly achieved in text applications by assuming the existence of anchor words that are known to appear in some topics but not others. In essence, these amount to zero restrictions on elements of $\mathbf B$. 

In the fixed-DGP case we consider in Appendix~\ref{sec:appendix_examples.topic.fixed}, anchor words are also sufficient for $\mf B$ to be consistently estimable as $n \to \infty$ and thus, by construction, identified. See, e.g., \cite{bingOptimalEstimationSparse2020}, \cite{wuSparseTopicModeling2023}, and \cite{keUsingSVDTopic2022}. The same is true for the sequence-of-DGPs case we consider in Section~\ref{sec:topic.theory.main}. But, in that case, we are in effect sampling each true multinomial probability $\mf p_i$, since each $C_i$ is growing with $n$. Let $\mc P \subseteq \Delta^{V-1}$ denote the support of $\mf p_i$.\footnote{In the fixed-$n$ case which conditions on the $n$ observed units, we have $\mc P = \{\mf p_1,\ldots,\mf p_n\}$.} Then $\mf B$ is identified if it is the unique element of $(\Delta^{V-1})^K$ for which $\{(\mf{B} \mf{B}^T)^{-1} \mf{B} \mf p : \mf p \in \mc P\} \subseteq \Delta^{K-1}$. Anchor words are sufficient for this, but not necessary when the $\bs w_i$ have rich support.

To see this, consider the following illustration within the context of the AI/ML-generated index running example (Application 3 in Section~\ref{sec:applications}). Recall that here we have 
\[
 \mathbf B^T = \left[ \begin{array}{cc}
 \beta_1 & \beta_0 \\ (1-\beta_1) & (1-\beta_0) \end{array} \right], \quad \quad \bs w_i = \left[ \begin{array}{c} \theta_i \\ 1-\theta_i \end{array} \right].
\]
Suppose $\theta_i$ has probability density function that is strictly positive on $(0,1)$. Then $\mc P$ is the convex hull of $[\beta_1, 1-\beta_1]^T$ and $[\beta_0,1-\beta_0]^T$. These extreme values of $\mc P$ identify $\beta_1$ and $\beta_0$, and therefore $\mf B$ is identified without anchor words.

\section{Supplemental Results and Proofs} \label{sec:appendix_supplemental}

\paragraph{Notation} Let $\|\,\cdot\|$ denote the Euclidean norm when applied to vectors and the spectral norm when applied to matrices. Let $\|\,\cdot\,\|_F$ denote the Frobenius norm.

\subsection{Fixed-DGP Asymptotics}

\begin{proof}[Proof of Theorem~\ref{theorem:two-step.general}]
First consider the denominator. We have
\[
\begin{aligned}
 \frac 1n \sum_{i=1}^n \hat{\bs{\xi}}_i^{\phantom{T}}\!\!\!\;\hat{\bs{\xi}}_i^T
 & = \frac 1n \sum_{i=1}^n {\bs{\xi}}_i^{\phantom{T}}\!\!\!\;{\bs{\xi}}_i^T 
 + \frac 1n \sum_{i=1}^n (\hat{\boldsymbol{\xi}}_i - \bs{\xi}_i)(\hat{\bs{\xi}}_i - \bs{\xi}_i)^T
 + \frac 1n \sum_{i=1}^n (\hat{\boldsymbol{\xi}}_i - \bs{\xi}_i)\bs \xi_i^T
 + \frac 1n \sum_{i=1}^n \bs \xi_i(\hat{\bs{\xi}}_i - \bs{\xi}_i)^T \\
 & = \frac 1n \sum_{i=1}^n {\bs{\xi}}_i^{\phantom{T}}\!\!\!\;{\bs{\xi}}_i^T
 + \left[ \begin{array}{cc}
 \frac 1n \sum_{i=1}^n (\hat{\boldsymbol{\theta}}_i - \bs{\theta}_i)(\hat{\bs{\theta}}_i - \bs{\theta}_i)^T & \mathbf 0 \\
 \mathbf 0 & \mathbf 0 \end{array} \right] \\
 & \quad \quad
 + \left[ \begin{array}{cc}  \frac 1n \sum_{i=1}^n (\hat{\boldsymbol{\theta}}_i - \bs{\theta}_i)\bs{\theta}_i^T & \frac 1n \sum_{i=1}^n (\hat{\boldsymbol{\theta}}_i - \bs{\theta}_i)\mf{q}_i^T
 \\  \mathbf 0 & \mathbf 0 \end{array} \right] \\
 & \quad \quad 
 + \left[ \begin{array}{cc}  \frac 1n \sum_{i=1}^n \bs{\theta}_i(\hat{\boldsymbol{\theta}}_i - \bs{\theta}_i)^T &  \mathbf 0 \\
 \frac 1n \sum_{i=1}^n \mf{q}_i(\hat{\boldsymbol{\theta}}_i - \bs{\theta}_i)^T &  \mathbf 0 \end{array} \right].
\end{aligned}
\]
Hence,
\[
 \frac 1n \sum_{i=1}^n \hat{\bs{\xi}}_i^{\phantom{T}}\!\!\!\;\hat{\bs{\xi}}_i^T 
 \to_p \E\Big[{\bs{\xi}}_i^{\phantom{T}}\!\!\!\;{\bs{\xi}}_i^T\Big] 
 + \left[ \begin{array}{cc}
 \mathbf H + \mathbf W + \mathbf W^T & \mathbf 0 \\
 \mathbf 0 & \mathbf 0 
 \end{array} \right]
\]
by Assumption~\ref{assumption:fixed.general}\ref{assumption:general.lln}. The right-hand side is finite by Assumption~\ref{assumption:fixed.general}\ref{assumption:general.moments} and invertible by assumption.

For the numerator term, we have
\[
 \begin{aligned}
 \frac 1n \sum_{i=1}^n \hat{\bs{\xi}}_i Y_i
 & = \frac 1n \sum_{i=1}^n \bs{\xi}_i^{\phantom{T}}\!\!\!\;{\bs{\xi}}_i^T \bs \psi + \frac 1n \sum_{i=1}^n (\hat{\bs{\xi}}_i - \bs \xi_i) {\bs{\xi}}_i^T \bs \psi + \frac 1n \sum_{i=1}^n \hat{\bs{\xi}}_i \varepsilon_i \\
 & = \frac 1n \sum_{i=1}^n \bs{\xi}_i^{\phantom{T}}\!\!\!\;{\bs{\xi}}_i^T \bs \psi 
 + \left[ \begin{array}{cc}  \frac 1n \sum_{i=1}^n (\hat{\boldsymbol{\theta}}_i - \bs{\theta}_i)\bs{\theta}_i^T & \frac 1n \sum_{i=1}^n (\hat{\boldsymbol{\theta}}_i - \bs{\theta}_i)\mf{q}_i^T
 \\  \mathbf 0 & \mathbf 0 \end{array} \right] \bs \psi
 + \frac 1n \sum_{i=1}^n \hat{\bs{\xi}}_i \varepsilon_i.
 \end{aligned}
\]
Hence,
\[
 \frac 1n \sum_{i=1}^n \hat{\bs{\xi}}_i Y_i
 \to_p \left( 
 \E\Big[{\bs{\xi}}_i^{\phantom{T}}\!\!\!\;{\bs{\xi}}_i^T\Big]
 + \left[ \begin{array}{cc}
 \mathbf G & \mathbf 0 \\
 \mathbf 0 & \mathbf 0 
 \end{array} \right] \right) \boldsymbol \psi
\]
by Assumption~\ref{assumption:fixed.general}\ref{assumption:general.lln}. The first result follows by Slutsky's theorem. The second result then follows because $(\mathbf{A} + \mathbf{Q})^{-1} = \mathbf{A}^{-1} - \mathbf{A}^{-1} \mathbf{Q} \mathbf{A}^{-1} + O(\|\mathbf{Q}\|^2)$ for $\mathbf{A}$ invertible and $\mathbf{Q}$ small. 
\end{proof}

\subsection{AI/ML-Generated Labels}

\begin{proof}[Proof of Theorem~\ref{theorem:two-step.ml}]
Assumption~\ref{assumption:fixed.general}\ref{assumption:general.moments} holds by Assumption~\ref{assumption:fixed.ml}\ref{assumption:dgp.ml} and the fact that $\|\bs{\theta}_i\| \leq 1$.
The first part of Assumption~\ref{assumption:fixed.general}\ref{assumption:general.lln} holds by the law of large numbers and the fact that $\E[\| \bs \xi_i\|^2 ] < \infty$. For the second part, by Assumption~\ref{assumption:fixed.ml}\ref{assumption:theta.hat.ml} and the law of large numbers,
\[
 \begin{aligned}
 \frac 1n \sum_{i=1}^n (\hat{\boldsymbol{\theta}}_i - \bs{\theta}_i)(\hat{\bs{\theta}}_i - \bs{\theta}_i)^T
 & = \frac 1n \sum_{i=1}^n (\mf r(\mf x_i, U_i) - \bs{\theta}_i)(\mf r(\mf x_i, U_i) - \bs{\theta}_i)^T + o_p(1) \\
 & \to_p \E \left[ (\mf r(\mf x_i, U_i) - \bs{\theta}_i)(\mf r(\mf x_i, U_i) - \bs{\theta}_i)^T \right] \\
 & = \E \left[ \mr{diag}(\bs \pi(\mf x_i) + \mf p(\mf x_i)) - \bs \pi(\mf x_i) \mf p(\mf x_i)^T - \mf p(\mf x_i) \bs \pi(\mf x_i)^T \right] =:\mathbf H,
 \end{aligned}
\]
where the final line is by independence of $\bs \theta_i$ and $\mf r(\mf x_i, U_i)$ conditional on $\mf x_i$. Similarly, for the third part, we have 
\begin{multline*}
 \frac 1n \sum_{i=1}^n (\hat{\boldsymbol{\theta}}_i - \bs{\theta}_i)\bs{\theta}_i^T
 = \frac 1n \sum_{i=1}^n (\mf r(\mf x_i, U_i) - \bs{\theta}_i)\bs{\theta}_i^T + o_p(1) \\
 \to_p \E \left[ (\mf r(\mf x_i, U_i) - \bs{\theta}_i)\bs{\theta}_i^T \right] 
 = \E \left[ \bs \pi(\mf x_i) \mf p(\mf x_i)^T - \mr{diag}(\mf p(\mf x_i)) \right] =:\mathbf W.
\end{multline*}
For the fourth part, first note that
\[
 \left\|\frac 1n \sum_{i=1}^n (\hat{\bs{\theta}}_i - \mf r(\mf x_i, U_i)) \mf q_i^T \right\|
 \leq 
 \max_{1 \leq i \leq n} \left\| \hat{\boldsymbol{\theta}}_i - \mf r(\mf x_i, U_i) \right\| \times \frac 1n \sum_{i=1}^n \left\| \mf q_i \right\| \to_p 0,
\]
by Assumption~\ref{assumption:fixed.ml}\ref{assumption:theta.hat.ml}-\ref{assumption:dgp.ml}. Then by Assumption~\ref{assumption:fixed.ml}\ref{assumption:dgp.ml}-\ref{assumption:q.unbiased} and the law of large numbers,
\[
 \frac 1n \sum_{i=1}^n (\mf r(\mf x_i, U_i) - \bs \theta_i) \mf q_i^T
 \to_p \E \left[ (\mf r(\mf x_i, U_i) - \bs \theta_i) \mf q_i^T \right] = \mathbf 0.
\]

For the final part of Assumption~\ref{assumption:fixed.general}\ref{assumption:general.lln}, first note
\begin{equation} \label{eq:proposition:two-step.2}
 \frac 1n \sum_{i=1}^n \hat{\bs{\xi}}_i \varepsilon_i = 
 \left[ \begin{array}{c} \frac 1n \sum_{i=1}^n \hat{\bs{\theta}}_i \varepsilon_i \\ 
 \frac 1n \sum_{i=1}^n \mathbf q_i \varepsilon_i 
 \end{array} \right],
\end{equation}
where $\frac 1n \sum_{i=1}^n \mathbf q_i \varepsilon_i  \to_p \mf 0$ by the law of large numbers and Assumption~\ref{assumption:fixed.ml}\ref{assumption:dgp.ml}. Moreover, 
\[
 \left\| \frac 1n \sum_{i=1}^n \left(\hat{\bs{\theta}}_i - \mathbf r(\mathbf x_i, U_i) \right) \varepsilon_i  \right\| 
 \leq 
 \max_{1 \leq i \leq n} \left\| \hat{\boldsymbol{\theta}}_i - \mf r(\mf x_i, U_i) \right\| \times \frac 1n \sum_{i=1}^n |\epsilon_i| \to_p 0,
\]
by  Assumption~\ref{assumption:fixed.ml}\ref{assumption:theta.hat.ml}-\ref{assumption:dgp.ml}. Finally, 
\[
 \frac 1n \sum_{i=1}^n \mathbf r(\mathbf x_i, U_i) \varepsilon_i 
 \to_p \E \left[ \mathbf r(\mathbf x_i, U_i) \varepsilon_i \right] = \E \left[ \bs \pi(\mf x_i) \varepsilon_i \right] = \mathbf 0 
\]
by the law of large numbers, independence of $\mathbf r(\mathbf x_i, U_i)$ and $\varepsilon_i$ conditional on $\mf x_i$ (for the first equality), and Assumption~\ref{assumption:fixed.ml}\ref{assumption:fixed.ml.eps} (for the second equality). 
\end{proof}

\begin{lemma}\label{lem:drifting.ml}
Let $\mf z_i$ be a random vector with finite fourth moment and let Assumption~\ref{assumption:drifting.ml}\ref{assumption:theta.hat.ml.2} and (\ref{eq:kappa.ml.0}) hold. Then
\[
 \frac{1}{\sqrt n} \sum_{i=1}^n \left( \mf z_i (\hat{\bs \theta}_i - \bs \theta_i)^T - \E \left[ \mf z_i (\mf r(\mf x_i, U_i) - \bs \theta_i)^T \right] \right) \to_p 0.
\]
\end{lemma}

\begin{proof}[Proof of Lemma~\ref{lem:drifting.ml}]
First note that 
\begin{multline*}
 \left\| \frac{1}{\sqrt n} \sum_{i=1}^n \mf z_i (\hat{\bs \theta}_i - \bs \theta_i)^T
 - 
 \frac{1}{\sqrt n} \sum_{i=1}^n \mf z_i (\mf r(\mf x_i, U_i) - \bs \theta_i)^T \right\| \\
 \leq \sqrt n \max_{1 \leq i \leq n} \| \hat{\boldsymbol{\theta}}_i - \mf r(\mf x_i, U_i) \| \times  \frac 1n \sum_{i=1}^n \|\mf z_i\| \to_p 0,
\end{multline*}
by Assumption~\ref{assumption:drifting.ml}\ref{assumption:theta.hat.ml.2} and the fact that $\E[\|\mf z_i\| ] < \infty$. Now let $\mf X_{i,n} = \mf z_i (\mf r(\mf x_i, U_i) - \bs \theta_i)^T - \E \left[ \mf z_i (\mf r(\mf x_i, U_i) - \bs \theta_i)^T \right]$. With $D$ denoting the dimension of $\mf z_i$, we have 
\[
 \begin{aligned}
 \E \left[ \left\| \frac{1}{\sqrt n} \sum_{i=1}^n \mathbf{X}_{i,n} \right\|^2_F \right]
 & = \sum_{j=1}^D \sum_{k=1}^K \E \left[ \left( \mathbf{X}_{i,n} \right)_{j,k}^2 \right] \\
 & \leq \sum_{j=1}^D \sum_{k=1}^K \E \left[ \left( \mf{z}_{i,j} \right)^2 \left( r_k(\mf x_i, U_i) - \theta_{i,k} \right)^2 \right] \\
 & \leq \sum_{j=1}^D \sum_{k=1}^K \E \left[ \left( \mf{z}_{i,j} \right)^4 \right]^{1/2} \E \left[ \left( r_k(\mf x_i, U_i) - \theta_{i,k} \right)^4 \right]^{1/2} \\
 & \leq \mathrm{constant} \times \sum_{k=1}^K \E \left[ \left( r_k(\mf x_i, U_i) - \theta_{i,k} \right)^2 \right]^{1/2}  \to 0,
 \end{aligned}
\]
where the second inequality is by Cauchy-Schwarz, the third is because $\E[\|\mf z_i\|^4] < \infty$ and $r_k(\mf x_i, U_i) - \theta_{i,k}$ takes values in $\{-1,0,1\}$, and convergence to zero is by (\ref{eq:kappa.ml.0}) because $\E \left[ \left( r_k(\mf x_i, U_i) - \theta_{i,k} \right)^2 \right] = (\mf H)_{k,k}$. The result now follows by Chebyshev's inequality.
\end{proof}

\begin{proof}[Proof of Theorem~\ref{theorem:drifting.ml}]
Assumption~\ref{assumption:drifting.general}\ref{assumption:drifting.moments} holds by Assumption~\ref{assumption:drifting.ml}\ref{assumption:dgp.ml.2} and the fact that $\|\bs \theta_i\| \leq 1$.

We now verify Assumption~\ref{assumption:drifting.general}\ref{assumption:drifting.lln}. First by Lemma~\ref{lem:drifting.ml} and Assumption~\ref{assumption:drifting.ml}, we have
\begin{equation} \label{eq:drifting.ml.1}
 \frac{1}{\sqrt n} \sum_{i=1}^n (\hat{\bs \theta}_i - \bs \theta_i) \mf q_i^T 
 = \frac{1}{\sqrt n} \sum_{i=1}^n \left( (\hat{\bs \theta}_i - \bs \theta_i) \mf q_i^T - \E[( \bs \pi(\mf x_i) - \mf p(\mf x_i) ) \mf q_i^T ] \right) \to_p \mf 0,
\end{equation}
which establishes the final part of Assumption~\ref{assumption:drifting.general}\ref{assumption:drifting.lln}. Similarly, by Assumption~\ref{assumption:drifting.ml}\ref{assumption:theta.hat.ml.2},
\[
 \left\| \frac{1}{\sqrt n} \sum_{i=1}^n \hat{\bs \theta}_i (\hat{\bs \theta}_i - \bs \theta_i)^T
 - 
 \frac{1}{\sqrt n} \sum_{i=1}^n \mf r(\mf x_i, U_i) (\hat{\bs \theta}_i - \bs \theta_i)^T \right\| \to_p 0.
\]
Hence, it follows by Lemma~\ref{lem:drifting.ml} and (\ref{eq:kappa.ml}) that 
\begin{equation} \label{eq:drifting.ml.7}
 \frac{1}{\sqrt n} \sum_{i=1}^n \hat{\bs \theta}_i (\hat{\bs \theta}_i - \bs \theta_i)^T \to_p \kappa \, \bs \Omega,
\end{equation}
which establishes the second part of Assumption~\ref{assumption:drifting.general}\ref{assumption:drifting.lln}. For the first part of Assumption~\ref{assumption:drifting.general}\ref{assumption:drifting.lln}, we note
\[
 \frac 1n \sum_{i=1}^n \hat{\bs{\xi}}_i^{\phantom{T}}\!\!\!\;\hat{\bs{\xi}}_i^T
 = 
 \frac 1n \sum_{i=1}^n \hat{\bs{\xi}}_i(\hat{\bs{\xi}}_i - \bs{\xi}_i)^T + \frac 1n \sum_{i=1}^n (\hat{\bs{\xi}}_i - \bs{\xi}_i)\bs{\xi}_i^T + \frac 1n \sum_{i=1}^n {\bs{\xi}}_i^{\phantom{T}}\!\!\!\;{\bs{\xi}}_i^T \,.
\]
Displays (\ref{eq:drifting.ml.1}) and (\ref{eq:drifting.ml.7}) together imply that the first term on the right-hand side is asymptotically negligible. Moreover, $\frac 1n \sum_{i=1}^n {\bs{\xi}}_i^{\phantom{T}}\!\!\!\;{\bs{\xi}}_i^T \to_p \E \left[ {\bs{\xi}}_i^{\phantom{T}}\!\!\!\;{\bs{\xi}}_i^T \right]$ by the law of large numbers and Assumption~\ref{assumption:drifting.ml}\ref{assumption:dgp.ml.2}. It therefore remains to show that the second term on the right-hand side of the above display is asymptotically negligible. In view of (\ref{eq:drifting.ml.1}) it is enough to show 
\[
 \frac{1}{n} \sum_{i=1}^n {\bs \theta}_i (\hat{\bs \theta}_i - \bs \theta_i)^T \to_p \mf 0 .
\]
By Lemma~\ref{lem:drifting.ml}, 
\[
 \left\| \frac{1}{n} \sum_{i=1}^n {\bs \theta}_i (\hat{\bs \theta}_i - \bs \theta_i)^T  - \E \left[ {\bs \theta}_i (\mf r(\mf x_i, U_i) - \bs \theta_i)^T  \right] \right\| \to_p 0,
\]
where $\E \left[ {\bs \theta}_i (\mf r(\mf x_i, U_i) - \bs \theta_i)^T  \right] = \E\left[ \mf p(\mf x_i) \bs \pi(\mf x_i)^T - \mr{diag}(\mf p(\mf x_i)) \right] \to \mf 0$ by (\ref{eq:kappa.ml.0}), noting the diagonal entries are $\E[p_k(\mf x_i)(\pi_k(\mf x_i)-1)]$ where $0 \leq \E[p_k(\mf x_i)(1-\pi_k(\mf x_i))] \leq (\mf H_{k,k}) \to 0$ and the off-diagonals are $\E[p_k(\mf x_i)\pi_j(\mf x_i)]$, where $\sum_{j \neq k} \E[p_k(\mf x_i)\pi_j(\mf x_i)] \leq \E[p_k(\mf x_i)(1-\pi_k(\mf x_i))] \to 0$. This completes the verification of Assumption~\ref{assumption:drifting.general}\ref{assumption:drifting.lln}.

Now consider Assumption~\ref{assumption:drifting.general}\ref{assumption:drifting.clt}. For the first part, we have
\[
 \frac{1}{\sqrt{n}} \sum_{i=1}^n (\hat{\bs{\xi}}_i - \bs{\xi}_i) \varepsilon_i = 
 \left[ \begin{array}{c}
 \frac{1}{\sqrt n} \sum_{i=1}^n (\hat{\bs{\theta}}_i - \bs{\theta}_i) \varepsilon_i \\
 \mf 0
 \end{array} \right].
\]
Note that $\sqrt n \, \E \left[ (\mf r(\mf x_i, U_i) - \bs{\theta}_i) \varepsilon_i \right] = \sqrt n \, \E \left[ \bs \pi(\mf x_i) \varepsilon_i \right] \to \mf 0$ by Assumption~\ref{assumption:drifting.ml}\ref{assumption:errors.ml.2} and the fact that $\E[\varepsilon_i \bs \theta_i] = \mf 0$. Hence, by Lemma~\ref{lem:drifting.ml} using Assumption~\ref{assumption:drifting.ml}\ref{assumption:theta.hat.ml.2}-\ref{assumption:dgp.ml.2}, we have
\[
 \frac{1}{\sqrt{n}} \sum_{i=1}^n (\hat{\bs{\theta}}_i - \bs{\theta}_i) \varepsilon_i \to_p \mf 0.
\]
It follows that $\frac{1}{\sqrt{n}} \sum_{i=1}^n \hat{\bs{\xi}}_i \varepsilon_i = \frac{1}{\sqrt{n}} \sum_{i=1}^n \bs{\xi}_i \varepsilon_i + o_p(1)$. The first part of Assumption~\ref{assumption:drifting.general}\ref{assumption:drifting.clt} now holds by the central limit theorem and Assumption~\ref{assumption:drifting.ml}\ref{assumption:dgp.ml.2}.

To complete the verification of Assumption~\ref{assumption:drifting.general}\ref{assumption:drifting.clt}, we start by writing
\begin{multline*}
 \quad \frac 1n \sum_{i=1}^n \hat \varepsilon_i^2 \hat{\bs{\xi}}_i^{\phantom{T}}\!\!\!\;\hat{\bs{\xi}}_i^T 
 = \frac 1n \sum_{i=1}^n \varepsilon_i^2 \bs{\xi}_i^{\phantom{T}}\!\!\!\;\bs{\xi}_i^T
 + \frac 1n \sum_{i=1}^n \varepsilon_i^2 \left( \hat{\bs{\xi}}_i^{\phantom{T}}\!\!\!\;\hat{\bs{\xi}}_i^T - {\bs{\xi}}_i^{\phantom{T}}\!\!\!\;{\bs{\xi}}_i^T \right) \\
 + \frac 1n \sum_{i=1}^n (\hat{\varepsilon}_i^2 - \varepsilon_i^2)\hat{\bs{\xi}}_i^{\phantom{T}}\!\!\!\;\hat{\bs{\xi}}_i^T 
 =: T_{1,n} + T_{2,n} + T_{3,n}, \quad
\end{multline*}
where $T_{1,n} \to_p \E \left[ \varepsilon_i^2 \bs{\xi}_i^{\phantom{T}}\!\!\!\;\bs{\xi}_i^T \right]$ by Assumption~\ref{assumption:drifting.ml}\ref{assumption:dgp.ml.2}. To show $T_{2,n} \to_p \mf 0$, it suffices to show
\begin{equation}\label{eq:drifting.ml.3}
 T_{2,a,n} := \frac 1n \sum_{i=1}^n \varepsilon_i^2 \left( \hat{\bs{\theta}}_i^{\phantom{T}}\!\!\!\;\hat{\bs{\theta}}_i^T - {\bs{\theta}}_i^{\phantom{T}}\!\!\!\;{\bs{\theta}}_i^T \right)  \to_p \mf 0,
\end{equation}
and
\begin{equation}\label{eq:drifting.ml.4}
 T_{2,b,n} := \frac 1n \sum_{i=1}^n \varepsilon_i^2 \mf q_i \left( \hat{\bs{\theta}}_i - {\bs{\theta}}_i \right)^T  \to_p \mf 0.
\end{equation}
For $T_{2,a,n}$, we may first deduce by Assumption~\ref{assumption:drifting.ml}\ref{assumption:theta.hat.ml.2}-\ref{assumption:dgp.ml.2} that
\[
 \left\| T_{2,a,n} - \frac 1n \sum_{i=1}^n \varepsilon_i^2 \left( \mf r(\mf x_i, U_i) \mf r(\mf x_i, U_i)^T - {\bs{\theta}}_i^{\phantom{T}}\!\!\!\;{\bs{\theta}}_i^T \right) \right\| \to_p 0.
\] 
Then since ${\bs{\theta}}_i^{\phantom{T}}\!\!\!\;{\bs{\theta}}_i^T = \mr{diag}(\bs \theta_i)$ and similarly for $\mf r(\mf x_i, U_i)$, we have by Cauchy--Schwarz that
\[
 \left\|  \frac 1n \sum_{i=1}^n \varepsilon_i^2 \left( \mf r(\mf x_i, U_i) \mf r(\mf x_i, U_i)^T - {\bs{\theta}}_i^{\phantom{T}}\!\!\!\;{\bs{\theta}}_i^T \right) \right\|
 \leq 
 \left( \frac 1n \sum_{i=1}^n \varepsilon_i^4 \right)^{1/2} \left( \frac 1n \sum_{i=1}^n \| \mf r(\mf x_i, U_i) - \bs \theta_i \|^2 \right)^{1/2},
\]
where the first term on the right-hand side is $O_p(1)$ by Assumption~\ref{assumption:drifting.ml}\ref{assumption:dgp.ml.2} and the second term is $o_p(1)$  because $\E \left[ \| \mf r(\mf x_i, U_i) - \bs \theta_i \|^2 \right] = \mr{tr}(\mf H) \to 0$ by (\ref{eq:kappa.ml.0}), proving (\ref{eq:drifting.ml.3}).
For $T_{2,b,n}$, we may similarly deduce by Assumption~\ref{assumption:drifting.ml}\ref{assumption:theta.hat.ml.2}-\ref{assumption:dgp.ml.2} that
\[
 \left\| T_{2,b,n} - \frac 1n \sum_{i=1}^n \varepsilon_i^2 \mf q_i \left( \mf r(\mf x_i, U_i) - {\bs{\theta}}_i \right)^T \right\| \to_p 0.
\] 
Then by H\"older's inequality, we have 
\begin{multline*}
 \left\| \frac 1n \sum_{i=1}^n \varepsilon_i^2 \mf q_i \left( \mf r(\mf x_i, U_i) - {\bs{\theta}}_i \right)^T \right\| \\
 \leq  \left( \frac 1n \sum_{i=1}^n \varepsilon_i^4 \right)^{1/2} \left( \frac 1n \sum_{i=1}^n \|\mf q_i\|^4 \right)^{1/4} \left( \frac 1n \sum_{i=1}^n \|\mf r(\mf x_i, U_i) - {\bs{\theta}}_i\|^4 \right)^{1/4} .
\end{multline*}
The first two terms on the right-hand side are $O_p(1)$ by Assumption~\ref{assumption:drifting.ml}\ref{assumption:dgp.ml.2}. For the third term, note that $\|\mf r(\mf x_i, U_i) - {\bs{\theta}}_i\|^4 \leq 4\|\mf r(\mf x_i, U_i) - {\bs{\theta}}_i\|^2$ because $\|\mf r(\mf x_i, U_i) - {\bs{\theta}}_i\|^2 \in \{0,1,2\}$. Hence, the third term is $o_p(1)$  by (\ref{eq:kappa.ml.0}), proving (\ref{eq:drifting.ml.4}).

Finally, note $\hat \varepsilon_i^2 - \varepsilon_i^2 = ({\bs\xi}_i^T {\bs \psi} - \hat{\bs\xi}_i^T \hat{\bs \psi})^2 + 2 \varepsilon_i ({\bs\xi}_i^T {\bs \psi} - \hat{\bs\xi}_i^T \hat{\bs \psi})$. Hence, to show $T_{3,n} \to_p \mf 0$, it suffices to show
\begin{equation}\label{eq:drifting.ml.5}
 T_{3,a,n} := \frac 1n \sum_{i=1}^n ({\bs\xi}_i^T {\bs \psi} - \hat{\bs\xi}_i^T \hat{\bs \psi})^2  \hat{\bs{\xi}}_i^{\phantom{T}}\!\!\!\;\hat{\bs{\xi}}_i^T   \to_p \mf 0,
\end{equation}
and
\begin{equation}\label{eq:drifting.ml.6}
 T_{3,b,n} := \frac 1n \sum_{i=1}^n \varepsilon_i ({\bs\xi}_i^T {\bs \psi} - \hat{\bs\xi}_i^T \hat{\bs \psi})  \hat{\bs{\xi}}_i^{\phantom{T}}\!\!\!\;\hat{\bs{\xi}}_i^T \to_p \mf 0.
\end{equation}
By H\"older's inequality, we have 
\[
 \|T_{3,a,n}\| \leq \left( \frac 1n \sum_{i=1}^n ({\bs\xi}_i^T {\bs \psi} - \hat{\bs\xi}_i^T \hat{\bs \psi})^4 \right)^{1/2} \left( \frac 1n \sum_{i=1}^n \|\hat{\bs{\xi}}_i\|^4 \right)^{1/2},
\]
and
\[
 \|T_{3,b,n}\| \leq \left( \frac 1n \sum_{i=1}^n ({\bs\xi}_i^T {\bs \psi} - \hat{\bs\xi}_i^T \hat{\bs \psi})^4 \right)^{1/4} \left( \frac 1n \sum_{i=1}^n \|\hat{\bs{\xi}}_i\|^4 \right)^{1/2} \left( \frac 1n \sum_{i=1}^n \varepsilon_i^4 \right)^{1/4},
\]
where $\frac 1n \sum_{i=1}^n \|\hat{\bs{\xi}}_i\|^4 = O_p(1)$  and $\frac 1n \sum_{i=1}^n \varepsilon_i^4 = O_p(1)$ by Assumption~\ref{assumption:drifting.ml}\ref{assumption:theta.hat.ml.2}-\ref{assumption:dgp.ml.2}. Moreover, ${\bs\xi}_i^T {\bs \psi} - \hat{\bs\xi}_i^T \hat{\bs \psi} = (\bs{\theta}_i - \hat{\bs \theta}_i)^T \bs{\gamma} + \hat{\bs \xi}_i^T(\bs \psi - \hat{\bs \psi})$. Then, since $(a + b)^4 \leq 8 a^4 + 8 b^4$, we have
\[
 \frac 1n \sum_{i=1}^n ({\bs\xi}_i^T {\bs \psi} - \hat{\bs\xi}_i^T \hat{\bs \psi})^4
 \leq
 \left( \frac 8n \sum_{i=1}^n \| \bs{\theta}_i - \hat{\bs \theta}_i \|^4 \right) \left\| \bs{\gamma} \right\|^4 + \left( \frac 8n \sum_{i=1}^n \|\hat{\bs{\xi}}_i\|^4 \right) \|\hat{\bs \psi} - \bs \psi\|^4,
\]
where the second term on the right-hand side is $o_p(1)$ by consistency of $\hat{\bs \psi}$ and the fact that $\frac 1n \sum_{i=1}^n \|\hat{\bs{\xi}}_i\|^4 = O_p(1)$, as established above. For the first term on the right-hand side, we have $\frac 1n \sum_{i=1}^n \| \bs{\theta}_i - \hat{\bs \theta}_i \|^4 = \frac 1n \sum_{i=1}^n \| \bs{\theta}_i - r(\mf x_i, U_i) \|^4 + o_p(1)$  by Assumption~\ref{assumption:drifting.ml}\ref{assumption:theta.hat.ml.2}. Then arguing as above we have $\frac 1n \sum_{i=1}^n \| \bs{\theta}_i - r(\mf x_i, U_i) \|^4  \to_p 0$, proving (\ref{eq:drifting.ml.5}) and (\ref{eq:drifting.ml.6}).
\end{proof}

\subsection{Topic Models}

Throughout this section and the next, let $\hat{\mf p}_i = \frac{\mf x_i}{C_i}$ and $\mf p_i = \E[\frac{\mf x_i}{C_i}|C_i, \bs w_i]$. The next two lemmas apply in both fixed-DGP and sequence-of-DGPs frameworks. 
\begin{lemma}\label{lem:expected.p}
Suppose that (\ref{eq:obs.1}) holds. Then
\[
 \E\left[ \left. \hat{\mf{p}}_i^{\phantom{T}}\!\!\!\;\hat{\mf{p}}_i^T \right|  C_i, \bs{w}_i \right] = \mf{B}^{T} \bs{w}_i^{\phantom{T}}\!\!\!\;\bs{w}_i^T \mf B + \frac{1}{C_i} \left( \mr{diag}(\mf{B}^T \bs{w}_i ) - \mf{B}^{T} \bs{w}_i^{\phantom{T}}\!\!\!\;\bs{w}_i^T \mf B \right) ,
\]
and
\[
 \E\left[ \left. (\hat{\mf{p}}_i - \mf{p}_i) (\hat{\mf{p}}_i - \mf{p}_i)^T \right|  C_i, \bs{w}_i \right] = \frac{1}{C_i}  \left( \mr{diag}(\mf{B}^T \bs{w}_i ) - \mf{B}^{T} \bs{w}_i^{\phantom{T}}\!\!\!\;\bs{w}_i^T \mf B \right) .
\]
\end{lemma}

\begin{proof}[Proof of Lemma~\ref{lem:expected.p}]
First note by (\ref{eq:obs.1}) that 
\[
\begin{aligned}
 \E\left[ \left. \hat{\mf{p}}_i^{\phantom{T}}\!\!\!\;\hat{\mf{p}}_i^T \right|  C_i, \bs{w}_i  \right]
 & = \frac{1}{C_i^2} \E\left[ \left. \mf{x}_i^{\phantom{T}}\!\!\!\;\mf{x}_i^T \right| C_i, \bs{w}_i \right]  \\
 & = \frac{1}{C_i^2} \left( \E\left[ \left. \mf{x}_i \right| C_i, \bs{w}_i \right] \E\left[ \left. \mf{x}_i \right| C_i, \bs{w}_i \right]^{T} + \mathrm{Var}\left[ \left. \mf{x}_i \right| C_i, \bs{w}_i \right] \right)  \\
 & = \mf{B}^{T} \bs{w}_i^{\phantom{T}}\!\!\!\;\bs{w}_i^T \mf B + \frac{1}{C_i} \left( \mr{diag}(\mf{B}^T \bs{w}_i ) - \mf{B}^{T} \bs{w}_i^{\phantom{T}}\!\!\!\;\bs{w}_i^T \mf B \right)   ,
\end{aligned}
\]
where the last line follows from the mean and variance of the multinomial distribution. The second result now follows because $\E\left[ \left. \hat{\mf p}_i \right| C_i, \bs{w}_i \right] = \mf{p}_i = \mathbf{B}^T \boldsymbol{w}_i$. 
\end{proof}

\begin{lemma}\label{lem:convergence.theta}
Let Assumption~\ref{assumption:fixed}\ref{assumption:B.rank}-\ref{assumption:theta.hat} hold. Then
\begin{multline*}
 \Bigg\| \frac 1n \sum_{i=1}^n \hat{\bs{\theta}}_i^{\phantom{T}}\!\!\!\;\hat{\bs{\theta}}_i^T 
 - \E\left[\bs{\theta}_i^{\phantom{T}}\!\!\!\;\bs{\theta}_i^T\right] \\
 - \E\left[\frac{1}{C_i} \right] \left( \bs{S} (\mf{B} \mf{B}^T)^{-1} \mf{B} \, \mr{diag}(\mf{B}^T \E[\bs{w}_i]) \mf{B}^T (\mf{B} \mf{B}^T)^{-1} \bs{S}^T - \E\left[\bs{\theta}_i^{\phantom{T}}\!\!\!\;\bs{\theta}_i^T\right] \right) \Bigg\| \to_p 0.
\end{multline*}
\end{lemma}

\begin{proof}[Proof of Lemma~\ref{lem:convergence.theta}]
In view of Assumption~\ref{assumption:fixed}\ref{assumption:theta.hat}, we have
\[
 \left\|\frac 1n \sum_{i=1}^n \hat{\bs{\theta}}_i^{\phantom{T}}\!\!\!\;\hat{\bs{\theta}}_i^T
 - \bs{S} (\hat{\mf{B}} \hat{\mf{B}}^T)^{-1} \hat{\mf{B}} \left( \frac 1n \sum_{i=1}^n \hat{\mf{p}}_i^{\phantom{T}}\!\!\!\;\hat{\mf{p}}_i^T \right) \hat{\mf{B}}^T(\hat{\mf{B}} \hat{\mf{B}}^T)^{-1} \bs{S}^T \right\| \to_p 0
\]
where $(\hat{\mf{B}} \hat{\mf{B}}^T)^{-1}$ exists with probability approaching one by Assumption~\ref{assumption:fixed}\ref{assumption:B.rank}-\ref{assumption:B.consistent}. Each element of $\hat{\mf{p}}_i^{\phantom{T}}\!\!\!\;\hat{\mf{p}}_i^T$ is bounded between $0$ and $1$, so we may deduce by Chebyshev's inequality that 
\[
 \left\| \frac 1n \sum_{i=1}^n \hat{\mf{p}}_i^{\phantom{T}}\!\!\!\;\hat{\mf{p}}_i^T 
 - \E \left[ \hat{\mf{p}}_i^{\phantom{T}}\!\!\!\;\hat{\mf{p}}_i^T \right] \right\| \to_p 0.
\]
Hence, by Assumption~\ref{assumption:fixed}\ref{assumption:B.consistent} and Slutsky's theorem, we have
\[
 \left\| \frac 1n \sum_{i=1}^n \hat{\bs{\theta}}_i^{\phantom{T}}\!\!\!\;\hat{\bs{\theta}}_i^T
 - \bs{S} (\mf{B} \mf{B}^T)^{-1} \mf{B} \, \E \left[ \hat{\mf{p}}_i^{\phantom{T}}\!\!\!\;\hat{\mf{p}}_i^T \right] \mf{B}^T (\mf{B} \mf{B}^T)^{-1} \bs{S}^T \right\| \to_p 0 .
\]
The result follows by Lemma~\ref{lem:expected.p} and independence of $C_i$ and $\bs{w}_i$.
\end{proof}

\begin{proof}[Proof of Theorem~\ref{theorem:two-step.topic}]
Assumption~\ref{assumption:fixed.general}\ref{assumption:general.moments} holds by Assumption~\ref{assumption:fixed}\ref{assumption:dgp.q} and the fact that $\bs{\theta}_i = \boldsymbol S \boldsymbol w_i$ where $\boldsymbol w_i$ takes values in $\Delta^{K-1}$, hence $\|\bs{\theta}_i\| \leq 1$.

The first part of Assumption~\ref{assumption:fixed.general}\ref{assumption:general.lln} holds because $\E[\| \bs \xi_i\|^2 ] < \infty$. For the second part, we have 
\[
 \frac 1n \sum_{i=1}^n \hat{\bs{\theta}}_i^{\phantom{T}}\!\!\!\;\hat{\bs{\theta}}_i^T 
 \to_p 
 \E\left[\bs{\theta}_i^{\phantom{T}}\!\!\!\;\bs{\theta}_i^T\right] + \mathbf H
\]
by Lemma~\ref{lem:convergence.theta}. Further, $\frac 1n \sum_{i=1}^n {\bs{\theta}}_i^{\phantom{T}}\!\!\!\;{\bs{\theta}}_i^T \to_p \E \left[ {\bs{\theta}}_i^{\phantom{T}}\!\!\!\;{\bs{\theta}}_i^T \right]$. To complete the proof of the second, third, and fourth parts of Assumption~\ref{assumption:fixed.general}\ref{assumption:general.lln}, it suffices to show that 
\begin{equation} \label{eq:proposition:two-step.1}
 \frac 1n \sum_{i=1}^n (\hat{\boldsymbol{\theta}}_i^{\phantom{T}} - {\bs{\theta}}_i^{\phantom{T}}\!\!\!\;)\bs{\xi}_i^T \to_p \mf 0.
\end{equation}
To this end, in view of Assumption~\ref{assumption:fixed}\ref{assumption:theta.hat}-\ref{assumption:dgp.q}, we have
\begin{multline*}
 \left\|\frac 1n \sum_{i=1}^n \hat{\bs{\theta}}_i^{\phantom{T}}\!\!\!\;\bs{\xi}_i^T
 - \bs{S} (\hat{\mf{B}} \hat{\mf{B}}^T)^{-1} \hat{\mf{B}} \left( \frac 1n \sum_{i=1}^n \hat{\mf{p}}_i^{\phantom{T}}\!\!\!\;\bs{\xi}_i^T \right) \right\| \\
 \leq \left( \max_{1 \leq i \leq n} \| \hat{\boldsymbol{\theta}}_i - \bs S (\hat{\mf{B}} \hat{\mf{B}}^T)^{-1} \hat{\mf{B}} \hat{\mf{p}}_i \| \right) \times \frac 1n \sum_{i=1}^n \| \bs{\xi}_i\| \to_p 0.
\end{multline*}
Note $(\hat{\mf{B}} \hat{\mf{B}}^T)^{-1} \hat{\mf{B}} \to_p (\mf{B} \mf{B}^T)^{-1} \mf{B}$ by Assumption~\ref{assumption:fixed}\ref{assumption:B.rank}-\ref{assumption:B.consistent}, and $\frac 1n \sum_{i=1}^n \hat{\mf{p}}_i^{\phantom{T}}\!\!\!\!\;\bs{\xi}_i^T = O_p(1)$ by Assumption~\ref{assumption:fixed}\ref{assumption:dgp.q} and the fact that $\|\hat{\mf{p}}_i^{\phantom{T}}\!\!\!\;\bs{\xi}_i^T \| \leq \|\bs{\xi}_i\|$. Hence,
\begin{multline*}
 \left\| \bs{S} (\hat{\mf{B}} \hat{\mf{B}}^T)^{-1} \hat{\mf{B}} \left( \frac 1n \sum_{i=1}^n \hat{\mf{p}}_i^{\phantom{T}}\!\!\!\;\bs{\xi}_i^T \right) 
 - \bs{S} ({\mf{B}} {\mf{B}}^T)^{-1} {\mf{B}} \left( \frac 1n \sum_{i=1}^n \hat{\mf{p}}_i^{\phantom{T}}\!\!\!\;\bs{\xi}_i^T \right)  \right\| \\
 \leq \| \bs{S} \| \left\| (\hat{\mf{B}} \hat{\mf{B}}^T)^{-1} \hat{\mf{B}} - (\mf{B} \mf{B}^T)^{-1} \mf{B} \right\| \times \left\| \frac 1n \sum_{i=1}^n \hat{\mf{p}}_i^{\phantom{T}}\!\!\!\;\bs{\xi}_i^T  \right\| \to_p 0.
\end{multline*}
Finally,
\[
 \left\| \bs{S} ({\mf{B}} {\mf{B}}^T)^{-1} {\mf{B}} \left( \frac 1n \sum_{i=1}^n \hat{\mf{p}}_i^{\phantom{T}}\!\!\!\;\bs{\xi}_i^T \right) - \frac 1n \sum_{i=1}^n \bs{\theta}_i^{\phantom{T}}\!\!\!\;\bs{\xi}_i^T \right\| 
 = \left\| \bs{S} ({\mf{B}} {\mf{B}}^T)^{-1} {\mf{B}} \left( \frac 1n \sum_{i=1}^n (\hat{\mf{p}}_i^{\phantom{T}} - \mf{p}_i^{\phantom{T}}\!\!)\bs{\xi}_i^T  \right) \right\|.
\]
But 
\[
  \frac 1n \sum_{i=1}^n (\hat{\mf{p}}_i^{\phantom{T}} - \mf{p}_i^{\phantom{T}}\!\!)\bs{\xi}_i^T  \to_p \E \left[ (\hat{\mf{p}}_i^{\phantom{T}} - \mf{p}_i^{\phantom{T}}\!\!)\bs{\xi}_i^T \right] = \mf 0 
\]
by independence of $\mf x_i$ and $\mf q_i$ conditional on $(C_i,\bs w_i)$, and the fact that $\E[\hat{\mf p}_i| C_i,\bs w_i] = \mf p_i$. This proves (\ref{eq:proposition:two-step.1}), from which we also conclude that $\mathbf G = \mathbf 0$.

To verify the final part of Assumption~\ref{assumption:fixed.general}\ref{assumption:general.lln}, first note
\begin{equation} \label{eq:proposition:two-step.2}
 \frac 1n \sum_{i=1}^n \hat{\bs{\xi}}_i \varepsilon_i = 
 \left[ \begin{array}{c} \frac 1n \sum_{i=1}^n \hat{\bs{\theta}}_i \varepsilon_i \\ 
 \frac 1n \sum_{i=1}^n \mathbf q_i \varepsilon_i 
 \end{array} \right],
\end{equation}
where $\frac 1n \sum_{i=1}^n \mathbf q_i \varepsilon_i  \to_p \mf 0$ by Assumption~\ref{assumption:fixed}\ref{assumption:dgp.q}. Moreover, by Assumption~\ref{assumption:fixed}\ref{assumption:theta.hat}-\ref{assumption:dgp.q}, we have
\[
 \left\| \frac 1n \sum_{i=1}^n \hat{\bs{\theta}}_i \varepsilon_i 
 - \bs{S} (\hat{\mf{B}} \hat{\mf{B}}^T)^{-1} \hat{\mf{B}} \left( \frac 1n \sum_{i=1}^n \hat{\mf{p}}_i \varepsilon_i \right) \right\| \to_p 0.
\]
We also have $(\hat{\mf{B}} \hat{\mf{B}}^T)^{-1} \hat{\mf{B}} \to_p (\mf{B} \mf{B}^T)^{-1} \mf{B}$ by Assumption~\ref{assumption:fixed}\ref{assumption:B.rank}-\ref{assumption:B.consistent}. Moreover, Assumption~\ref{assumption:fixed}\ref{assumption:dgp.q} and the fact that $\hat{\mathbf p}_i$ takes values in the simplex imply that $\E[ \|\hat{\mathbf p}_i \varepsilon_i\|] < \infty$. Hence,
\[
 \frac 1n \sum_{i=1}^n \hat{\mf{p}}_i \varepsilon_i 
 \to_p \E \left[  \hat{\mf{p}}_i \varepsilon_i \right] = \mathbf 0 
\]
by independence of $(\mf{x}_i,C_i)$ and $\varepsilon_i$ conditional on $(\bs w_i, \mf q_i)$, and  $\E[ \varepsilon_i ( \bs w_i, \mf q_i)] = \mf 0$. 
\end{proof}

\subsection{Additional Results for the Proof of Theorem~\ref{theorem:two-step.drifting}}

\begin{lemma}\label{lem:theta.bias}
Let Assumption~\ref{assumption:drifting}\ref{assumption:kappa.topic}-\ref{assumption:theta.rate} hold. Then
\[
 \frac{1}{\sqrt{n}} \sum_{i=1}^n \hat{\bs{\theta}}_i^{\phantom{T}}\!\!\!\;(\bs{\theta}_i - \hat{\bs{\theta}}_i)^T \to_p - \kappa \left( \bs S (\mf{B} \mf{B}^T)^{-1} \mf{B} \, \mr{diag}(\mf{B}^T \E[\bs{w}_i]) \mf{B}^T (\mf{B} \mf{B}^T)^{-1} \bs S^T - \E[ \bs{\theta}_i^{\phantom{T}}\!\!\!\;\bs{\theta}_i^T ] \right)  .
\]
\end{lemma}

\begin{proof}[Proof of Lemma~\ref{lem:theta.bias}]
First note that $\left\| \frac{1}{\sqrt{n}} \sum_{i=1}^n \hat{\bs{\theta}}_i^{\phantom{T}}\!\!\!\;(\bs{\theta}_i - \hat{\bs{\theta}}_i)^T  - T_{1,n} - T_{2,n} \right\| \to_p 0$ by Assumption~\ref{assumption:drifting}\ref{assumption:theta.rate}, where
\[
\begin{aligned}
 T_{1,n} & = \bs S (\hat{\mf{B}} \hat{\mf{B}}^T)^{-1} \hat{\mf{B}} \left( \left( \frac 1n \sum_{i=1}^n \hat{\mf{p}}_i^{\phantom{T}}\!\!\!\; \mf{p}_i^T \right) \sqrt{n} \left( \mf{B}^T (\mf{B} \mf{B}^T)^{-1} - \hat{\mf{B}}^T(\hat{\mf{B}} \hat{\mf{B}}^T)^{-1} \right) \right) \bs S^T  \\
 T_{2,n} & = \bs S (\hat{\mf{B}} \hat{\mf{B}}^T)^{-1} \hat{\mf{B}} \left( \frac{1}{\sqrt{n}} \sum_{i=1}^n \hat{\mf{p}}_i ( \mf{p}_i - \hat{\mf{p}}_i)^T \right) \hat{\mf{B}}^T(\hat{\mf{B}} \hat{\mf{B}}^T)^{-1} \bs S^T .
\end{aligned}
\]
Assumption~\ref{assumption:drifting}\ref{assumption:B.rank.2}-\ref{assumption:B.rate} implies that $\sqrt n(\mf{B}^T (\mf{B} \mf{B}^T)^{-1} - \hat{\mf{B}}^T(\hat{\mf{B}} \hat{\mf{B}}^T)^{-1}) \to_p \mathbf 0$. Moreover, as $\|\frac 1n \sum_{i=1}^n \hat{\mf{p}}_i^{\phantom{T}}\!\!\!\; \mf{p}_i^T\| \leq 1$, it follows that $T_{1,n} \to_p \mathbf 0$. 

For term $T_{2,n}$, note by Lemma~\ref{lem:expected.p} that 
\begin{align}
 \E \left[ \hat{\mf{p}}_i \left( \hat{\mf{p}}_i - \mf{p}_i \right)^T \right]
 & = \E \left[ \left( \hat{\mf{p}}_i - \mf{p}_i \right) \left( \hat{\mf{p}}_i - \mf{p}_i \right)^T \right] \notag \\
 & = \E \left[ \frac{1}{C_i} \right] \left( \mr{diag}\left(\mf{B}^T \E\left[\bs{w}_i\right]\right) - \mf{B}^T \E\left[ \bs{w}_i^{\phantom{T}}\!\!\!\;\bs{w}_i^T \right] \mf{B} \right) . \label{eq:two-step.drifting.1}
\end{align}
Let $\mf{X}_i = \hat{\mf{p}}_i \left( \hat{\mf{p}}_i - \mf{p}_i \right)^T  - \E \left[ \hat{\mf{p}}_i \left( \hat{\mf{p}}_i - \mf{p}_i \right)^T \right]$. Then with $\|\cdot\|_F$ denoting the Frobenius norm,
\[
 \begin{aligned}
 \E \left[ \left\| \frac{1}{\sqrt n} \sum_{i=1}^n \mathbf{X}_i \right\|^2_F \right]
 = \sum_{j=1}^V \sum_{k=1}^V \E \left[ \left( \mathbf{X}_i \right)_{j,k}^2 \right] 
 & \leq \sum_{j=1}^V \sum_{k=1}^V \E \left[ \left( \hat{\mf{p}}_{i,j} \right)^2 \left( \hat{\mf{p}}_{i,k} - \mf{p}_{i,k} \right)^2 \right] \\
 & \leq \sum_{k=1}^V \E \left[ \left( \hat{\mf{p}}_{i,k} - \mf{p}_{i,k} \right)^2 \right] \to 0,
 \end{aligned}
\]
where the second inequality is because $\hat{\mathbf p}_i$ is in the simplex and the convergence to zero holds in view of (\ref{eq:kappa.topic}) and (\ref{eq:two-step.drifting.1}). It follows that
\[
 \left\| \frac{1}{\sqrt{n}} \sum_{i=1}^n \hat{\mf{p}}_i ( \mf{p}_i - \hat{\mf{p}}_i)^T - \sqrt n \, \E \left[ \hat{\mf{p}}_i \left( \mf{p}_i - \hat{\mf{p}}_i \right)^T \right] \right\| \to_p 0.
\]
We conclude that $T_{2,n} \to_p - \kappa \left( \bs S (\mf{B} \mf{B}^T)^{-1} \mf{B} \, \mr{diag}(\mf{B}^T \E[\bs{w}_i]) \mf{B}^T (\mf{B} \mf{B}^T)^{-1} \bs S^T - \E[ \bs{\theta}_i^{\phantom{T}}\!\!\!\;\bs{\theta}_i^T ] \right)$ by (\ref{eq:two-step.drifting.1}) and Assumption~\ref{assumption:drifting}\ref{assumption:kappa.topic}-\ref{assumption:B.rate}
\end{proof}

\begin{lemma}\label{lem:theta.q.bias}
Let Assumption~\ref{assumption:drifting}\ref{assumption:kappa.topic}-\ref{assumption:dgp.q.2} hold. Then
\[
 \frac{1}{\sqrt{n}} \sum_{i=1}^n \bs{\xi}_i^{\phantom{T}}\!\!\!\;(\bs{\theta}_i - \hat{\bs{\theta}}_i)^T \to_p \mathbf 0 .
\]
\end{lemma}

\begin{proof}[Proof of Lemma~\ref{lem:theta.q.bias}]
First note that $\left\| \frac{1}{\sqrt{n}} \sum_{i=1}^n \bs{\xi}_i^{\phantom{T}}\!\!\!\;(\bs{\theta}_i - \hat{\bs{\theta}}_i)^T  - T_{1,n} - T_{2,n} \right\| \to_p 0$ by Assumption~\ref{assumption:drifting}\ref{assumption:theta.rate}-\ref{assumption:dgp.q.2},
where
\[
\begin{aligned}
 T_{1,n} & = \left( \left( \frac 1n \sum_{i=1}^n \bs{\xi}_i^{\phantom{T}}\!\!\!\; \mf{p}_i^T \right) \sqrt{n} \left( \mf{B}^T (\mf{B} \mf{B}^T)^{-1} - \hat{\mf{B}}^T(\hat{\mf{B}} \hat{\mf{B}}^T)^{-1} \right) \right) \bs S^T \\
 T_{2,n} & = \left( \frac{1}{\sqrt{n}} \sum_{i=1}^n \bs{\xi}_i ( \mf{p}_i - \hat{\mf{p}}_i)^T \right) \hat{\mf{B}}^T(\hat{\mf{B}} \hat{\mf{B}}^T)^{-1} \bs S^T .
\end{aligned}
\]
Assumption~\ref{assumption:drifting}\ref{assumption:B.rank.2}-\ref{assumption:B.rate} implies $\sqrt n(\mf{B}^T (\mf{B} \mf{B}^T)^{-1} - \hat{\mf{B}}^T(\hat{\mf{B}} \hat{\mf{B}}^T)^{-1}) \to_p \mathbf 0$. We also have that $\|\frac 1n \sum_{i=1}^n \bs{\xi}_i^{\phantom{T}}\!\!\!\; \mf{p}_i^T\| \leq \frac 1n \sum_{i=1}^n  \|\bs{\xi}_i\| = O_p(1)$,  by Assumption~\ref{assumption:drifting}\ref{assumption:dgp.q.2}. Hence, $T_{1,n} \to_p \mathbf 0$. For $T_{2,n}$, note that $\E [\bs{\xi}_i ( \hat{\mf{p}}_i - \mf{p}_i )^T ] = \mathbf 0$ by independence of $\mf{x}_i$ and $\mf{q}_i$ conditional on $(C_i,\bs{w}_i)$ and the fact that $\E[\hat{\mf p}_i| C_i,\bs w_i] = \mf p_i$. Let $\mf{X}_i = \bs{\xi}_i \left( \hat{\mf{p}}_i - \mf{p}_i \right)^T$ and let $D$ denote the dimension of $\bs{\xi}_i$. Then
\[
 \begin{aligned}
 \E \left[ \left\| \frac{1}{\sqrt n} \sum_{i=1}^n \mathbf{X}_i \right\|^2_F \right]
 = \sum_{j=1}^D \sum_{k=1}^V \E \left[ \left( \mathbf{X}_i \right)_{j,k}^2 \right]
 & = \sum_{j=1}^D \sum_{k=1}^V \E \left[ \left( \bs{\xi}_{i,j} \right)^2 \left( \hat{\mf{p}}_{i,k} - \mf{p}_{i,k} \right)^2 \right] \\
 & \leq \sum_{j=1}^D \sum_{k=1}^V \E \left[ \left( \bs{\xi}_{i,j} \right)^4 \right]^{1/2} \E \left[ \left( \hat{\mf{p}}_{i,k} - \mf{p}_{i,k} \right)^4 \right]^{1/2} \\
 & \leq \mathrm{constant} \times \sum_{k=1}^V \E \left[ \left( \hat{\mf{p}}_{i,k} - \mf{p}_{i,k} \right)^2 \right]^{1/2}  \to 0,
 \end{aligned}
\]
where the first inequality is by Cauchy-Schwarz, the second is by Assumption~\ref{assumption:drifting}\ref{assumption:dgp.q.2} and the fact that $|\hat{\mf{p}}_{i,k} - \mf{p}_{i,k}| \leq 1$, and convergence to zero is by (\ref{eq:two-step.drifting.1}) and Assumption~\ref{assumption:drifting}\ref{assumption:kappa.topic}. It follows that $\frac{1}{\sqrt n} \sum_{i=1}^n \mathbf{X}_i  \to_p \mathbf0$. We conclude by Assumption~\ref{assumption:drifting}\ref{assumption:B.rank.2}-\ref{assumption:B.rate} that $T_{2,n} \to_p \mathbf0$.
\end{proof}

\begin{lemma}\label{lem:p.uniform}
Let Assumption~\ref{assumption:drifting}\ref{assumption:C} hold. Then
\[
 \max_{1 \leq i \leq n} \|\hat{\mf{p}}_i - \mf{p}_i\| \to_p 0.
\]
\end{lemma}

\begin{proof}[Proof of Lemma~\ref{lem:p.uniform}]
Let $\|\,\cdot\,\|_1$ be the $\ell^1$ norm. As $C_i\hat{\mf{p}}_i|(C_i, \bs{w}_i)  \sim \mbox{Multinomial}(C_i, \mf{p}_i)$, for all $t > 0$ we have
\[
 \Pr\left( \left. \max_{1 \leq i \leq n} \|\hat{\mf{p}}_i - \mf{p}_i\|_1 > t \right| \{(C_i, \bs{w}_i)\}_{i=1}^n \right)
 \leq \sum_{i=1}^n (2^V - 2) \exp\left\{-\frac{C_i t^2}{2}\right\} 
\]
by the union bound and Lemma 1 of \cite{mardiaConcentrationInequalitiesEmpirical2019}. Then by Assumption~\ref{assumption:drifting}\ref{assumption:C}, 
\[
 \Pr\left( \max_{1 \leq i \leq n} \|\hat{\mf{p}}_i - \mf{p}_i\|_1 > t  \right) 
 \leq n (2^V - 2) \exp \left\{-\frac{c (\log n)^{1+\epsilon} t^2}{2} \right\} ,
\]
where $c,\epsilon > 0$. Hence, $\max_{1 \leq i \leq n} \|\hat{\mf{p}}_i - \mf{p}_i\|_1 \to_p 0$. The result now follows because the $\ell^1$ norm is weakly greater than the Euclidean norm.
\end{proof}

\end{document}